\documentclass{scrartcl}

\KOMAoptions{paper=letter} % fontsize=12pt

\usepackage[todo]{group}
\usepackage{thm-restate}
\usepackage[bf]{caption} 
\usetikzlibrary{arrows.meta}
\usetikzlibrary{shapes.multipart}

\makeatletter
\newcommand*{\textoverline}[1]{$\overline{\hbox{#1\vphantom{\"A}}}\m@th$}
\makeatother

\newcommand{\lex}{\mathbf{lex}}

\newcommand{\IR}{\textsf{WLOC}\xspace}

\newcommand{\WSMON}{\textsf{WSMON}\xspace}
\newcommand{\WMON}{\textsf{WMON}\xspace}
\newcommand{\CWCS}{\textsf{COS}\xspace}

\newcommand{\IUA}{\textsf{IUA}\xspace}

\newcolumntype{L}[1]{>{\raggedright\let\newline\\\arraybackslash\hspace{0pt}}m{#1}}
\newcolumntype{C}[1]{>{\centering\let\newline\\\arraybackslash\hspace{0pt}}m{#1}}
\newcolumntype{R}[1]{>{\raggedleft\let\newline\\\arraybackslash\hspace{0pt}}m{#1}}

\title{Characterizing the Top Cycle via Strategyproofness}
\author{Felix Brandt \qquad Patrick Lederer\\Technische Universit\"at M\"unchen}

\begin{document}
\maketitle

\begin{abstract}
Gibbard and Satterthwaite have shown that the only single-valued social choice functions (SCFs) that satisfy non-imposition (i.e., the function's range coincides with its codomain) and strategyproofness (i.e., voters are never better off by misrepresenting their preferences) are dictatorships. 
%The convenient---but rather restrictive---assumption of single-valuedness has been widely criticized. 
In this paper, we consider set-valued social choice correspondences (SCCs) that are strategyproof according to Fishburn's preference extension and, in particular, the \emph{top cycle}, an attractive SCC that returns the maximal elements of the transitive closure of the weak majority relation. Our main theorem shows that, under mild conditions, the top cycle is the \emph{only} non-imposing strategyproof SCC whose outcome only depends on the quantified pairwise comparisons between alternatives. This result effectively turns the Gibbard-Satterthwaite impossibility into a complete characterization of the top cycle by moving from SCFs to SCCs. We also leverage key ideas of the proof of this statement to obtain a more general characterization of strategyproof SCCs.
%, all of which only depend on the pairwise majority relation.
\end{abstract}
		
\section{Introduction}
	
One of the most influential results in microeconomic theory, the Gibbard-Satterthwaite theorem, states that dictatorships are the only single-valued social choice functions (SCFs) that are non-imposing (i.e., every alternative is returned for some preference profile) and strategyproof (i.e., voters are unable to obtain a better outcome by misrepresenting their preferences) when there are at least three alternatives.
The convenient but rather restrictive assumption of single-valuedness has been criticized by various scholars. 
%Various scholars have identified the restriction to single-valued SCFs as a major shortcoming of the theorem. 
For instance,
\citet{Gard76a} asserts that ``[resoluteness] is a rather restrictive and unnatural assumption.'' In a similar vein,
\citet{Kell77a} writes that ``the Gibbard-Satterthwaite theorem [\dots]
%on the impossibility of nondictatorial, strategy-proof social choice
uses an assumption of singlevaluedness which is unreasonable'' and \citet{Tayl05a} that ``if there is a weakness to the Gibbard-Satterthwaite theorem, it is the assumption that winners are unique.'' 
%This sentiment is echoed by various other authors \citep[see, e.g.,][]{Barb77b,DuSc00a,Nehr00a,BDS01a,ChZh02a}.
The problem with single-valuedness is that 
%\PL: this is indeed the problem, but I do not really like the formulation because it is quite vague; in other domains you do exactly the same and it is fine
%the SCFs has to return a single alternative \emph{based on the preferences only}.
it is in conflict with the basic fairness notions of anonymity and neutrality which require that all voters and all alternatives are treated equally. For example, if half of the voters favor $a$ and the other half $b$, there is no fair way of selecting a single winner because both alternatives are equally acceptable. In the context of social choice, these fairness conditions are imperative because elections should be unbiased. One way to deal with this problem is to identify a set of winning candidates with the understanding that one of these candidates will eventually be selected by some tie-breaking rule independent of the voters' preferences. Ties can, for example, be broken by lottery or by letting a chairperson or a committee pick the winner.\footnote{These tie-breaking rules are common in real-world elections. Tied elections on the state level within the US are sometimes decided by lottery. The US Vice President acts as the President of the Senate and frequently breaks ties in the Senate. If no candidate in a US presidential election obtains an absolute majority of the voters, then the House of Representatives elects the winner among the best three candidates.}

As a result, a large body of research investigates so-called social choice correspondences (SCCs), which return sets of alternatives. In particular, several papers have shown statements that mimic the negative consequences of the Gibbard-Satterthwaite theorem \citep[e.g.,][]{DuSc00a,BDS01a,ChZh02a,Beno02a,Sato14a}. 
%,Rodr07a,Sato08a
These results are based on relatively strong assumptions about the manipulators' preferences over sets (which in turn are based on the voters' beliefs about how ties are broken). For example, all of these results rely on the assumption that a voter who prefers $a$ to $b$ to $c$ will engage in a manipulation in which the outcome changes from set $\{a,c\}$ to set $\{b\}$. %\todo{FB: check whether this is really true.}
However, it is quite possible that no voter entertains such preferences over sets.
%, even when allowing for arbitrary preferences over alternatives. 
By contrast, the voters' preferences over sets we surmise in this paper are systematically deduced from their preferences over alternatives, which leads to a weaker notion of strategyproofness. In more detail, we consider a preference extension attributed to \citep{Fish72a}, according to which a manipulation is only successful if the manipulator can change the outcome from a set $Y$ to another set $X$ such that he prefers all alternatives in $X\setminus Y$ to all alternatives in $Y$ and all alternatives in $X$ to all alternatives in $Y\setminus X$. 
%Two natural justifications for this extension are the existence of a linear tie-breaking ordering or of a tie-breaking lottery using \emph{a priori} probabilities, unknown to the voters \citep[see, e.g.,][]{Gard79a,ErSa09a,ChZh02a,BSS19a}. 
Two natural justifications for this extension are the existence of a chairperson who breaks ties or of \emph{a priori} probabilities of the voters how ties are broken \citep[see, e.g.,][]{Gard79a,ChZh02a,ErSa09a,BSS19a}. 
%\todo{PL: Speaking about a linear tie-breaking order seems conflicting with our previous argument that fairness is imperative in voting; I strongly prefer the formulation based on a chair person as tie-breaker (which might be mathematically less accurate, but it is more natural because it is not as straightforward to turn this into a single-valued mechanism). Similar, instead of tie-breaking lotteries unknown to the voters, I'd talk about a priori beliefs of the voters on how ties are broken.}
The resulting notion of strategyproofness, often called Fishburn-strategyproofness, allows for positive results. For example, the rather indecisive \emph{omninomination rule}, which returns all alternatives that are top-ranked by at least one voter, is strategyproof according to this notion.

A particularly promising approach to construct attractive strategyproof SCCs is to focus on the pairwise comparisons between alternatives \citep[see, e.g.,][]{Gard76a,MaPa81a,Band83a,Cake03a,Bran11c}. 
% This idea, which can be traced back to the Marquis de Condorcet and thus the origins of social choice theory, is central to many influential concepts in social choice theory. One such concept is the class of pairwise (aka C2) SCCs, whose outcome only depends on the weighted majority comparisons. This class was introduced by \citet{Fish77a} and includes many important SCCs such as Borda's rule, Copeland's rule, the top cycle, the uncovered set, the essential set, the Simpson-Kramer rule, Kemeny's rule, ranked pairs, and Schulze's rule. A very prominent concept that arises from the pairwise comparisons between alternatives is that of a Condorcet winner, which is an alternative that is preferred to every other alternative by a majority of voters \citep{Cond85a}.
For instance, \citet{Bran11c} shows that several attractive SCCs that only rely on the pairwise majority relation, such as the uncovered set and the bipartisan set, satisfy a strategyproofness notion which is slightly weaker than Fishburn-strategyproofness. In this paper, we thus focus on the class of pairwise (aka C2) SCCs, whose outcome only depends on the weighted majority comparisons. This class was introduced by \citet{Fish77a} and includes many important SCCs such as Borda's rule, Copeland's rule, the top cycle, the essential set, the Simpson-Kramer rule, Kemeny's rule, ranked pairs, and Schulze's rule \citep[see Chapters 3 and 4 in][for an overview of these SCCs]{BCE+14a}. Indeed, it is well-known that almost all other SCCs (e.g., positional scoring rules or runoff rules) are manipulable according to Fishburn's preference extension \citep[see, e.g.,][pp.~44--51]{Tayl05a}.

%To be precise, Theorem 2.2.2 in Taylor proves for 7 rules that they fail single-winner strategyproofness.
%Notable exceptions are the aforementioned omninomination rule and the Pareto rule, which returns all Pareto-optimal alternatives.

% BrBr11a,BBH15a,BrGe15a
%This idea, which can be traced back to the Marquis de Condorcet and thus the origins of social choice theory, is central to many influential concepts in social choice theory. 
A prominent concept that arises from pairwise comparisons between alternatives is that of a Condorcet winner, an alternative that is preferred to every other alternative by some majority of voters \citep{Cond85a}. Condorcet winners need not exist, but many scholars agree that an SCC should uniquely return the Condorcet winner whenever one exists. The non-existence of Condorcet winners can be addressed by extending the notion of Condorcet winners to so-called dominant sets of alternatives. A set of alternatives $X$ is dominant if every element of $X$ is preferred to every element not in $X$ by a majority of voters. Dominant sets are guaranteed to exist since the set of all alternatives is trivially dominant, and they can be ordered by set inclusion.\footnote{\label{fn:dominantsets}Assume for contradiction that two dominant sets $X,Y$ are not contained in each other. Then, there exists $x\in X\setminus Y$ and $y\in Y\setminus X$. The definition of dominant sets requires that $x$ is majority-preferred to $y$ and that $y$ is majority-preferred to $x$, a contradiction.}
These observations have led to the definition of the \emph{top cycle}, an SCC that returns the unique smallest dominant set for any given preference profile. This set consists precisely of the maximal elements of the transitive closure of the weak majority relation. The top cycle has been reinvented several times and is known under various names such as Good set \citep{Good71a}, Smith set \citep{Smit73a}, weak closure maximality \citep{Sen77a}, and GETCHA \citep{Schw86a}.

%\footnote{\todo{PL: I think it might be nice to introduce pairwise SCCs as one of the two (or three) main ideas for social choice (the others are positional rules and run-off rules): we decide the winner by comparing alternatives with each other (which can even be tracked back to Condorcet). This idea is also roughy discussed in \citet{Saar95a} Since the other classes are known to not contain any attractive strategyproof SCC, we focus on pairwise SCFs (for which there are positive results, see \citet{BrBr11a}). Could argue that positional SCCs are very indecisive. TC is relatively decisive in the absence of a Condorcet winner.}\todo{FB: This kind of argument doesn't really work anymore when non-pairwise SCCs like $\tc(\po)$ are strategyproof.}} 

In this paper, we characterize the class of strategyproof pairwise SCCs under relatively mild and common technical assumptions, namely the conditions of non-imposition, homogeneity, and neutrality. Our first result shows that every strategyproof pairwise SCC that satisfies these conditions always returns a dominant set. An important variant of this characterization is obtained when replacing non-imposition and neutrality with set not-imposition (every set of alternatives is returned for some preference profile): the top cycle is the \emph{only} strategyproof pairwise SCC that satisfies set non-imposition and homogeneity. This result effectively turns the Gibbard-Satterthwaite impossibility theorem into a complete characterization of the top cycle by moving from SCFs to SCCs. 

On top of strategyproofness, the top cycle is very robust in terms of changes to the set of feasible alternatives and preferences of the voters: it is invariant under removing losing alternatives as well as modifications of preferences between losing alternatives, and has been characterized repeatedly by choice consistency conditions implied by the the weak axiom of revealed preference. Finally, it is one of the most straightforward Condorcet extensions and can be easily computed. 
The main disadvantage of the top cycle is its possible inclusion of Pareto-dominated alternatives. %\citep{Ferejohn and Grether?}.
We believe that this drawback is tolerable because empirical results suggest that the top cycle only rarely contains Pareto-dominated alternatives. This is due to the fact that Pareto dominances are increasingly unlikely for large numbers of voters and the persistent observation that an overwhelming number of real-world elections admit Condorcet winners \citep[see, e.g.,][]{RGMT06a,Lasl10a,GeLe11a,BrSe15a}. In these cases, the top cycle consists of a single Pareto-optimal alternative. Moreover, as we point out in \Cref{rem:tcpo}, the SCC that returns the top cycle of the subset of Pareto-optimal alternatives satisfies all our axioms except pairwiseness. In particular, this SCC satisfies strategyproofness with respect to Fishburn's extension.
%\todo{FB: This disadvantage can be rectified by taking $\tc(\po)$. This weakens our argument for restricting attention to pairwise SCCs, however.}

\section{Related Work}

\citet{Gard79a} initiated the study of strategyproofness with respect to Fishburn's preference extension. He attributed this extension to Fishburn because it is the weakest extension that satisfies a set of axioms proposed by \citet{Fish72a}. A small number of SCCs were shown to be Fishburn-strategyproof, sometimes by means of stronger strategyproofness notions: the Pareto rule---which returns all Pareto-optimal alternatives \citep{Feld79a}, the omninomination rule---which returns all top-ranked alternatives \citep{Gard76a}, the Condorcet rule---which returns the Condorcet winner whenever one exists and all alternatives otherwise \citep{Gard76a}, the SCC that returns the Condorcet winner whenever one exists and all Pareto-optimal alternatives otherwise \citep{BrBr11a}, and the top cycle \citep{Band83a,BrBr11a,SaZw10a}. All other commonly studied SCCs fail Fishburn-strategyproofness \citep[see, e.g.,][]{Tayl05a,BBH15a}. A universal example showing the Fishburn-manipulability of many SCCs is given in \Cref{Fig:manipulation} of \Cref{sec:prelims}.

More recently, the limitations of Fishburn-strategyproofness were explored. \citet{BrGe15a} studied majoritarian SCCs, i.e., SCCs whose outcome only depends on the pairwise majority relation, and showed that no majoritarian SCC satisfies Fishburn-strategyproofness and Pareto-optimality. The condition of majoritarianess can be replaced with the much weaker condition of anonymity when allowing for ties in the preferences \citep{BSS19a}. Both results were obtained with the help of computer-aided theorem proving techniques. \citet[][Remark 3]{BrGe15a} observed that, in the absence of majority ties, the top cycle could be the finest majoritarian Condorcet extension that satisfies Fishburn-strategyproofness when there are at least five alternatives. A computer verified this claim for five, six, and seven alternatives using 24 hours of runtime. The claim now follows immediately from our \Cref{thm:F-SP} (see also \Cref{rem:variants}), irrespective of majority ties.

\citet{ChZh02a} considered a much stronger notion of strategyproofness based on Fishburn's preference extension: they require that the outcome when voting honestly is comparable and preferred to every choice set obtainable by a manipulation according to Fishburn's extension. Their main result shows that only constant and dictatorial SCCs are strategyproof according to this definition. \citet{BDS01a} derive a similar conclusion for a weaker notion of strategyproofness based on Fishburn's extension (but still stronger than the one considered in this paper). In their model, voters submit preference relations over sets of alternatives that adhere to certain structural restrictions. When these restrictions are given by Fishburn's extension, they prove that only dictatorships satisfy strategyproofness and unanimity. 

Several choice-theoretic characterizations of the top cycle exist. When assuming that choices from two-element sets are made according to majority rule, the influential characterization by \citet{Bord76a} entails that the top cycle is the finest SCC satisfying $\beta^+$, an expansion consistency condition implied by the weak axiom of revealed preference.\footnote{This result was recently rediscovered by \citet{ENO19a}.} \citet{EhSp08a} have shown that, in the absence of majority ties, the refinement condition can be replaced with two contraction consistency conditions. \citet{Bran11b}, \citet{Houy11a}, and \citet{BBSS14a} provide further characterization using choice consistency conditions. We are not aware of a characterization of the top cycle using strategyproofness.

% \todo{FB: What are the drawbacks of COND$\cap$Pareto? It sounds quite attractive (Pareto-optimality, Condorcet-consistency, and F-SP) It seems this also satisfies alpha-hat aka SSP.}
% \todo{PL: I'd guess that this SCF is even less decisive than the top cycle (even though the top cycle violates Pareto-optimality) because it always chooses all top-ranked alternatives. This means for instance, that things like the Condorcet loser property are violated.}
%\todo{FB: The following non-pairwise SCFs satisfy Fishburn-strategyproofness (as well as SSP): $\cond(\po)=\po(\cond)$, $\cond/ \omni$, $\tc(\po)$. The last one, a Pareto-optimal refinement of the top cycle, is particularly attractive. It may be possible to  build finer SCFs based on similar ideas.}

\section{Preliminaries}
\label{sec:prelims}

Let $\mathbb{N}=\{1,2,\dots\}$ denote an infinite set of voters and $A$ a finite set of $m$ alternatives. 
Moreover, let $\mathcal{F}(\mathbb{N})$ denote the set of all finite and non-empty subsets of $\mathbb{N}$. 
Intuitively, $\mathbb{N}$ is the set of all possible voters, whereas an element $N\in\mathcal{F}(\mathbb{N})$ represents a concrete electorate. 
Given an electorate $N\in\mathcal{F}(\mathbb{N})$, each voter $i\in N$ has a \emph{preference relation}, represented by a strict total order \emph{$\succ_i$} on $A$. 
% connex/connected, transitive, and anti-symmetric. 
%We write $\succ_i$ for the strict part of $R_i$, i.e., $x \succ_i y$ if $x \succsim_i y$ and not $y \succsim_i x$. 
The set of all preference relations on $A$ is denoted by $\mathcal{R}(A)$. A \emph{preference profile $R$} is a vector of preference relations, i.e., $R\in \mathcal{R}(A)^{N}$ for some electorate $N\in \mathcal{F}(\mathbb{N})$. The set of all preference profiles on $A$ is denoted by $\mathcal{R}^*(A)=\bigcup_{N\in \mathcal{F}(\mathbb{N})} \mathcal{R}(A)^N$.

For a preference profile $R\in\mathcal{R}^*(A)$, let 
\[g_R(x,y)\;=\;|\{i\in N\colon x\succ_i y\}|-|\{i\in N\colon y\succ_i x\}| \tag{Majority margin}\] 
be the \emph{majority margin} of $x$ over $y$ in $R$. It describes how many more voters prefer $x$ to $y$ than $y$ to $x$. 
Whenever $g_R(x,y)\ge 0$ for some pair of alternatives, we say that $x$ weakly (majority) dominates $y$, denoted by $x\succsim_R y$. Note that the relation $\succsim_R$, which we call \emph{majority relation}, is complete. Its strict part will be denoted by $\succ_R$, i.e., $x \succ_R y$ iff $x \succsim_R y$ and not $y \succsim_R x$, and its indifference part by $\sim_R$, i.e., $x \sim_R y$ iff $x \succsim_R y$ and $y \succsim_R x$. Whenever the number of voters is odd, there can be no majority ties and $\succsim_R$ is anti-symmetric. We will extend both individual preference relations and the majority relation to sets of alternatives using the shorthand notation $X\succ Y$ whenever $x\succ y$ for all $x\in X$ and $y\in Y$.

The majority relation gives rise to a number of important concepts in social choice theory. A \emph{Condorcet winner} is an alternative $x$ such that ${x} \succ_R A\setminus \{x\}$.
%Condorcet winners are unique whenever they exist. 
In a similar vein, a \emph{Condorcet loser} is an alternative $x$ with $A\setminus \{x\} \succ_R {x}$. Neither Condorcet winners nor Condorcet losers need to exist, but whenever they do, each of them is unique. A natural extension of these ideas to sets of alternatives is formalized via the notion of dominant sets.
A non-empty set $X\subseteq A$ is \emph{dominant} if $X \succ_R A\setminus X$. Whenever a Condorcet winner exists, it forms a singleton dominant set. In contrast to Condorcet winners, dominant sets are guaranteed to exist since the set of all alternatives $A$ is trivially dominant. For every majority relation $\succsim_R$, the set of dominant sets is totally ordered by set inclusion, i.e., each majority relation induces a hierarchy of dominant sets that are strictly contained in each other.\footref{fn:dominantsets}

For an illustration of these concepts, consider the example given in \Cref{fig:example}, which shows a preference profile $R$ and its majority relation. The weights on the edges of the majority relation indicate the majority margins. The profile $R$ neither admits a Condorcet winner nor a Condorcet loser since each alternative has at least one incoming and outgoing edge. There are two dominant sets, $\{a,b,c\}$ and $\{a,b,c,d,e\}$. Note that the notions of Condorcet winners and dominant sets are independent of the exact weights of the edges, but only depend on their directions.

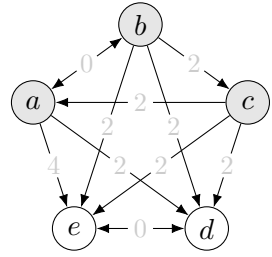
\begin{figure}[tb]
	\centering
		$
			\begin{array}{r@{\colon\;\;}l@{\;\;}l@{\;\;} l@{\;\;} l@{\;\;} l@{\;\;} l@{\;\;} l@{\;\;} l@{\;\;} l@{\;\;} l}
			1	& a	&	\succ_1	& e	&	\succ_1	&	b	&	\succ_1	&	c &	\succ_1	& d\\
			2	& b	&	\succ_2	& c	&	\succ_2	&	a	&	\succ_2	&	d &	\succ_2	& e\\
			3	& c	&	\succ_3	& a	&	\succ_3	&	b	&	\succ_3	&	e &	\succ_3	& d\\
			4	& d	&	\succ_4	& b	&	\succ_4	&	c	&	\succ_4	&	a &	\succ_4	& e
			\end{array}
		$
	\qquad
	\qquad
	\qquad
	\begin{tikzpicture}[baseline=-1ex]
	\tikzstyle{mynode}=[fill=white,circle,draw,minimum size=1.5em,inner sep=0pt]
	\tikzstyle{mylabel}=[fill=white,circle,inner sep=0.5pt]
	  \node[mynode,fill = black!10] (a) at (162:1.5) {$a$};
	  \node[mynode,fill = black!10] (b) at (90: 1.5) {$b$};
	  \node[mynode,fill = black!10] (c) at (18: 1.5) {$c$};
	  \node[mynode] (d) at (306:1.5) {$d$};
	  \node[mynode] (e) at (234:1.5) {$e$};
	  	  
	  \draw[Latex-Latex] (a) edge node[mylabel] {\small$\textcolor{black!20}0$} (b);
	  \draw[-Latex] (c) edge node[mylabel] {\small$\textcolor{black!20}2$} (a);
	  \draw[-Latex] (b) edge node[mylabel] {\small$\textcolor{black!20}2$} (c);
	  \draw[-Latex] (a) edge node[mylabel] {\small$\textcolor{black!20}2$} (d);
	  \draw[-Latex] (a) edge node[mylabel] {\small$\textcolor{black!20}4$} (e);
	  \draw[-Latex] (b) edge node[mylabel] {\small$\textcolor{black!20}2$} (d);
	  \draw[-Latex] (b) edge node[mylabel] {\small$\textcolor{black!20}2$} (e);
	  \draw[-Latex] (c) edge node[mylabel] {\small$\textcolor{black!20}2$} (d);
	  \draw[-Latex] (c) edge node[mylabel] {\small$\textcolor{black!20}2$} (e);
	  \draw[Latex-Latex] (d) edge node[mylabel] {\small$\textcolor{black!20}0$} (e);

	\end{tikzpicture}
			\caption{A preference profile with $N=\{1,2,3,4\}$ and $A=\{a,b,c,d,e\}$ (left-hand side) and the corresponding weighted majority graph (right-hand side). An edge from $x$ to $y$ with weight $w$ denotes that $g_{R}(x,y)=w$. Edges with weight $0$ are bidirectional since, in this case, both alternatives weakly majority dominate each other. The smallest dominant set, $\{a,b,c\}$, is highlighted in gray.
			}
			\label{fig:example}
\end{figure}

\subsubsection*{Social Choice Correspondences}

This paper is concerned with \emph{social choice correspondences (SCCs)}. An SCC maps a preference profile to a non-empty subset of alternatives called the choice set, i.e., it is a function of the form $f:\mathcal{R}^*(A)\rightarrow 2^A\setminus\{\emptyset\}$. Note that we employ a so-called variable population framework, i.e., SCCs are defined for all electorates. In this paper, we focus on two important classes of SCCs: \emph{majoritarian SCCs} and \emph{pairwise SCCs}. An SCC $f$ is called \emph{majoritarian} if its outcome merely depends on the majority relation, i.e., $f(R)=f(R')$ for all $R, R'\in\mathcal{R}^*(A)$ with ${\succ_R}={\succ_R'}$. Furthermore, an SCC $f$ is \emph{pairwise} if its outcome merely depends on the majority margins, i.e., $f(R)=f(R')$ for all $R,R'\in \mathcal{R}^*(A)$ with $g_R=g_{R'}$.
Majoritarian SCCs can be interpreted as functions that map an unweighted graph $(A,\succsim_R)$ to a non-empty subset of its vertices, while pairwise SCCs may additionally use the majority margins $g_R(x,y)$ as weights of the edges. The classes of majoritarian and pairwise SCCs are very rich and contain a variety of well-studied SCCs. For instance, Copeland's rule, the uncovered set, and the bipartisan set are majoritarian, and Borda's rule, the Simpson-Kramer rule, the essential set, Kemeny's rule, ranked pairs, and Schulze's rule are pairwise (the interested reader may consult Chapters 3 and 4 in \citet{BCE+14a} for definitions of these SCCs). All SCCs listed above, except Borda's rule, are \emph{Condorcet extensions}, i.e., they uniquely return the Condorcet winner whenever one exists. 

We say that an SCC $f$ is a \emph{refinement} of an SCC $g$ if $f(R)\subseteq g(R)$ for all preference profiles $R\in \mathcal{R}^*(A)$. 
In this case, $g$ is also said to be \emph{coarser} than $f$. For example, Copeland's rule, the uncovered set, the essential set, Kemeny's rule, ranked pairs, and Schulze's rule are known to be refinements of the top cycle while the Condorcet rule is a coarsening of the top cycle.

The top cycle is a majoritarian SCC that returns the smallest dominant set for a given preference profile. Every preference profile admits a unique smallest dominant set because dominant sets are ordered by set inclusion. Alternatively, the top cycle can be defined based on paths with respect to the majority relation. A path from an alternative $x$ to an alternative $y$ in $\succsim_R$ is a sequence of alternatives $(x_1,\dots, x_k)$ such that $x_1=x$, $x_k=y$, and $x_i\succsim_R x_{i+1}$ for all $i\in \{1,\dots, k-1\}$. The transitive closure $\succsim_R^\ast$ of the majority relation contains all pairs of alternatives $(x,y)$ such that there is a path from $x$ to $y$ in $\succsim_R$. Then, the top cycle can be defined as the set of alternatives that are maximal according to $\succsim_R^\ast$.

\begin{align*}
\tc(R) \;=\; \bigcap \left\{X\subseteq A\colon X\succ_R A\setminus X\right\}\;
=\;\{x\in A\colon x \succsim_R^\ast A \}\tag{Top cycle}
\end{align*}

In other words, the top cycle consists precisely of those alternatives that reach every other alternative on some path in the majority graph. For instance, the top cycle of the example profile in \Cref{fig:example} is $\{a,b,c\}$. Here, it is important that we interpret majority ties as bidirectional edges as otherwise, there would be no path from $a$ to $b$.\footnote{There is a refinement of the top cycle, sometimes called the Schwartz set or GOCHA, which is defined as the union of undominated sets or, alternatively, as the set of alternatives that reach every other alternative on a path according to $\succ_R$ (rather than $\succsim_R$) \citep[see, e.g.,][]{Schw72a,Deb77a,Schw86a}. We will not consider it further because it violates rather mild consistency and strategyproofness conditions.}

% $\{D_1,D_2,\dots,D_k\}$ is totally ordered by set inclusion, i.e., $D_1\subsetneq D_2\subsetneq \dots \subsetneq D_k$.
%
% An important observation is that every alternative in $D_i\setminus D_{i-1}$ has a path to every alternative in $A\setminus D_{i-1}$, see \Cref{rem:CyclesInDominantSets} for details. This fact can also be inverted: if there is a path from an alternative $a$ to an alternative $b$ in $R_M$, then the smallest dominant set containing $a$ is a subset of or equal to the smallest dominant set containing $b$. This means in particular that an alternative $a$ is only contained in the largest dominant set if all other alternatives have a path to $a$ in $R_M$. 	

On top of majoritarianess and pairwiseness, which restrict the informational basis of SCCs, we now introduce a number of additional properties of SCCs. 

\begin{itemize}
\item An SCC is \emph{non-imposing} if for every alternative $x\in A$, there is a profile $R\in\mathcal{R}^*(A)$ such that $f(R)=\{x\}$.
%\item An SCC is \emph{set non-imposing} if for every non-empty set of alternatives $X\subseteq A$, there is a profile $R$ such that $f(R)=X$.
\item An SCC is \emph{neutral} if $f(R')=\pi(f(R))$ for all electorates $N\in\mathcal{F}(\mathbb{N})$ and profiles $R, R'\in \mathcal{R}(A)^N$ for which there is a permutation $\pi:A\rightarrow A$ such that $x \succ_i y$ iff $\pi(x) \succ_i' \pi(y)$ for all alternatives $x,y\in A$ and voters $i\in N$.
\item An SCC is \emph{homogeneous} if for all preference profiles $R\in\mathcal{R}^*(A)$, $f(R)=f(k R)$ where the profile $k R$ consists of $k$ copies of $R$.
\end{itemize}
Non-imposition is a mild decisiveness requirement demanding that every alternative will be selected uniquely for some configuration of preferences. It is weaker than Pareto-optimality and unanimity (an alternative that is top-ranked by all voters has to be elected uniquely).
Neutrality requires that if alternatives are relabelled in a preference profile, the alternatives in the corresponding choice set are relabelled accordingly.
Homogeneity states that cloning the entire electorate will not affect the choice set. All three of these properties are very mild and satisfied by all SCCs typically considered in the literature, including the top cycle.

\subsubsection*{Fishburn's Extension and Strategyproofness}

An important desirable property of SCCs is strategyproofness, which demands that voters should never be better off by lying about their preferences. In order to make this formally precise for social choice \emph{correspondences}, we need to make assumptions about the voters' preferences over sets. In this paper, we extend the voters' preferences over alternatives to incomplete preference over sets by using Fishburn's preference extension. 
Given two sets of alternatives $X, Y\subseteq A$, $X\neq Y$, and a preference relation $\succ_i$, Fishburn's extension is defined by 
\[
X \succ_i^F Y \quad\text{iff}\quad X\setminus Y \succ_i Y \text{ and } X \succ_i Y\setminus X\text. \tag{Fishburn's extension}
\]
Fishburn's extension is frequently considered in social choice theory and can be justified in various ways \citep[see, e.g.,][]{Gard79a,ChZh02a,ErSa09a,BSS19a}. For example, one motivation assumes that a single alternative will eventually be selected from each choice set according to a linear tie-breaking ordering (such as the preference relation of a chairperson) and that voters are unaware of the concrete ordering used to break ties. Then, set $X$ is preferred to set $Y$ iff for all tie-breaking orderings, the voter weakly prefers the alternative selected from $X$ to that selected from $Y$ and there is at least one ordering for which this comparison is strict. Another motivation is based on a function that assigns an \emph{a priori} weight to each alternative such that each choice set can be mapped to a lottery over the chosen alternatives such that the probabilities are proportional to the alternatives' weights. Again the voters are unaware of the concrete weight function and prefer set $X$ to set $Y$ iff for all utility functions consistent with their ordinal preferences and all \emph{a priori} weight functions, the expected utility derived from $X$ is higher than that derived from $Y$.

Strategyproofness based on Fishburn's extension can be defined as follows.
An SCC $f$ is \emph{(Fishburn-)strategyproof} if for all electorates $N\in\mathcal{F}(\mathbb{N})$ and profiles $R\in\mathcal{R}(A)^N$, there is no profile $R'\in\mathcal{R}(A)^N$ such that ${\succ_j}={\succ_j'}$ for all $j\in N\setminus \{i\}$ and $f(R')\succ_i^F f(R)$.

\begin{figure}[tb]
	\centering	 $
	 \begin{array}{r@{\colon\;\;}l@{\;\;}l@{\;\;} l@{\;\;} l@{\;\;} l@{\;\;} l@{\;\;} l@{\;\;} l@{\;\;} l@{\;\;} l}
			i\in \{1,2\}	& a	&	\succ_i	& b	&	\succ_i	&	c\\
			i\in \{3,4\}	& c	&	\succ_i	& a	&	\succ_i	&	b\\
			5	& b	&	\succ_5	& c	&	\succ_5	&	a
			\end{array}
		$
		\hspace{2cm}
   	 $
   	 \begin{array}{r@{\colon\;\;}l@{\;\;}l@{\;\;} l@{\;\;} l@{\;\;} l@{\;\;} l@{\;\;} l@{\;\;} l@{\;\;} l@{\;\;} l}
			i\in \{1,2\}	& a	&	\succ_i	& b	&	\succ_i	&	c\\
			i\in \{3,4\}	& c	&	\succ_i	& a	&	\succ_i	&	b\\
			5	& c	&	\succ_5	& b	&	\succ_5	&	a
			\end{array}
   		$
		\caption{Example showing the Fishburn-manipulability of plurality rule. Plurality rule chooses $\{a,c\}$ for the left profile and $\{c\}$ for the right profile, so voter $5$ can manipulate by deviating from the profile on the left to the one on the right. The same example shows that many other popular SCCs are Fishburn-manipulable, e.g., Borda's rule, Nanson's rule, Black's rule, the maximin rule, Bucklin's rule, Young's rule, and Kemeny's rule (we refer to \citet{BCE+14a} for definitions of these SCCs).}
		% Tideman's rule,  (FB: this rule is very little known; it's not ranked pairs)
		\label{Fig:manipulation}
\end{figure}

Even though Fishburn-strategyproofness seems like a relatively weak strategyproofness notion, it is only satisfied by very few SCCs. In particular, the omninomination rule, the Pareto rule, the top cycle, and the Condorcet rule are Fishburn-strategyproof, while virtually all other commonly studied SCCs fail Fishburn-strategyproofness. As an example of this claim, we give an example demonstrating the Fishburn-manipulability of plurality rule (which chooses the alternatives top-ranked by most voters) and many other SCCs in \Cref{Fig:manipulation}.

\section{Results}\label{sec:PF}

	Our first result is a complete characterization of pairwise strategyproof SCCs in terms of dominant set rules. An SCC is a \emph{dominant set rule} if for every preference profile $R$, it returns a dominant set with respect to $\succsim_R$. Examples of dominant set rules are the top cycle, the Condorcet rule (which returns the Condorcet winner if it exists and all alternatives otherwise) and the Condorcet non-loser rule (which returns all alternatives but a Condorcet loser). Observe that---even though dominant set rules may seem rather restricted---they allow for a fair degree of freedom in the choice of dominant sets. For instance, dominant set rules can take the majority margins into account. This is demonstrated by the SCC that returns an alternative $x$ as unique winner if $g_{R}(x,y)>2$ for all $y\in A\setminus \{x\}$ and otherwise returns all alternatives. It is also possible to define rather unnatural dominant set rules such as the SCC that returns the smallest dominant set whenever the majority graph contains a cycle with identical weights and the second smallest dominant set otherwise. 
	
	For our analysis, it suffices to consider the particularly simple subclass of robust dominant set rules: a dominant set rule $f$ is \emph{robust} if $f(R')\subseteq f(R)$ for all preference profiles $R, R'$ such that $f(R)$ is dominant in $R'$. In other words, if the choice set $f(R)$ for some profile $R$ is also dominant in another profile $R'$, then no alternative outside of $f(R)$ can be chosen for $R'$.
It is easily seen that the top cycle, the Condorcet rule, and the Condorcet non-loser rule are robust dominant set rules, whereas the two artificial examples given above fail this condition. Moreover, robust dominant set rules are majoritarian and therefore also homogeneous.
	
Our first theorem shows that---under mild additional assumptions---robust dominant set rules are the \emph{only} pairwise SCCs that satisfy strategyproofness.
	
	\begin{restatable}{theorem}{FSP}\label{thm:F-SP}
		Let $f$ be a pairwise SCC that satisfies non-imposition, homogeneity, and neutrality. Then, $f$ is strategyproof iff it is a robust dominant set rule. 
	\end{restatable}
	
	The direction from right to left is relatively straightforward: every robust dominant set rule $f$ is strategyproof because robustness prohibits successful manipulations of dominant set rules. Consider, for example, that voter $i$ manipulates from a profile $R$ to another profile $R'$ such that $f(R')\subsetneq f(R)$. According to Fishburn's extension we have that $f(R') \succ_i f(R)\setminus f(R')$. Moreover, $f(R') \succ_{R'} A\setminus f(R')$ because $f(R')$ is dominant in $\succsim_{R'}$. Since voter $i$ can only weaken alternatives in $f(R')$ against those in $f(R)\setminus f(R')$ when moving from $R$ to $R'$, $f(R')$ will only be strengthened against $f(R)\setminus f(R')$ when moving from $R'$ to $R$. Hence, $f(R') \succ_{R} f(R)\setminus f(R')$. This implies that $f(R')$ is also dominant in $\succsim_{R}$ because $f(R')\subsetneq f(R)$ and $f(R)$ is dominant in $R$. 
Since $f(R')$ is dominant in both $R$ and $R'$, robustness implies $f(R)\subseteq f(R')$, which is at variance with our initial assumption that $f(R')\subsetneq f(R)$.
	A similar argument applies to the case that $f(R')\not\subseteq f(R)$ where strategyproofness now implies that $f(R')\setminus f(R)\succ_i f(R)$. 

	The converse direction---i.e., every pairwise SCC that satisfies non-imposition, homogeneity, neutrality, and strategyproofness is a robust dominant set rule---is much more difficult to prove. As a first step, we investigate the consequences of strategyproofness for pairwise SCCs. It turns out that we can abstract away from the concrete preferences of the voters and derive multiple axioms that describe how the choice set is affected when modifying the pairwise comparisons between pairs of alternatives.
	For instance, we show that rearranging unchosen alternatives in the voters' preferences does not affect the choice set of strategyproof and pairwise SCCs. The proofs of these implications heavily use the fact that two voters with preferences inverse to each other can be added to a preference profile without affecting the outcome of pairwise SCCs.
	%While pairwise SCCs dismiss information about the voters' exact preferences, strategyproofness requires this information to compute feasible outcomes. Since we want to be able to swap alternatives without knowing a voter's exact preferences, we derive multiple axioms describing how a strategyproof and pairwise SCC has to behave in various situations. 
	%For instance, we show that every strategyproof and pairwise SCC satisfies independence of irrelevant comparisons, which requires that, if reordering some alternatives does not affect their membership in the choice set, the choice set should not change at all. After deriving these axioms, we only use these properties instead of strategyproofness as they are easier to handle. 
	As a second step, we use these axioms to derive some insights on the structure of choice sets returned by pairwise SCCs that satisfy strategyproofness, non-imposition, homogeneity, and neutrality. In more detail, we show for such an SCC $f$ that \emph{(i)} it chooses a single winner iff it is the Condorcet winner and \emph{(ii)} for every alternative $x\in A$, either $f(R)=\{x\}$ or there is an alternative $y\in f(R)\setminus \{x\}$ such that $y\succsim_R x$. The first condition is called strong Condorcet-consistency, and we show that every pairwise, homogeneous, strategyproof, and strongly Condorcet-consistent coarsening of the top cycle is a robust dominant set rule.
	As the last step, we show that every pairwise SCC that satisfies the given axioms is a coarsening of the top cycle, which completes the proof of the theorem.

	 \Cref{thm:F-SP} has a number of important and perhaps surprising consequences. For instance, it implies that every strategyproof and pairwise SCC that satisfies the given conditions has to be majoritarian. In other words, every such SCC completely ignores the absolute values of majority margins, even though these values allow for the definition of more sophisticated SCCs.\footnote{This insight resembles the fact that individual preference intensities can usually not be used by strategyproof voting rules \citep[see, e.g.,][]{Nand12a,EMMS20a}. Note, however, that majority margins represent collective preference intensities.} Moreover, \Cref{thm:F-SP} entails that many strategyproofness notions are equivalent under the assumptions of the theorem because robust dominant set rules satisfy much stronger notions of strategyproofness than Fishburn-strategyproofness (see \Cref{rem:strongsp}).
	 
	 Another interesting consequence of \Cref{thm:F-SP} is that the top cycle is the finest pairwise SCC that satisfies strategyproofness, non-imposition, homogeneity, and neutrality since it returns the smallest dominant set for any given preference profile. As shown in the sequel, we can turn this observation into a characterization of the top cycle by replacing non-imposition and neutrality with set non-imposition. An SCC $f$ satisfies \emph{set non-imposition} if for every non-empty set $X\subseteq A$, there is a profile $R$ such that $f(R)=X$. In other words, every set is chosen for some preference profile $R$, which is in line with the original motivation of non-imposition for SCFs: the functions's image coincides with its codomain. For neutral and pairwise SCCs, set non-imposition can be interpreted as a weak efficiency notion. To see this, assume that there is some set $X\subseteq A$ that is never returned by $f$. Now, consider a profile $R$ such that $X\succ_i A\setminus X$ for all $i\in N$ and $x\sim_R y$ for all $x,y\in X$ and $x,y\in A\setminus X$. Neutrality and pairwiseness imply that $f$ can only return $X$, $A$, or $A\setminus X$. Since $f$ never returns $X$ by assumption, we have that $A\setminus X \subseteq f(R)$. However, every voter $i\in N$ prefers $X$ to $A\setminus X$ and the choice set of $f$ is thus very inefficient.
	 
The following lemma shows how set non-imposition can be used to single out the top cycle among all robust dominant set rules.

	\begin{lemma}\label{lem:TCweak}
		The top cycle is the only robust dominant set rule that satisfies set non-imposition. 
	\end{lemma}
	\begin{proof}
		It has already been stated that \tc is a robust dominant set rule. 
		%Moreover, it follows from \Cref{thm:F-SP} that \tc also satisfies strategyproofness. 
		Moreover, \tc satisfies set non-imposition because every set $X$ is the smallest dominant set for every profile $R$ such that $x\sim_R y$ for all $x,y\in X$ and $X\succ_R A\setminus X$; the existence of such a profile follows from McGarvey's construction \citep{McGa53a}. %Hence, \tc satisfies all axioms. 
		
		For the other direction, assume for contradiction that there is another robust dominant set rule $f\neq \tc$. Since $f$ is not the top cycle, there is a profile $R$ with dominant sets $D_1,\dots, D_k$ such that $D_i\subseteq D_j$ iff $i\leq j$ and $f(R)=D_i$ with $i\geq 2$. This means that there is no profile $R'$ such that $f(R')=D_1$ since otherwise, robustness from $R'$ to $R$ implies that $f(R)\subseteq f(R')=D_1$ because $f(R')$ is dominant in $\succsim_R$. In other words, $f$ violates set non-imposition, contradicting our initial assumption. As a consequence, the top cycle is indeed the only robust dominant set rule that satisfies set non-imposition.
	\end{proof}
	
The combination of \Cref{thm:F-SP} and \Cref{lem:TCweak} already characterizes the top cycle as the only pairwise SCC that satisfies strategyproofness, set non-imposition, homogeneity, and neutrality. It turns out that neutrality is not required for this characterization as the insights of the proof of \Cref{thm:F-SP} can be leveraged to establish that only robust dominant set rules satisfy pairwiseness, strategyproofness, set non-imposition, and homogeneity. In particular, the axioms of \Cref{thm:TC} also imply strong Condorcet-consistency, which together with our previous insights and \Cref{lem:TCweak} yields the following characterization.
	
	\begin{restatable}{theorem}{TC}\label{thm:TC}
		The top cycle is the only pairwise SCC that satisfies strategyproofness, set non-imposition, and homogeneity.
	\end{restatable}

%It should be stressed that, in contrast to \Cref{thm:F-SP}, this characterization does not require neutrality. 
%As a consequence, 
%\Cref{thm:TC} does not simply follow from \Cref{thm:F-SP}, but the proofs of both results are similar because the axioms in \Cref{thm:TC} can be shown to imply strong Condorcet-consistency.

	We conclude the paper with a number of remarks.
	
	\begin{remark}[Independence of the axioms]\label{rem:independence}
		We can show that all of the axioms, except non-imposition, are required for the direction from left to right of \Cref{thm:F-SP}.
		%showing that every pairwise SCC that satisfies strategyproofness, homogeneity, non-imposition, and neutrality is a robust dominant set rule. 
		If we only omit pairwiseness, the omninomination rule satisfies all required axioms but is no dominant set rule. If we dismiss neutrality, the following SCC based on two special alternatives $a$ and $b$ satisfies all requirements, but is no dominant set rule: $f^{ab}$ returns $\{a\}$ if $a\succ_R  A\setminus \{a,b\}$ and $a\succsim_R b$; otherwise it returns the outcome of the Condorcet rule. All axioms except homogeneity are satisfied by the SCC $\tc^\ast$ which returns the top cycle with respect to the relation $x \succsim_{R}^\ast y$ iff $g_{R}(x,y)\geq -1$. However, $\tc^\ast$ is no robust dominant set rule because it depends on the majority margins. It is open whether non-imposition is required for \Cref{thm:F-SP}. We discuss a variant of \Cref{thm:F-SP} which uses strong Condorcet-consistency instead of neutrality and non-imposition in the appendix. For this variant, it is easy to prove that all axioms are indeed required.
		
		For the converse direction of \Cref{thm:F-SP}, none of the auxiliary axioms is required. In particular, every robust dominant set rule is homogenous and pairwise because robustness entails majoritarianess. Moreover, these SCCs satisfy strategyproofness regardless of whether they are neutral or non-imposing. For instance, the SCC that chooses the set $\{a,b,c\}$ if it is a dominant set and all alternatives otherwise is neither neutral nor non-imposing, but it is a robust dominant set rule and strategyproof.
		
	For \Cref{thm:TC}, we can show the independence of all axioms. Borda's rule only violates strategyproofness, the Condorcet rule only violates set non-imposition, the omninomination rule only violates pairwiseness, and $\tc^\ast$ only violates homogeneity.
	\end{remark}
	
	\begin{remark}[Tournaments]\label{rem:UC}
		A significant part of the literature focuses on the special case when there are no majority ties and the majority graph is a tournament \citep[see, e.g.,][]{Lasl97a,BBH15a}. This, for example, happens when the number of voters is odd.
		% We can't really assume state Theorem 1 for odd n because of homogeneity. 
		In the absence of majority ties and when $m\leq 4$, there is a strategyproof SCC known as the uncovered set which satisfies all requirements of \Cref{thm:F-SP} but is no dominant set rule. When $m\geq 5$, the uncovered set violates strategyproofness and \Cref{thm:F-SP} holds even in the absence of majority ties. 
%A well-known majoritarian refinement of the top cycle is the uncovered set \citet{McKe86a}, which chooses all maximal elements with respect to the covering relation $\succsim^C$. This relation is defined by $x\succsim^C y$ if $x\succ_R y$, $z\succ_R x$ implies $z \succ_R y$, and $z\sim_R x$ implies $z \succsim_R y$ for all $z\in A\setminus \{x,y\}$. This SCC satisfies all requirements of \Cref{thm:F-SP} but is no dominant set rule if $m= 4$ and $n$ is odd. Note that the latter assumption is common in the analysis of majoritarian SCCs as it ensures that the majority relation is strict. 
% Doesn't need to be McKelvey UC, the other variants of UC work as well (also MC, BP, etc.). The rule only differs from TC on one tournament
	\end{remark}

	\begin{remark}[Dropping homogeneity]\label{rem:variants}
The example given in \Cref{rem:independence} for the independence of homogeneity only shows that robustness might be violated if we dismiss homogeneity, but the considered SCC is still a dominant set rule. It turns out that this observation is true in general if we mildly strengthen non-imposition to \emph{unanimity} (a unanimously top-ranked alternatively will be selected uniquely): every pairwise SCC that satisfies strategyproofness, unanimity, and neutrality is a dominant set rule if $m\neq 4$. The last condition is required because of the uncovered set discussed in \Cref{rem:UC}. By weakening robustness, one can thus obtain an alternative characterization of strategyproof SCCs based on \emph{weak robustness}: an SCC $f$ is weakly robust if $f(R')\subseteq f(R)$ for all preference profiles $R$, $R'$ such that $g_{R}(x,y)\leq g_{R'}(x,y)$ for all $x\in f(R)$, $y\in A\setminus f(R)$. Then, if $m\neq 4$, every pairwise SCC that satisfies unanimity and neutrality is strategyproof iff it is a weakly robust dominant set rule. 
	\end{remark}
	%See tc_work.tex for a proof.

	\begin{remark}[Weakening neutrality]
		Another variant of \Cref{thm:F-SP} can be obtained by weakening neutrality to the following condition: $x\in f(R)$ iff $y\in f(R)$ for every preference profile $R$ and all pairs of alternatives $x,y\in A$ such that $g_{R}(x,y)=0$ and $g_{R}(x,z)=g_R(y,z)$ for all $z\in A\setminus \{x,y\}$. This condition is void if there is an odd number of voters. As a result, homogeneity becomes more important for the proof as some steps only work for an even number of voters. 
	\end{remark}
	
	\begin{remark}[Weakening non-imposition]
		In the presence of neutrality, non-imposition can be weakened to a condition that merely requires that the SCC returns a singleton set for at least one profile. If we weaken neutrality, this is no longer possible and our proof suggests that, among the three auxiliary axioms, non-imposition plays the most important role as it is crucial for deriving strong Condorcet-consistency.
	\end{remark}
	
	\begin{remark}[Strengthening strategyproofness]
		\label{rem:strongsp}
		Fishburn-strategyproofness is a rather weak strategyproofness notion, which makes the direction from left to right in our characterizations strong. However, robust dominant set rules---especially the top cycle---are actually much more resistant against manipulation. In order to formalize this, we introduce a new preference extension, denoted by $\succsim^{F+}$ based on the relation $\succ^{\exists}_i$ over the subsets of $A$. This relation is defined as $X\succsim_i^{\exists} Y$ iff $X=\emptyset$, $Y=\emptyset$, or there are alternatives $x\in X$, $y\in Y$ such that $x\succ_i y$. Then, 
		%The idea underlying $\succsim^{\exists}_i$ is that a set $X$ is preferred to another set $Y$ if there is an alternative in $X$ that is preferred to an alternative in $Y$. Then, 
		%voter $i$ prefers a set $X$ weakly to another set $Y$ according to $\succsim^{F+}$ iff
		\begin{align*}
			X\succsim_i^{F+} Y \quad\text{iff}\quad X\setminus Y\succ_i Y\setminus X \text{ and } X\setminus Y\succ_i^\exists X\cap Y \text{ and } X\cap Y\succ_i^{\exists} Y\setminus X.
		\end{align*}
		Clearly, $X\succsim_i^{F} Y$ implies $X\succsim_i^{F+} Y$ and consequently $F+$-strategyproofness is stronger than Fishburn-strategyproofness. We define an even stronger notion of strategyproofness based on the $\succsim^{F+}$ extension as follows: an SCC $f$ is \emph{strongly} $\succsim^{F+}$-strategyproof if $f(R)\succsim_i^{F+} f(R')$ for all voters $i\in N$ and preference profiles $R$, $R'$ with ${\succ_j}={\succ_j'}$ for all $j\in N\setminus \{i\}$. Strong $\succsim^{F+}$-strategyproof requires that all choice sets for manipulated preference profiles are comparable to the original choice set, making it much stronger than both $F+$-strategyproofness and Fishburn-strategyproofness. 
		Strong Fishburn-strategyproofness can be defined analogously. The top cycle is strongly $\succsim^{F+}$-strategyproof. Interestingly,  \citet{ChZh02a} have shown that only dictatorial and constant SCCs satisfy the slightly stronger notion of strong Fishburn-strategyproofness, which obviously rules out the top cycle. 
		%Moreover, it should be mentioned that not all robust dominant set rules satisfy strong $\succsim^{F1}$-strategyproofness, e.g., the Condorcet non-loser rule violates this axiom. Instead, they satisfy a slightly weaker variant of strategyproofness depending on the choice sets $f(R)$ and $f(R')$ before and after a manipulation: if $f(R)\subseteq f(R')$, $f(R)\subseteq f(R')$, or $f(R)\cap f(R')=\emptyset$, then every robust dominant set rule $f$ satisfies strong $\succsim^{F1}$-strategyproofness. The remaining case is that $f(R)\setminus f(R')\neq\emptyset$, $f(R')\setminus f(R)\neq\emptyset$, and $f(R)\cap f(R')\neq\emptyset$, where we only can guarantee that $f(R)\setminus f(R')\succ_i f(R')\setminus f(R)$ and that at least one of the two existence conditions is true. This is still a very strong strategyproofness notion and thus, it follows that many strategyproofness notions are equivalent in our setting (i.e., pairwise SCCs that satisfy non-imposition, homogeneity, and neutrality) because all such SCCs that satisfy Fishburn-strategyproofness also satisfy this more demanding notion. 
	\end{remark}
	
	\begin{remark}[Group-strategyproofness]
		An SCC $f$ is group-strategyproof if for all preference profiles $R$, $R'$ and sets of voters $G\subseteq N$ such that ${\succ_j}={\succ_j'}$ for $j\in N\setminus G$, it holds that $f(R')\not\succ_i^F f(R)$ for some voter $i\in G$. Since every robust dominant set rule is group-strategyproof, it follows from \Cref{thm:F-SP} that strategyproofness is equivalent to group-strategyproofness for pairwise SCCs that satisfy homogeneity, non-imposition, and neutrality.
	\end{remark}
	
	\begin{remark}[Pareto-optimality]\label{rem:tcpo}
		The main disadvantage of the top cycle is that it may return Pareto-dominated alternatives. In fact, \emph{every} strategyproof pairwise SCC that satisfies our assumptions violates Pareto-optimality. However, it is possible to circumvent this impossibility by first removing all Pareto-dominated alternatives and then computing the top cycle of the remaining alternatives. This SCC, $\tc(\po)$ where $\po$ stands for the Pareto rule, was already considered by \citet{Bord79a} and can be shown to be strategyproof. In fact, it satisfies all conditions of \Cref{thm:F-SP} except pairwiseness since it is not possible to compute the set of Pareto-dominated alternatives based on the majority margins only. Interestingly, the ``converse'' SCC, $\po(\tc)$, which first computes the top cycle and then removes all Pareto-dominated alternatives, is nested in between $\tc(\po)$ and $\tc$, but violates strategyproofness.
%	\todo{FB: Suzumura (1983) also discusses $\tc(\po)$}
	\end{remark}

	\begin{remark}[Fishburn-efficiency]
		As discussed in the previous remark, the top cycle fails Pareto-optimality. 
However, the top cycle satisfies the weaker notion of \emph{Fishburn-efficiency}, which requires that for every profile $R$, there is no set of alternatives $X$ such that $X\succ_i^F f(R)$ for all $i \in N$. Fishburn-efficiency can be seen as a weak form of \emph{ex ante} efficiency, where outcomes are compared before ties are broken.
It is easy to see that the top cycle is the only robust dominant set rule satisfying this axiom since every other such rule already violates set non-imposition.
It can moreover be shown that the top cycle is the coarsest majoritarian SCC that satisfies Fishburn-efficiency, i.e., every majoritarian SCC $f$ that is Fishburn-efficient satisfies that $f(R)\subseteq \tc(R)$ for all preference profiles $R$. Since the top cycle is also the finest majoritarian SCC that satisfies Fishburn-strategyproofness, neutrality, and non-imposition, it can be completely characterized using strategyproofness and efficiency. 
 % "a" finest when dropping non-imposition?
	\end{remark}
	
	\begin{remark}[Beyond majority relations]
		Dominant set rules can be defined with respect to any complete binary relation derived from the preference profile. To formalize this idea, let the \emph{information base} $I(R)$ denote a function that maps $R\in\mathcal{R}^*(A)$ to a complete binary relation $\succsim_{I(R)}$ on $A$. Applying a dominant set rule to $\succsim_{I(R)}$ clearly results in an SCC. Moreover, if $I(R)$ is \emph{local} (i.e., $a\succsim_{I(R)} b$ iff $a\succsim_{I(R')} b$ for all $a,b\in A$ and $R, R'\in\mathcal{R}(A)^N$ such that $a\succ_i b$ iff $a\succ_i' b$ for all $i\in N$) and \emph{monotone} (i.e., $a\succsim_{I(R)} b$ implies $a\succsim_{I(R')} b$ for all $a,b\in A$ and $R,R'\in\mathcal{R}^*(A)$ such that $R'$ is derived from $R$ by reinforcing $a$ against $b$ in the preference relation of a voter $i$), then every robust dominant set rule on $\succsim_{I(R)}$ is strategyproof. This proves, for instance, that dominant set rules based on \emph{supermajority relations} (i.e., $a\succsim_{I(R)} b$ iff $g_R(a,b)\geq -k$ for some $k\in \mathbb{N}$) or on \emph{shifted majority relations} (i.e., $a\succ_{I(R)} b$ if $g_{R}(a,b)>k$, $a\sim_{I(R)} b$ if $g_{R}(a,b)=k$, and $b\succ_{I(R)} a$ otherwise) are strategyproof. 
		
		For some information bases $I(R)$, it is even possible to prove statements analogous to \Cref{thm:F-SP} when demanding exclusive dependence on $I(R)$. To this end, we say an SCC is \emph{$I(R)$-based} for some information basis $I(R)$ if $f(R)=f(R')$ for all preference profiles $R,R'\in\mathcal{R}^*(A)$ such that ${\succsim_{I(R)}}={\succsim_{I(R')}}$. This extends the definition of majoritarianess. 
		For instance, it is easy to derive from our proof that robust dominant set rules on a supermajority relation $I(R)$ are the only neutral, non-imposing, strategyproof and $I(R)$-based SCCs. An equivalent statement holds for shifted majority relations $I(R)$ when defining neutrality based on~$\succsim_{I(R)}$.
	\end{remark}
	
	\begin{remark}[Fixed electorates]
		A non-standard assumption in our model is that of a variable electorate. 
		This assumption is necessary, because, when fixing the number of voters, the Pareto rule satisfies all axioms of \Cref{thm:F-SP} and \Cref{thm:TC} but fails to be a dominant set rule. We now sketch two approaches to adapt our results to a fixed electorate framework. 
		Firstly, we may replace pairwiseness and homogeneity with majoritarianess. Since the construction of \citet{McGa53a} allows us to build every majority relation with at most $m^2$ voters, we need at most $m^2+2$ voters for our results to hold under majoritarianess. The second approach is to restrict attention to profiles whose maximal majority margin is bounded by a constant $c\geq 2$. This is possible because we never need to increase the maximal majority margin to a value larger than $c$ in our proofs.\footnote{This is not in conflict with the fact that we sometimes use homogeneity to duplicate preference profiles in the proof, because it is either possible to entirely avoid these homogeneity applications, or to ensure that all majority margins are $1$ before duplicating the profile.} Using again \citeauthor{McGa53a}'s construction, every profile with a majority margin of at most $c$ can be built with $cm^2$ voters and we thus need at most $cm^2+2$ voters for our proof. Hence, we can show that every non-imposing, neutral, strategyproof, homogeneous, and pairwise SCC is a robust dominant set rule for profiles with maximal majority margin of at most $c$ if there are $cm^2+2$ voters (here, homogeneity is defined for majority margins).
	\end{remark}

\section*{Acknowledgements}
This material is based on work supported by the Deutsche Forschungsgemeinschaft under grants {BR~2312/11-2} and {BR~2312/12-1}. 
Results from this paper were presented at the COMSOC video seminar (November 2021) and the 16th Meeting of the Society of Social Choice and Welfare in Mexico City (June 2022). 
The authors thank Florian Brandl and the anonymous referees for helpful comments.

	\newpage
	\appendix
	
	\section{Omitted proofs}
	\label{AppendixC}
	
	This appendix contains the proofs of \Cref{thm:F-SP,thm:TC}. Proof sketches for these results were given in \Cref{sec:PF} and we here focus on the details. Since the proofs are rather involved, we divide them into multiple lemmas which are organized in subsections to highlight related ideas. In particular, we discuss additional notation in \Cref{app:notation}, some general results on the structure of the top cycle in \Cref{app:structure}, implications of Fishburn-strategyproofness for pairwise SCCs in \Cref{app:axioms}, a variant of \Cref{thm:F-SP} that relies on strong Condorcet-consistency in \Cref{app:condorcet}, and finally the proofs of our main results in \Cref{app:mainresults}.
	
	\subsection{Notation}
	\label{app:notation}
	
	Before discussing our proofs, we need to introduce some additional notation. First, we specify how we denote preference relations. We usually write preference relations as comma-separated lists. In these lists, we use $\lex(X)$ and $\lex(X)^{-1}$ to indicate that the alternatives in a set $X$ are ordered lexicographically or inversely lexicographically. For instance, $a, \lex(\{b,c\}),d$ is equivalent to $a,b,c,d$ and means that $a$ is preferred to $b$, $b$ to $c$, and $c$ to $d$. Similarly, $a, \lex(\{b,c\})^{-1},d$ is equivalent to $a,c,b,d$. Furthermore, we occasionally interpret a voter's preference relation as a set of tuples and use set operations such as set intersections and set differences to form new preference relations. In particular, we write $\succ_i\!|_X$ to denote the restriction of $\succ_i$ to $X$, i.e., $\succ_i\!|_X={\succ_i}\cap X^2$. We use the same notation for the majority relation, i.e., $\succsim_R\!|_{X}$ denotes the restriction of $\succsim_R$ to $X$. For instance, $\succsim_R\!|_{X}=\;\succsim_R'\!|_{X}$ means that the majority relations of $R$ and $R'$ agree on the alternatives in $X$. 
	
	The second important concept is that of cycles in the majority relation. A \emph{cycle} in a majority relation $\succsim_R$ is a sequence of $q\geq 2$ alternatives $(a_1, \dots, a_q)$ such that $a_{i} \succsim_R a_{i+1}$ for all $i\in \{1,\dots, q-1\}$, $a_q \succsim_R a_1$, and $a_i\neq a_j$ for all distinct $i,j\in \{1,\dots, q\}$. Informally, a cycle is a path in $\succsim_R$ that starts and ends at the same alternative and visits every alternative on the cycle (except the first one) only once. For instance, in the majority graph in \Cref{fig:example}, $C=(a,b,c)$ is a cycle. While slightly overloading notation, we denote with $C$ both the ordered sequence of alternatives that defines a cycle and the set of alternatives contained in the cycle.
	
	Finally, we introduce the notions of connectors and connected sets. The \emph{connected set} $A_x$ of an alternative $x\in A$ in a profile $R$ contains all alternatives (except $x$) that drop out of the top cycle if we remove $x$ from the preference profile, i.e., $A_x=\tc(R)\setminus \left(\tc(R|_{A\setminus \{x\}})\cup \{x\}\right)$. %\todo{FB: call this the ``tail'' of alternative $x$?}
	The notion of connected sets helps us to distinguish the alternatives in the top cycle further: we say an alternative $x\in A$ is a \emph{connector} in $R$ if $A_x\neq\emptyset$. This means intuitively that, if we remove $x$ from the preference profile, $x$ and additional alternatives drop out of the top cycle. In other words, $x$ connects the alternatives in $A_x$ to the rest of the top cycle. Note that an alternative $x\not\in \tc(R)$ cannot be a connector since $A_x=\emptyset$ for these alternatives and that connectors only exist if $|\tc(R)|\geq 3$.
	
	\subsection{Structure of the Top Cycle}
	\label{app:structure}
	
	For the proofs of our results, it will be helpful to have a deeper understanding of the structure of the top cycle. In more detail, we first show that the top cycle is closely connected to cycles in the majority relation. \citet{Moul86a} has shown such a statement under the assumption that there are no majority ties: there is a cycle in the majority relation $\succsim_R$ that connects all the alternatives in $\tc(R)$. Since we need to allow for majority ties, we generalize this result by interpreting majority ties as bidirectional edges.
	
	%\todo{FB: Moulin attributes this to \citet{Mill77a}}
	
	\begin{lemma}\label{lem:HC}
		Let $R$ be a preference profile. It holds for a set $X\subseteq A$ with $|X|\geq 2$ that $\tc(R)=X$ iff there is a cycle $C=(a_1, \dots, a_{|X|})$ in $\succsim_R$ such that $C=X$ and $X \succ_R A\setminus X$. Furthermore, $\tc(R)=\{x\}$ iff $x$ is the Condorcet winner in $R$. 
	\end{lemma}
	\begin{proof}
		We first prove that $\tc(R)=\{x\}$ iff $x$ is the Condorcet winner in $R$. Thus, note that $\{x\}=\tc(R)$ implies that $x \succ_R A\setminus \{x\}$ because the top cycle returns a dominant set. Hence, $x$ is the Condorcet winner if it is the unique winner of the top cycle. Next, let $x$ denote the Condorcet winner in a preference profile $R$. It follows that $x \succ_R A\setminus \{x\}$ and therefore, $\{x\}$ is a dominant set. Even more, it is obviously the smallest dominant set and thus $\tc(R)=\{x\}$, which proves the first claim.
		
		Next, we focus on sets of alternatives $X\subseteq A$ with $|X|\geq 2$ and show first that if $X=\tc(R)$, there is a cycle $C$ in $\succsim_R$ such that $C=X$ and $X\succ_R A\setminus X$. Since the latter condition directly follows from the definition of the top cycle, we only have to show that there is a cycle in $\succsim_R$ containing all alternatives in $X$. Note for this that if there is an alternative $x\in X$ with $x \succ_R X\setminus\{x\}$, this alternative is the Condorcet winner and $\tc(R)=\{x\}\neq X$. Consequently, for every alternative $x\in X$, there is another alternative $y\in X\setminus \{x\}$ such that $y \succsim_R x$. This means that there is a cycle in $\succsim_R\!|_{X}$. Let $C=(a_1, ..., a_q)$ denote an inclusion maximal cycle in $\succsim_R|_X$ and assume for contradiction that there is an alternative $y\in X\setminus C$. 
		
		As a first step, consider the case that there are two distinct alternatives $a_i$, $a_j\in C$ such that $a_i \succsim_R y$ and $y \succsim_R a_j$. In this case, we can extend the cycle $C$ by adding $y$, which contradicts the inclusion maximality of $C$. Note for this that we can find two alternatives $a_k, a_{k+1}\in C$ such that $a_{k+1}$ is the successor of $a_k $ in $C$, $a_k \succsim_R y$, and $y \succsim_R a_{k+1}$. Otherwise, it holds for all $a_l\in C$ that $a_l \succsim_R y$ implies for its successor $a_{l+1}$ in $C$ that $a_{l+1} \succ_R y$. If we start at $a_i$ and subsequently apply this argument along the cycle $C$, we derive eventually that $a_l \succ_R y$ for all $a_l\in C$, which contradicts that $y \succsim_R a_j$. Hence, there must be such alternatives $a_k$ and $a_{k+1}$ and we can extend the cycle $C$ to $C'=(a_1, \dots a_k, y, a_{k+1}, \dots, a_q)$.
				
		As a consequence of the last case, it holds for all alternatives $x\in X\setminus C$ that either $x \succ_R C$ or $C \succ_R x$. We partition the alternatives in $X\setminus C$ with respect to these two options into the sets $X_1=\{x\in X\setminus C\colon x \succ_R C\}$ and $X_2=\{x\in X\setminus C\colon C \succ_R x\}$. If $X_1=\emptyset$, then $C \succ_R A\setminus C$, which contradicts that $X=\tc(R)$ because $C$ is a smaller dominant set than $X$. If $X_2=\emptyset$ or $X_1 \succ_R X_2$, then $X_1 \succ_R A\setminus X_1$ which again contradicts that $X=\tc(R)$ because $X_1$ is now a smaller dominant set than $X$. Thus, both $X_1$ and $X_2$ are non-empty and there is a pair of alternatives $x_1\in X_1$, $x_2\in X_2$ such that $x_2 \succsim_R x_1$. However, this means that we can extend the cycle $C$ by adding $x_1$ and $x_2$ as $a_1 \succ_R x_2$, $x_2 \succsim_R x_1$, and $x_1 \succ_R a_2$. This contradicts the inclusion maximality of $C$ and therefore, the initial assumption that $C\neq X$ was incorrect. 
		
		Finally, we prove that $\tc(R)=X$ for a set $X\subseteq A$ with $|X|\geq 2$ if there is a cycle $C=(a_1, ..., a_{|X|})$ in $\succsim_R$ with $C=X$ and $X \succ_R A\setminus X$. Note for this that $X$ is a dominant set in $\succsim_R$ if it satisfies these conditions. Since dominant sets are totally ordered by set inclusion and the top cycle is the smallest dominant set, it follows that $\tc(R)\subseteq X$. Next, assume that $X\setminus \tc(R)\neq \emptyset$, which means that $\tc(R)\succ_R X\setminus \tc(R)$ because of the definition of the top cycle. However, then there cannot be a cycle in $\succsim_R$ that connects all alternatives in $X$ because there is no path from an alternative in $X\setminus \tc(R)$ to an alternative in $\tc(R)$. This contradicts our assumptions and thus, the assumption $X\setminus \tc(R)\neq \emptyset$ was incorrect. Hence, it follows that $X=\tc(R)$. 	
	\end{proof}
	
	\Cref{lem:HC} is one of the most important insights for our subsequent proofs as the existence of the cycle provides paths between all alternatives $x,y\in \tc(R)$. This insight will also be used in the next lemma, where we investigate connected sets.
	
	\begin{lemma}\label{lem:connectors}
		Let $R$ be a preference profile and suppose that $x$ is a connector in $R$. Moreover, let $y\in A_x$ denote an alternative in the connected set of $x$. It holds that $A_y\subseteq A_x$ unless $x\succ_R A\setminus \{x,y\}$. 
	\end{lemma}
	\begin{proof}
		Consider an arbitrary preference profile $R$ and a connector $x$ in $R$. Note that the existence of a connector implies that $k=|\tc(R)|\geq 3$. Thus, let $C=(a_1, \dots, a_k)$ denote a cycle connecting the alternatives in $\tc(R)$; such a cycle exists because of \Cref{lem:HC}. Since connectors need to be in the top cycle, it follows that there is an index $i$ such that $x=a_i$. In the sequel, we assume without loss of generality that $x=a_1$ since we can decide on the starting point of the cycle. 
		
		As a first step, we show that there is an index $l\in \{2,\dots, k-1\}$ such that $A_x=\{a_{l+1}, \dots, a_k\}$. Consider for this the profile $R^{-x}=R|_{A\setminus \{x\}}$ derived from $R$ by removing $x$ from the preference profile. We determine next the top cycle in $R^{-x}$ because $A_x=\tc(R)\setminus \left(\tc(R^{-x})\cup\{x\}\right)$. First, note that $\tc(R^{-x})\subseteq \tc(R)\setminus \{x\}$ because all alternatives in $\tc(R)\setminus \{x\}\neq\emptyset$ still strictly dominate all alternatives outside of this set. This implies that $a_2$, the successor of $x=a_1$ on $C$, is in $\tc(R^{-x})$ because it can reach every other alternative $a_i\in\tc(R)\setminus \{a_1, a_2\}$ via $\succsim_{R^{-x}}$: we can simply traverse the cycle $C$ to go from $a_2$ to $a_i$. Now, if $a_2\succ_R \tc(R)\setminus \{a_1, a_2\}$, then $a_2$ is the Condorcet winner in $R^{-x}$ and thus, $A_x=\{a_3, \dots, a_k\}$, i.e., $l=3$ satisfies our condition. Otherwise, let $h_1\in \{3,\dots, k\}$ denote the largest index such that $a_{h_1}\succsim_R a_2$. It follows from the definition of the top cycle that $a_{h_1}\in\tc(R^{-x})$ because $a_2\in \tc(R^{-x})$ and thus, $\{a_2, \dots, a_{h_1}\}\subseteq \tc(R^{-x})$ because all these alternatives can reach $a_{h_1}$ by traversing the cycle $C$. It is easy to see that we can repeat this argument: if $\{a_2,\dots, a_{h_1}\}\succ_R \{a_{h_1+1},\dots, a_k\}$, then $\tc(R^{-x})=\{a_2, \dots, a_{h+1}\}$ and $A_x=\{a_{h_1+1},\dots, a_k\}$. Otherwise, we can find the largest index $h_2\in \{h_1+1, \dots, k\}$ such that $a_{h_2}$ dominates an alternative in $\{a_2, \dots, a_{h+1}\}$. Then, it follows that $a_{h_2}\in\tc(R^{-x})$ and consequently, $\{a_2,\dots, a_{h_2}\}\subseteq\tc(R^{-x})$ because we can traverse the cycle $C$ to find a path for every such alternative $a_i$ to $a_{h_2}$. By repeating this argument, we eventually arrive at an index $l$ such that $A_x=\{a_{l+1}, \dots, a_k\}$ since $A_x\neq \emptyset$.
		
		Next, consider an arbitrary alternative $y\in A_x$. Our goal is to prove that $A_y\subseteq A_x$ and thus, we consider the profile $R^{-y}=R|_{A\setminus \{y\}}$. We will show that $\tc(R^{-x})\cup\{x\}\subseteq \tc(R^{-y})$ because then $A_y=\tc(R)\setminus \left(\tc(R^{-y})\cup\{y\}\right)\subseteq \tc(R)\setminus \left(\tc(R^{-x})\cup\{x\}\right)=A_x$. For this, we employ a case distinction with respect to $y$ and first suppose that $y$ is not the direct predecessor of $x$ on $C$, i.e., $y=a_i$ for some $i<k$. Hence, let $y'={a_{i+1}}$ denote the successor of $y$ on $C$ and note that our previous insights show that $y'\in A_x$, too. It holds that $y'\in \tc(R^{-y})$ because $y'$ can reach every alternative $a_j\in \tc(R)\setminus \{y,y'\}$ in $\succsim_{R^{-y}}$ by traversing the cycle $C$. Next, note that $\tc(R^{-x})\succ_{R^{-x}} y'$ because $y'\not\in \tc(R^{-x})$ and therefore also $\tc(R^{-x})\succ_{R^{-y}} y'$. This proves that $\tc(R^{-x})\subseteq \tc(R^{-y})$ because $y'\in\tc(R^{-y})$. In particular, $x'=a_2$, the successor of $x=a_1$ on the cycle $C$, is in $\tc(R^{-y})$ because $x'\in \tc(R^{-x})$. Since $x \succsim_{R^{-y}} x'$, it follows also that $x\in \tc(R^{-y})$, which proves that $\tc(R^{-x})\cup\{x\}\subseteq \tc(R^{-y})$ and thus $A_y\subseteq A_x$. 
		
		As second case, suppose that $y=a_k$, i.e., $y$ is the direct predecessor of $x=a_1$ on $C$. In this case, we immediately derive that $x\in \tc(R^{-y})$ because we can again traverse the cycle $C$ to find a path from $x$ to every other alternative $a_i\in\tc(R)\setminus \{x,y\}$ in $\succsim_{R^{-y}}$. Next, it is important that there is an alternative $z\in A\setminus \{x,y\}$ such that $z\succsim_R x$. If there is no such alternative, then $x\succ_R A\setminus \{x,y\}$ and we have nothing to show as this is the exception stated in the lemma. Since $z\neq y$ and $z\succsim_R x$, it follows also that $z\in \tc(R)$ and $z\in \tc(R^{-y})$. Now, if $z\in A_x$, then $\tc(R^{-x})\subseteq \tc(R^{-y})$ because $\tc(R^{-x})\succ_R z$. On the other hand, if $z\in \tc(R^{-x})=\tc(R)\setminus( A_x\cup\{x\})$, we use the fact that there is a cycle $C'$ connecting the alternatives $\tc(R^{-x})$ in $\succsim_{R^{-x}}$. This cycle exists also in $\succsim_R$ and, since $y\not\in \tc(R^{-x})$, also in $\succsim_{R^{-y}}$. Hence, there is a path from every alternative $a_i\in \tc(R^{-x})$ to $z$, which proves that $\tc(R^{-x})\cup\{x\}\subseteq \tc(R^{-y})$. Thus, it follows also in this case that $A_y\subseteq A_x$, which proves the lemma.
	\end{proof}
		
	\subsection{Implications of Strategyproofness}
	\label{app:axioms}
	
	In the context of pairwise SCCs, it is inconvenient to work with the preference relations of individual voters since the main idea of these SCCs is to abstract away from profiles. However, strategyproofness requires information about a voter's preference relation to deduce which choice sets are possible before and after a manipulation. In order to mitigate this tradeoff, we analyze the implications of strategyproofness for pairwise SCCs in this section. This leads to the definition of four axioms, all of which are satisfied by every pairwise and strategyproof SCC. Also, the first three of these axioms are weakened versions of a property known as set-monotonicity \citep[see][]{Bran11c, BBH15a}
	
	In more detail, we investigate how the choice set of a strategyproof and pairwise SCC is allowed to change if a voter reinforces or weakens an alternative against some other alternatives. Formally, reinforcing an alternative $a$ against some other alternative $b$ in the preference relation of voter $i$ means that voter $i$ switches from $b \succ_i a$ to $a \succ_i' b$ and nothing else changes in voter $i$'s preference relation or in the preference relations of other voters. Conversely, weakening an alternative $a$ against some other alternative $b$ in the preference relation of voter $i$ means that voter $i$ reinforces $b$ against $a$. Note that weakening or reinforcing an alternative $a$ against another alternative $b$ requires that $a$ and $b$ are adjacent in $\succ_i$, i.e., there is no alternative $z\in A\setminus \{a,b\}$ such that $a \succ_i z \succ_i b$ or $b \succ_i z \succ_i a$, respectively. %Finally, we write $R^{i:ab}$ for the profile derived from $R$ by reinforcing $a$ against $b$ in voter $i$'s preference relation. \todo{FB: This notation can be omitted.}
		
	Depending on whether the alternatives $a$ and $b$ are chosen, strategyproofness has different consequences when reinforcing $a$ against another alternative $b$. The first case that we consider is to reinforce a chosen alternative $a$ against another alternative $b$. A natural requirement in this situation is monotonicity, which demands that a chosen alternative is still chosen after reinforcing it \citep[see, e.g.,][]{Moul88a}. Unfortunately, we cannot show that strategyproofness implies monotonicity for pairwise SCCs. For instance, assume that a voter submits $b,a,c$ and $\{a,c\}$ is chosen. Next, voter $i$ reinforces $a$ against $b$ and as result $\{b,c\}$ is chosen. In this example, Fishburn's set extension does not allow to compare $\{a,c\}$ to $\{b,c\}$ and hence, this is no violation of strategyproofness. As a consequence, we consider a weakened variant of monotonicity, which we refer to as weak monotonicity (\WMON). This axiom requires that, if a voter reinforces a chosen alternative $a$ against another alternative $b$, then $a$ is still in the choice set unless $b$ is chosen after the manipulation but not before.
	
	\begin{definition}[Weak monotonicity (\WMON)]
		An SCC $f$ satisfies weak monotonicity (\WMON) if $a\in f(R)$ implies $a\in f(R')$ or $b\in f(R')\setminus f(R)$ for all alternatives $a,b\in A$ and preference profiles $R$, $R'$ for which there is a voter $i$ such that 
${\succ'_j}={\succ_j}$ for all $j\in N\setminus \{i\}$ and 
${\succ'_i}={\succ_i}\setminus\{(b,a)\}\cup\{(a,b)\}$.
%${\succ_i|_{A\setminus\{a,b\}}}={\succ'_i|_{A\setminus\{a,b\}}}$, $b\succ_i a$, and $a\succ_i b$.
%^{i:ab}
	\end{definition}

	\WMON has multiple important consequences. Firstly, if we reinforce a chosen alternative $a$ against another chosen alternative $b$, it guarantees that $a$ remains chosen because $b\not\in f(R')\setminus f(R)$. Secondly, if we reinforce a chosen alternative $a$ against an unchosen alternative $b$, either $a\in f(R')$ and $b\not\in f(R')$, or $a\not\in f(R')$ and $b\in f(R')$. If both alternatives were chosen after this step, we could reinforce $b$ against $a$ in voter $i$'s preference relation to revert back to the original preference profile $R$, and \WMON implies that $b$ remains chosen. However, this is in conflict with the assumption that $b$ is not chosen for $R$. On the other hand, it follows directly from the definition of \WMON that it is not possible that $a,b\not\in f(R')$ if $a\in f(R)$. 
% $\{a,b\}\cap f(R')=\emptyset$	
Finally, it should be mentioned that monotonicity implies weak monotonicity because it requires that a chosen alternative $a$ remains chosen after reinforcing it, i.e., it excludes additionally the case that $a$ becomes unchosen and $b$ becomes chosen after reinforcing $a$ against $b$. 

Unfortunately, \WMON does not guarantee that weakening an unchosen alternative means that the unchosen alternative remains unchosen. We thus introduce weak set-monotonicity (\WSMON) as our second axiom, which is concerned with what happens if we weaken an unchosen alternative not against a single alternative, but against all other alternatives. 
	
	%If we have full monotonicity, we get full SMON (by combining IIC and Monotonicity)
	\begin{definition}[Weak set-monotonicity (\WSMON)]
		An SCC $f$ satisfies weak set-monotonicity (\WSMON) if $f(R)=f(R')$ for all preference profiles $R$, $R'$ for which a voter $i\in N$ and an alternative $a\not\in f(R)$ exist such that ${\succ_j}={\succ_j'}$ for all $j\in N\setminus \{i\}$, ${\succ_i\!|_{A\setminus \{a\}}}={\succ_i'\!|_{A\setminus \{a\}}}$, $a \succ_i A\setminus \{a\}$, and $A\setminus \{a\} \succ_i' a$. 
	\end{definition}

The idea of \WSMON is that moving an unchosen alternative from the first place to the last place in a voter's preference relation should not affect the outcome. This is a weaker variant of set-monotonicity, which requires that weakening an unchosen alternative against a single alternative does not affect the choice set. Unfortunately, we cannot prove this stronger variant because we cannot even prove monotonicity at this point. However, pushing the top-ranked alternative to the bottom of the preference ranking is a rather common operation in the analysis of strategyproof SCCs which is often referred to as push-down lemma \citep[see, e.g.,][]{Zwic15a}. 
		
			The third situation that we are concerned with is that a voter only reorders unchosen alternatives. Intuitively, such an operation should not change the choice set as no relevant comparisons change. This idea is formalized as independence of unchosen alternatives (\IUA).
		
		\begin{definition}[Independence of unchosen alternatives (\IUA)]\label{def:IUA}
			An SCC $f$ satisfies independence of unchosen alternatives (\IUA) if $f(R)=f(R')$ for all preference profiles $R$, $R'$ for which a voter $i\in N$ and alternatives $B\subseteq A\setminus f(R)$ exist such that ${\succ_j}={\succ_j'}$ for all voters $j\in N\setminus \{i\}$ and ${\succ_i\setminus \succ_i\!|_B}={\succ_i'\setminus \succ_i'\!|_B}$. 
			%if $a,b\not\in f(R)$ implies $a,b\not\in f(R')$ for all alternatives $a,b\in A$ and preference profiles $R, R'$ for which there is a voter $i\in N$ such that $R'=R^{i:ab}$. 
		\end{definition}
		% $\{a,b\}\cap f(R)=\emptyset$

		Independence of unchosen alternatives, also called independence of losers, is a well-known axiom \citep[see, e.g.,][]{Lasl97a,Bran11b,Bran11c}, which requires that the choice set is invariant with respect to modifications of preferences between unchosen alternatives. In particular, if a voter reinforces an unchosen alternative against another unchosen alternative, the choice set is not allowed to change. Just like \WMON and \WSMON, \IUA is implied by set-monotonicity.
		%which indicates that every strategyproof and pairwise SCC might satisfy set-monotonicity. While we cannot prove this conjecture, \Cref{thm:F-SP} shows that it is true under mild additional assumptions because every robust dominant set rule satisfies set-monotonicity. 
				
		Finally, we introduce an axiom with a different spirit than the previous ones: instead of asking whether an alternative $a$ is chosen after weakening or reinforcing it, we ask whether alternatives that are not involved in the swap are chosen or not. Intuitively, it seems plausible that if an alternative is not affected by a manipulation, its membership in the choice set should not change. However, this condition, whose spirit is similar to the localizedness property used in the characterization of strategyproof randomized social choice functions by \citet{Gibb77a}, is extremely restrictive. Here, we consider a weaker variant: if a voter changes his preference relation between some alternatives $B$ and the inclusion of the alternatives in $B$ in the choice set is unaffected by this modification, then the choice set should not change at all. This idea leads to weak localizedness (\IR) which is formalized below. 
		
		\begin{definition}[Weak localizedness (\IR)]\label{def:IR}
			An SCC $f$ satisfies weak localizedness (\IR) if $f(R)=f(R')$ for all preference profiles $R$, $R'$ for which a voter $i\in N$ and alternatives $B\subseteq A$ exist such that ${\succ_j}={\succ_j'}$ for all voters $j\in N\setminus \{i\}$, ${\succ_i\setminus \succ_i\!|_B}={\succ_i'\setminus \succ_i'\!|_B}$, and $B\cap f(R)=B\cap f(R')$. 
		\end{definition}

		To the best of our knowledge, neither \IR nor similar axioms have been studied before for social choice correspondences. Furthermore, it should be mentioned that \IR---even though it might seem weak when considered in isolation---is quite powerful when combined with other axioms. For instance, the combination of \WMON and \IR implies that swapping two chosen alternatives can only affect the choice set if the weakened alternative becomes unchosen.
	
	We now prove that strategyproof and pairwise SCCs satisfy all axioms discussed in this section. 
	
		\begin{lemma}\label{lem:axioms}
			Every strategyproof and pairwise SCC satisfies \WMON, \WSMON, \IUA, and \IR.
		\end{lemma}
		\begin{proof}
			Let $f$ denote a strategyproof and pairwise SCC. We consider each axiom listed in the lemma separately, but each proof relies on the same idea: we assume for contradiction that $f$ fails the considered axiom, which means that there are two profiles $R$ and $R'$ that differ in the preference relation of a single voter $i$ and $f(R)$ and $f(R')$ violate the conditions of the axiom. We add two new voters $i^*$ and $j^*$ with inverse preferences such that voter $i^*$ the can make the same modification as voter $i$. This leads to new preference profiles $R^1$ and $R^2$ such that $f(R^1)=f(R)$ and $f(R^2)=f(R')$ due to pairwiseness. Finally, we can choose the preference relation of voter $i^*$ such that deviating from $R^1$ to $R^2$ is a manipulation, and thus obtain a contradiction to the strategyproofness of $f$.\medskip
			
			\textbf{\WMON}: Following the idea explained above, we assume for contradiction that $f$ violates \WMON. This means that there are preference profiles $R$, $R'$, alternatives $a\in f(R)$, $b\in A\setminus \{a\}$, and a voter $i\in N$ such that ${\succ'_j}={\succ_j}$ for all $j\in N\setminus \{i\}$ and 
	${\succ'_i}={\succ_i}\setminus\{(b,a)\}\cup\{(a,b)\}$,
	%$R'=R^{i:ab}$, 
	but $a\not\in f(R')$ and $b\not\in f(R')\setminus f(R)$. Next, we let $R^1$ denote the profile derived from $R$ by adding the voters $i^*$ and $j^*$. The preference relations of these voters are shown below, where $\bar A=A\setminus \{a,b\}$ and $\bar f(R)=f(R)\setminus \{a,b\}$. Moreover, the profile $R^2$ evolves out of $R^1$ by letting voter $i^*$ swap $a$ and $b$. 
			\begin{align*}
				\succ_{i^*}^1&=\lex\left(\bar A\setminus f(R)\right), \lex\left(\bar f(R) \cap f(R')\right), b,a, \lex\left(\bar f(R) \setminus f(R')\right)\\
				\succ_{j^*}^1&=\lex\left(\bar f(R) \setminus f(R')\right)^{-1},  a,b, \lex\left(\bar f(R) \cap f(R')\right)^{-1}, \lex\left(\bar A\setminus f(R)\right)^{-1}
			\end{align*} 
		
			Since the preference relations of voter $i^*$ and $j^*$ are inverse in $R^1$, pairwiseness implies that $f(R^1)=f(R)$. Moreover, this axiom also requires that $f(R^2)=f(R')$. It then follows that voter $i^*$ can manipulate by deviating from $R^1$ to $R^2$ as he prefers all alternatives in $f(R')\setminus f(R)$ to those in $f(R)$ and all alternatives in $f(R')$ to those in $f(R)\setminus f(R')$. This can be seen by making a case distinction on whether $b\in f(R)$: if $b\not\in f(R)$, then our contradiction assumption implies that $b\not\in f(R')$, too. Hence, no alternative in $(f(R)\setminus f(R'))\cup\{b\}$ is chosen, which ensures that this is a manipulation for voter $i^*$ since $a\in f(R)\setminus f(R')$. On the other hand, if $b\in f(R)$, then $b$ is either in $f(R)\cap f(R')$ or in $f(R)\setminus f(R')$. Both cases constitute again a manipulation as $b\succ_{i^*} f(R)\setminus (f(R')\cup\{b\})$ and $(f(R)\cap f(R'))\setminus \{b\}\succ_{i^*} b$. Hence, switching from $R^1$ to $R^2$ is in all cases a manipulation for voter $i^*$, which contradicts the strategyproofness of $f$. Consequently, the initial assumption that $f$ violates \WMON was incorrect.\medskip 				
				
			\textbf{\WSMON}: As second case, assume that $f$ fails \WSMON. Thus, there are preference profiles $R$, $R'$, a voter $i\in N$, and an alternative $a\not\in f(R)$ such that $R$ and $R'$ only differ in the fact that $a\succ_i A\setminus \{a\}$ and $A\setminus \{a\}\succ_i' a$, but $f(R)\neq f(R')$. Consider the profile $R^1$ which is derived from $R$ by adding the voters $i^*$ and $j^*$ with the preferences shown below. Moreover, $R^2$ evolves out of $R^1$ by letting voter $i^*$ make $a$ into his least preferred alternative.  
			
			\begin{align*}
			\succ_{i^*}^{1}&=a, \lex\left(A\setminus (\{a\}\cup f(R))\right), \lex\left(f(R)\cap f(R')\right), \lex\left(f(R)\setminus f(R')\right)\\
			\succ_{j^*}^{1}&=\lex\left(f(R)\setminus f(R')\right)^{-1}, \lex\left(f(R)\cap f(R')\right)^{-1}, \lex\left(A\setminus (\{a\}\cup f(R))\right)^{-1}, a
			\end{align*}
		
			It is again easy to verify that $f(R^1)=f(R)$ and $f(R^2)=f(R')$ because of pairwiseness. Thus, voter $i^*$ can manipulate $f$ by switching from $R^1$ to $R^2$ because he prefers all alternatives in $A\setminus f(R)$ to all alternatives in $f(R)$ and all alternatives in $f(R)\cap f(R')$ to all alternatives in $f(R)\setminus f(R')$. This contradicts the strategyproofness of $f$ and therefore the initial assumption that $f$ violates \WSMON was incorrect.\medskip
			
			\textbf{\IUA}: Thirdly, assume that $f$ violates \IUA, which means that there are preference profiles $R$, $R'$, a voter $i\in N$, and a set of alternatives $B\subseteq A\setminus f(R)$ such that ${\succ_j}={\succ_j'}$ for all voters $j\in N\setminus \{i\}$, ${\succ_i\setminus \succ_i\!|_B}={\succ_i'\setminus \succ_i'\!|_B}$, and $f(R)\neq f(R')$. Now, consider the profile $R^1$ derived from $R$ by adding two voters $i^*$ and $j^*$. The preference relations of these two voters are shown below, where $\bar A=A\setminus B$, $\succ_i\!|_B$ indicates that the alternatives in $B$ are ordered as in $\succ_i$, and $\succ_i^{-1}\!|_B$ that the alternatives in $B$ are ordered exactly inverse to $\succ_i$. Moreover, let $R^2$ denote the profile derived from $R^1$ by letting voter $i^*$ order the alternatives in $B$ as voter $i$ does in $R'$. 
				\begin{align*}
				\succ_{i^*}^{1}&=\succ_i\!|_B, \lex\left(\bar A \setminus f(R)\right), \lex\left(f(R)\cap f(R')\right), \lex\left(f(R)\setminus f(R')\right)\\
				\succ_{j^*}^{1}&=\lex\left(f(R)\setminus f(R')\right)^{-1}, \lex\left(f(R)\cap f(R')\right)^{-1}, \lex\left(\bar A\setminus f(R)\right)^{-1}, \succ_i\!|_B^{-1}
				\end{align*}
				
				Just as before, we infer from pairwiseness that $f(R^1)=f(R)$ and $f(R^2)=f(R')$. However, this means that voter $i^*$ can manipulate by deviating from $R^1$ to $R^2$: by construction, he prefers all alternatives in $A\setminus f(R)$ to all alternatives in $f(R)$, and all alternatives in $f(R)\cap f(R')$ to all alternatives in $f(R)\setminus f(R')$. Since $f(R)\neq f(R')$, this is in conflict with strategyproofness.\medskip
				
			\textbf{\IR}: Finally, suppose for contradiction that $f$ violates \IR. Hence, there are two preference profiles $R$, $R'$, a non-empty set of alternatives $B\subseteq A$, and a voter $i\in N$ such that ${\succ_j}={\succ_j'}$ for all $j\in N\setminus \{i\}$, ${\succ_i\setminus \succ_i\!|_B}={\succ_i'\setminus \succ_i'\!|_B}$, $f(R)\cap B =f(R')\cap B$, but $f(R)\neq f(R')$. Once again, we derive a new profile $R^1$ from $R$ by adding two voters $i^*$ and $j^*$. The preferences of these voters are shown below, where $\bar A=A\setminus B$ and $\bar f(R)=f(R)\setminus B$. Moreover, let $R^2$ denote the profile derived from $R^1$ by letting voter $i^*$ arrange the alternatives in $B$ according to $\succ_i'$. 
			
			\begin{align*}
			\succ_{i^*}^{1}&=\lex\left(\bar A\setminus f(R)\right), \succ_i\!|_B, \lex\left(\bar f(R)\cap f(R')\right), \lex\left(\bar f(R)\setminus f(R')\right)\\
			\succ_{j^*}^{1}&=\lex\left(\bar f(R)\setminus f(R')\right)^{-1}, \lex\left(\bar f(R)\cap f(R')\right)^{-1},  \succ_i^{-1}\!|_B, \lex\left(\bar A\setminus f(R)\right)^{-1}
			\end{align*}
		
			Also in this case, pairwiseness shows that $f(R^1)=f(R)$ and $f(R^2)=f(R')$. Thus, voter $i^*$ can manipulate by switching from $R^1$ to $R^2$ because $f(R')\succ_{i^*}^1 f(R)\setminus f(R')$ and $f(R')\setminus f(R)\succ_{i^*}^1 f(R)$. For both claims, it is important that $B$ is disjoint to both $f(R)\setminus f(R')$ and $f(R')\setminus f(R)$ since an alternative $x\in B$ is in $f(R)$ iff it is in $f(R')$. Hence, the first claim follows directly as the alternatives in $f(R)\setminus f(R')$ are the least preferred ones of voter $i^*$, and the second claim follows since $f(R')\setminus f(R)\subseteq A\setminus (B\cup f(R))$. Thus, deviating from $R^1$ to $R^2$ is a manipulation for voter $i^*$, which contradicts the strategyproofness of $f$.
		\end{proof}

\subsection{Consequences of Strong Condorcet-Consistency}
\label{app:condorcet}

In this section, we prove a variant of \Cref{thm:F-SP} which relies on \emph{strong Condorcet-consistency} instead of neutrality and non-imposition. This axiom requires of an SCC $f$ that $f(R)=\{x\}$ iff $x$ is the Condorcet winner in $R$. Less formally, strongly Condorcet-consistent SCCs have to elect the Condorcet winner whenever there is one, and cannot elect a single alternative in the absence of a Condorcet winner. The main result of this section states that robust dominant set rules are the only SCCs that satisfy pairwiseness, homogeneity, strong Condorcet-consistency, and strategyproofness. As we will see in \Cref{app:mainresults}, the combination of pairwiseness, homogeneity, strategyproofness, neutrality, and non-imposition implies strong Condorcet-consistency, which means that this auxiliary claim is actually more general than \Cref{thm:F-SP}. Moreover, the axioms of \Cref{thm:TC} imply strong Condorcet-consistency and we can therefore use the results of this section also to characterize the top cycle.

For proving the results of this section, we rely on the lemmas of the previous subsections. In particular, we often say that we reinforce an alternative $x$ against another alternative $y$ without specifying which voter reinforces $x$ against $y$. This is possible because we can always add two voters with inverse preferences such that one of them can perform the required manipulation. Adding these two voters does not affect the choice set because of pairwiseness and the consequences of the deviation will be specified by the axioms of \Cref{app:axioms}. Hence, we can abstract away from the exact preference profiles and focus on the majority margins. For the readers' convenience, we repeat the four axioms of the last section because they form the basis of the following proofs. Let $f$ be a pairwise SCC, $a,b\in A$, and $R,R'\in \mathcal{R}^*(A)$.
\begin{itemize}
	\item \WMON: If $R'$ is derived from $R$ by reinforcing $a$ against $b$ and $a\in f(R)$, then $a\in f(R)$ or $b\in f(R')\setminus f(R)$.
	\item \WSMON: If $R'$ is derived from $R$ by weakening $a$ against all other alternatives $x\in A\setminus \{a\}$ and $a\not\in f(R)$, then $f(R)=f(R')$.
	\item \IUA: If $R'$ is derived from $R$ by reordering some alternatives in $A\setminus f(R)$, then $f(R)=f(R')$.
	\item \IR: If $R'$ is derived from $R$ by reordering the alternatives in $B\subseteq A$ such that $f(R)\cap B=f(R')\cap B$, then $f(R)=f(R')$.
\end{itemize}

As shown in \Cref{app:axioms}, every strategyproof and pairwise SCC satisfies these axioms. We now use these properties to show our first key insight, namely that all such SCCs that satisfy strong Condorcet-consistency also satisfy a new property called Condorcet-stability (\CWCS). This axiom requires that there should be no alternative---within or outside of the choice set---that strictly dominates all other alternatives in the choice set. 
Note that \CWCS implies that an alternative can only be a single winner if it weakly dominates every other alternative.\footnote{This axiom is quite useful for characterizing SCCs that are strategyproof in profiles that admit a Condorcet winner: a majoritarian and non-imposing SCC is strategyproof in such profiles iff it satisfies \CWCS.}
% COS implies "reverse weak Condorcet-consistency"
%\todo{FB: Wouldn't this definition be nicer: There is no $x\in A$ such that $x\succ_R f(R)\setminus \{x\}$ whenever $f(R)\setminus \{x\}$ is non-empty.}

\begin{definition}[Condorcet-stability (\CWCS)] % could also be called "absorbing" or "choice set maximality"
	An SCC $f$ satisfies Condorcet-stability (\CWCS) if for every preference profile $R$, there is no alternative
	$x\in A$ such that $x\succ_R f(R)\setminus \{x\}$ whenever $f(R)\setminus \{x\}$ is non-empty.
	%and every alternative $x\in A$, either $f(R)=\{x\}$ or there is an alternative $y\in f(R)\setminus \{x\}$ such that $y \succsim_R x$.
\end{definition}
This condition is equivalent to requiring that every alternative is weakly dominated by another chosen alternative unless it is the unique winner, i.e., for every alternative $x\in A$ with $f(R)\ne\{x\}$ there is an alternative $y\in f(R)\setminus \{x\}$ such that $y \succsim_R x$. It is thus closely connected to the notion of external stability, which requires that for every alternative $x\in A\setminus f(R)$, there is an alternative $y\in f(R)$ such that $y \succsim_R x$ \citep[see, e.g.,][]{MGF90a, Dugg11a}. Indeed, Condorcet-stability is a stronger requirement than external stability as it also includes a notion of internal stability. 
%\todo{FB: could also mention work on dominating sets in graph theory. Erdös, ...}

As we show next, the conjunction of our axioms implies \CWCS. 

\begin{lemma}\label{lem:CWCS}
	Every pairwise SCC that is strategyproof and strongly Condorcet-consistent satisfies \CWCS. 
\end{lemma}
\begin{proof}
	Let $f$ denote a pairwise SCC that satisfies strategyproofness and strong Condorcet-consistency. First, recall that strong Condorcet-consistency requires that $f(R)=\{x\}$ iff $x$ is the Condorcet winner in $R$. Hence, strong Condorcet-consistency implies \CWCS for all profiles $R$ with a Condorcet winner $x$ because $f(R)=\{x\}$ entails $x\succ_R A\setminus \{x\}$. Next, we focus on profiles without a Condorcet winner and assume for contradiction that $f$ fails \CWCS for such a profile. More formally, this assumption means that there is a profile $R$ without a Condorcet winner and an alternative $a\in A$ such that $a\succ_R f(R)\setminus \{a\}$. It follows from the absence of a Condorcet winner that there is at least one alternative $x$ with $x \succsim_R a$, but no such alternative is chosen. Hence, we can repeatedly use \WSMON to weaken the alternatives $x$ with $x \succsim_R a$ against all other alternatives until we arrive at a profile $R'$ such that $a \succ_{R'} A\setminus\{a\}$. Now, \WSMON implies that the choice set does not change during these steps and thus it holds that $f(R')=f(R)$. This means in particular that $|f(R')|\geq 2$ because strong Condorcet-consistency requires that $|f(R)|\geq 2$. However, $a$ is the Condorcet winner in $R'$ and consequently, strong Condorcet-consistency also shows that $f(R')=\{a\}$. These two observations contradict each other and consequently the assumption that $f$ violates \CWCS was incorrect.
\end{proof}

\CWCS plays an important role in our proofs because we can use it to force an alternative into the choice set. In particular, the combination of strong Condorcet-consistency and \CWCS have rather strong consequences: the first axiom states that we choose a single winner iff it is the Condorcet winner and the second one requires therefore that every alternative is weakly dominated by a chosen alternative if there is no Condorcet winner. We now use this interaction to prove our first lemma for pairwise SCCs that satisfy homogeneity, strategyproofness, and strong Condorcet-consistency. %We prove this claim in three steps: first, we show that if an SCC $f$ satisfies all given axioms and is a dominant set rule, it is majoritarian. As a a second step, we prove that this means that $f$ is even a robust dominant set rule. Based on this insight, we prove as a last step that $f$ needs to be a robust dominant set rule if it always contains the top cycle.

\begin{lemma}\label{lem:TCnosubset}
	Let $f$ denote a pairwise SCC that satisfies strong Condorcet-consistency, homogeneity, and strategyproofness. If $\tc(R)\subseteq f(R)$ for all profiles $R$, then $f$ is a robust dominant set rule.
\end{lemma}
\begin{proof}
	Let $f$ denote a pairwise, homogeneous, strategyproof, and strongly Condorcet-consistent SCC. We prove this lemma in three steps: first, we show that if $\tc(R)\subseteq f(R)$ for all profiles $R$, then $f$ is a dominant set rule. Next, we prove that, if $f$ is a dominant set rule, our assumptions require it to be majoritarian. As last point, we show that $f$ is even robust if it is a majoritarian dominant set rule. Combining all three steps thus shows the lemma.\bigskip
	
	\textbf{Step 1: If $\tc(R)\subseteq f(R)$ for all profiles $R$, $f$ is a dominant set rule.}

		We prove this claim by contradiction and thus assume that $f$ always chooses a superset of $\tc$ but is no dominant set rule. This means that there is a profile $R$ such that $f(R)$ is no dominant set in $\succsim_R$, which implies that $f(R)\neq \tc(R)$. Moreover, there is no Condorcet winner in $R$ because, otherwise, strong Condorcet-consistency would require that $f$ chooses this alternative as unique winner, which would contradict that $f(R)$ is no dominant set. We infer from this observation that $|\tc(R)|>1$ because $\tc$ is strongly Condorcet-consistent. Next, note that there are alternatives $a\in f(R)$ and $b\in A\setminus f(R)$ such that $b \succsim_R a$ because $f(R)$ is no dominant set. Even more, $a\not\in \tc(R)$; otherwise, $b$ would also be in $\tc(R)$ because $b\succsim_R a$, which would imply that $\tc(R)\not\subseteq f(R)$ since $b\not\in f(R)$. However, this contradicts our assumptions.
	
		Next, let $x$ denote an alternative in $\tc(R)\subseteq f(R)$, which implies that $x\succ_R a$. We repeatedly reinforce $a$ against $x$ until we arrive at a profile $R'$ such that $a\succsim_{R'} x$. Moreover, let $\bar R'$ denote the last profile constructed before $R'$, i.e., $x\succ_{\bar R'} a$ and a single voter needs to reinforce $a$ against $x$ to derive $R'$. First, we show that $f(\bar R')=f(R)$ and for this consider two consecutive profiles $\hat R$ and $\hat R'$ in the sequence that leads from $R$ to $\bar R'$. This means that $\hat R'$ is derived from $\hat R$ by reinforcing $a$ against $x$. Moreover, it holds that ${\succsim_R}={\succsim_{\hat R}}={\succsim_{\hat R'}}$ because the majority relation between $a$ and $x$ has not changed yet. This implies that $\tc(R)=\tc(\hat R)=\tc(\hat R')$ and therefore $x\in \tc(\hat R)\subseteq f(\hat R)$ and $x\in \tc(\hat R')\subseteq f(\hat R')$. Now, if $a\in f(\hat R)$, then \WMON implies that $a\in f(\hat R')$ since $a$ is reinforced against $x$ to derive $\hat R'$. Finally, \IR implies then that $f(\hat R)=f(\hat R')$ because $a$ and $x$ are both chosen before and after the manipulation. Since we start this process at the profile $R$ with $\{a,x\}\subseteq f(R)$, it follows from a repeated application of this argument that $f(\bar R')=f(R)$. 
		
		%We derive the profile $R^1$ from $R$ by reinforcing $a$ against $x$. \WMON implies that $a\in f(R^1)$ because both $a$ and $x$ have been chosen before this step. Moreover, if $x\in f(R^1)$, then $f(R^1)=f(R)$ because of \IR. Hence, we can repeat this step until one of two things happens: either we derive a profile $R^2$\todo[noinline,size=\tiny]{FB: I think it would be better to still call this profile $R^1$, but this is not very important.} such that $a\in f(R^2)$, $x\not\in f(R^2)$, and $x \succ_{R^2} a$ (which happens if $x$ drops out of the choice set before the majority relation between $a$ and $x$ changes), or we have that $a \succsim_{R^2} x$ (which happens if we manage to flip the majority relation between $a$ and $x$). However, the case that $x\succ_{R^2} a$ and $x\not\in f(R^2)$ contradicts the assumption that $f$ always chooses a superset of $\tc$. The reason for this is that ${\succsim_{R}}={\succsim_{R^2}}$ because the majority relation between $a$ and $x$ did not change and thus, $\tc(R^2)=\tc(R)$. In particular, this means that $x\in \tc(R^2)$, which proves that $\tc(R^2)\not\subseteq f(R^2)$ since $x\not\in f(R^2)$.
		
	Finally, we show that $\{b,x\}\not\subseteq f(R')$ but $\{b,x\}\subseteq \tc(R')$. This contradicts the assumption that $\tc$ is always contained in $f$ and thus proves that our initial assumption that $f$ is no dominant set rule was incorrect. First, we show that $\{b,x\}\not\subseteq f(R')$. Observe for this that $b\not\in f(\bar R')=f(R)$ and that $\{a,x\}\subseteq f(\bar R')=f(R)$. Since $R'$ is derived from $\bar R'$ by reinforcing $a$ against $x$, \WMON{} implies that $a\in f(R')$. Now, if $x\in f(R')$, it follows from \IR that the choice set is not allowed to change, which implies that $b\not\in f(R')$. This means that $b\in f(R')$ is only possible if $x\not\in f(R')$, which shows that $\{b,x\}\not\subseteq f(R')$. Next, we prove the second claim that $\{b,x\}\subseteq \tc(R')$. For proving this, it is important that $|\tc(R)|>1$ and that $\succsim_{R'}$ differs from $\succsim_{R}$ only in the fact that $a\succsim_{R'} x$ and $x\succ_R a$. The first point means that there is a cycle $C$ in $\succsim_R$ connecting all alternatives in $\tc(R)$ because of \Cref{lem:HC} and the second one that this cycle also exists in $\succsim_{R'}$. Hence, there is a path from $x$ to every alternative $y\in \tc(R)\setminus\{x\}$. Moreover, there is a path from $x$ to every alternative $z\in A\setminus \tc(R)$ because we can go from $x$ to another alternative $y\in \tc(R)\setminus \{x\}$ using the cycle $C$ and from $y$ to $z$ because $y\succ_{R'} z$. This means that $x\in \tc(R')$ as it reaches every other alternative on some path. Furthermore, $a\in \tc(R')$ because $a \succsim_{R'} x$, and $b\in \tc(R')$ because $b \succsim_{R'} a$. This shows that $\{b,x\}\subseteq \tc(R')$, even though $\{b,x\}\not\subseteq f(R')$. Hence, $\tc(R')\not\subseteq f(R')$, which contradicts the assumption that $\tc(R)\subseteq f(R)$ for all profiles $R$. This proves that the assumption that $f$ is no dominant set rule was incorrect.\bigskip
		
	\textbf{Step 2: If $f$ is a dominant set rule, it is majoritarian.}
	
	Our goal in this step is to show that if $f$ is a dominant set rule, it is majoritarian. Thus, assume for contradiction that $f$ is a dominant set rule but violates majoritarianess. The latter point means that there are two preference profiles $R$ and $R'$ such that ${\succsim_R}={\succsim_{R'}}$ but $f(R)\neq f(R')$. We assume that both $R$ and $R'$ are defined by an even number of voters. This is without loss of generality as we can just duplicate the profiles if required. The majority relations do not change by this step since the majority margins are only doubled, and the choice sets do not change because of homogeneity. Thus, we can also work with these larger profiles instead. Next, observe that $f(R)\neq f(R')$ implies that $g_{R}\neq g_{R'}$ because $f$ is pairwise. Moreover, $f(R)$ and $f(R')$ are both dominant sets in $\succsim_R$ because $f$ is a dominant set rule. Since dominant sets are ordered by set inclusion, it follows that $f(R)\subsetneq f(R')$ or $f(R')\subsetneq f(R)$. We assume without loss of generality that $f(R)$ is a subset of $f(R')$; otherwise, we can just exchange the role of $R$ and $R'$ in the subsequent arguments. Our goal is to transform $R$ into a profile $R^*$ such that $g_{R^*}=g_{R'}$ and $f(R^*)\subseteq f(R)\subsetneq f(R')$. This is in conflict with the pairwiseness of $f$ and shows therefore that the assumption $f(R)\neq f(R')$ was incorrect. 
	
	We use the largest majority margin $c=\max_{x,y\in A} g_{R}(x,y)$ in $R$ for the derivation of $R^*$. In more detail, we first construct a profile $R^1$ such that $g_{R^1}(x,y)=c$ for all alternatives $x,y\in A$ with $x\succ_R y$. For this, we repeatedly use the following steps: first, identify a pair of alternatives $x,y\in A$ such that $x\succ_R y$ but the majority margin between $x$ and $y$ is not $c$ yet. Then, reinforce $x$ against $y$. By repeating these steps, we eventually arrive at a profile $R^1$ which satisfies $g_{R^1}(x,y)=c$ for all $x,y\in A$ with $x\succ_R y$. We show next that $f(R^1)\subseteq f(R)$ by a case distinction with respect to $x$ and $y$. For this, consider a single step of our process, and let $\bar R$ denote the profile before reinforcing $x$ against $y$ and $\bar R'$ denote the profile after reinforcing $x$ against $y$. If $x\not\in f(\bar R)$ and $y\not\in f(\bar R)$, it follows from \IUA that $f(\bar R)=f(\bar R')$. If $x\in f(\bar R)$ and $y\not\in f(\bar R)$, it follows from \WMON that $x\in f(\bar R)$ and $y\not\in f(\bar R')$ because we have $x\succ_R y$ and therefore also $x\succ_{\bar R'} y$. Hence, if $y\in f(\bar R')$, then $x\in f(\bar R')$ as $f$ chooses a dominant set. However, this is conflict with \WMON: if we revert the swap, this axiom implies that $y\in f(\bar R)$, which contradicts our assumptions. Thus, $x\in f(\bar  R')$, $y\not\in f(\bar R')$ and \IR implies that $f(\bar R)=f(\bar R')$. As third point, note that $x\not\in f(\bar R)$ and $y\in f(\bar R)$ is impossible because we assume that $x\succ_R y$. Hence, this case contradicts that $f(\bar R)$ is a dominant set. The last case is that $x\in f(\bar R)$ and $y\in f(\bar R)$. In this case, it follows from \WMON that $x\in f(\bar R')$. If now also $y\in f(\bar R')$, \IR implies that $f(\bar R)=f(\bar R')$. On the other hand, if $y\not\in f(\bar R')$, then $f(\bar R')\subsetneq f(\bar R)$. Otherwise, an alternative $z\in A\setminus f(\bar R)$ is in $f(\bar R')$, which is in conflict with the fact that $f(\bar R')$ is a dominant set since $y\succ_{\bar R'} z$ because $y\in f(\bar R)$ and $z\not\in f(\bar R)$. Hence, we derive in all possible cases that $f(\bar R')\subseteq f(\bar R)$. By repeatedly applying this argument, it follows that $f(R^1)\subseteq f(R)$. 
	
	Next, note that ${\succsim_{R^1}}={\succsim_{R}}$ because we only increase the majority margins between alternatives $x,y\in A$ with $x\succ_{R} y$. Furthermore, there are only two possible majority margins in $R^1$: if $x\sim_{R^1} y$, then $g_{R^1}(x,y)=0$ and if $x\succ_{R^1} y$, then $g_{R^1}(x,y)=c$. This means that we can use homogeneity to derive a profile $R^2$ with smaller majority margins: we set $g_{R^2}(x,y)=2$ for all alternatives $x,y\in A$ with $x\succ_{R^1} y$ and $g_{R^2}(x,y)=0$ for all alternatives $x,y\in A$ with $x\sim_{R^1} y$. Such a preference profile $R^2$ exists because we can use McGarvey's construction to build a preference profile for all majority margins that have the same parity \citep{McGa53a}. It follows from homogeneity and pairwiseness that $f(R^2)=f(R^1)$ because we can just multiply $R^2$ such that all majority margins are equal to those in $R^1$. Note here that the assumption that $R$ is defined by an even number of voters is important because it ensures that $c$ is a multiple of $2$. As last point, observe that ${\succsim_{R^2}}={\succsim_{R}}={\succsim_{R'}}$ because we did not change the sign of a majority margin. Moreover, $g_{R^2}(x,y)\leq g_{R'}(x,y)$ for all $x,y\in A$ because $R'$ is defined by an even number of voters. Hence, if $x\succ_{R'} y$, then $g_{R'}(x,y)\geq 2=g_{R^2}(x,y)$, and if $x\sim_{R'} y$, then $g_{R^2}(x,y)=g_{R'}(x,y)=0$. 
	
	As last step, we derive a preference profile $R^3$ with $g_{R^3}=g_{R'}$ from $R^2$ by applying the same process as in the construction of $R^1$: we repeatedly identify a pair of alternatives $x,y\in A$ such that $x\succ_{R} y$ and the current majority margin between $x$ and $y$ is less than the one in $R'$, and reinforce $x$ against $y$. Clearly, this process results in a profile $R^3$ with $g_{R^3}=g_{R'}$ and the same arguments as for $R^1$ show that $f(R^3)\subseteq f(R^2)$. We derive therefore from pairwiseness that $f(R')=f(R^3)\subseteq f(R^2)\subseteq f(R^1)\subseteq f(R)\subsetneq f(R')$, which is a contradiction because the last subset relation is by assumption strict. Hence, our initial assumption was incorrect and $f$ is indeed majoritarian.\bigskip
	
	\textbf{Step 3: If $f$ is majoritarian dominant set rule, it is robust.}
	
	As the last step, we show that $f$ is robust if it is a majoritarian dominant set rule. Thus, assume for contradiction that $f$ is a majoritarian dominant set rule that fails robustness. This means that there are two preference profiles $R$ and $R'$ such that $f(R)$ is dominant in $\succsim_{R'}$, but $f(R')\not\subseteq f(R)$. As a consequence, there is an alternative $y\in f(R')\setminus f(R)$. Moreover, since $f(R)$ is dominant in $\succsim_{R'}$, it follows that $f(R) \succ_{R'} y$ and hence, $f(R)\subsetneq f(R')$ as $f$ is a dominant set rule. We derive a contradiction to this assumption by constructing two preference profiles $R^2$ and $R^3$ such that $f(R^2)=f(R)$, $f(R^3)=f(R')$, and ${\succsim_{R^2}}= {\succsim_{R^3}}$. These observations are conflicting since ${\succsim_{R^2}}= {\succsim_{R^3}}$ requires that $f(R^2)=f(R^3)$ because of majoritarianess, but $f(R)\neq f(R')$. Note that we assume in the sequel that both $R$ and $R'$ are defined by an even number of voters as we want to introduce majority ties. This is without loss of generality as $f$ is homogeneous. 
	%For odd $n$:
	%If |f(R)|=1, we can reorder all other alternatives using IUA
	%If |f(R)|=3, we can still reorder all other alternatives using IUA. The rest follows from neutrality (in case of differently spinning alternatives)
	%If |f(R)|\geq 4$, we can revert nodes on the defining cycle: if a P_M b P_M c P_M d, we can ensure that a P_M c and b P_M d without breaking the cylce. Then, reverting b P_M c to c P_M d leads to a new cylce with a P_M c P_M b P_M d. With this process, we can order the alternatives on the cylce arbitrarily and thus, create the same profiles.  
	
	First, we explain how to derive $R^2$ from $R$. As a first step, we reorder the alternatives in $A\setminus f(R)$ to derive a profile $R^1$ with ${\succsim_{R^1}\!|_{A\setminus f(R)}}={\succsim_{R'}\!|_{A\setminus f(R)}}$. As a consequence of \IUA, it follows that $f(R^1)=f(R)$ since this step does not affect chosen alternatives. Next, let $D_{i^*}$ denote the dominant set in $R^1$ that is currently chosen, i.e., $f(R^1)=D_{i^*}$. 
	We derive the profile $R^2$ by repeating the following procedure with $R^1$ as starting profile: in the current preference profile $\bar R$, we choose a pair of alternatives $x,y\in D_{i^*}$ such that $y\succ_{\bar R} x$ and reinforce $x$ against $y$ until we arrive at a profile $\bar R'$ with $x\sim_{\bar R'} y$. It follows from \WMON and majoritarianess that $x\in f(\bar R')$ if $x,y\in f(\bar R)$. Moreover, as $f$ is a dominant set rule, $x\in f(\bar R')$ implies $y\in f(\bar R')$ because $y\succsim_{\bar R'} x$. Hence, we infer from \IR that $f(\bar R)=f(\bar R')$ if $x,y\in f(\bar R)$. Since $f(R^1)=D_i^*$, we can thus repeat this process until we arrive at a profile profile $R^2$ with $x\sim_{R^2} y$ for all $x,y\in D_{i^*}$, and it follows from the previous argument that $f(R^1)=f(R^2)$. Moreover, the majority relation of $R^2$ is completely specified: we have $f(R) \succ_{R^2} A\setminus f(R)$, $x\sim_{R^2} y$ for all $x,y\in f(R)$, and ${\succsim_{R^2}}|_{A\setminus f(R)}={\succsim_{R'}}|_{A\setminus f(R)}$.
	
	Finally, we apply the same construction as for $R^2$ to derive the profile $R^3$ from $R'$. In more detail, observe that, by assumption, $D_i^*=f(R)$ is dominant in $\succsim_{R'}$ and $D_{i^*}\subsetneq f(R')$. Hence, we can use the same construction as for $R^2$ to introduce majority ties between all alternatives in $D_{i^*}$ in $\succsim_{R'}$. The same reasoning as in the previous paragraph shows that this step does not change the choice set, and it hence holds for the resulting profile $R^3$ that $f(R')=f(R^3)$. In particular, $R^3$ has now the same majority relation as $R^2$, which is in conflict with majoritarianess since ${\succsim_{R^2}}={\succsim_{R^3}}$ but $f(R^2)=f(R)\neq f(R')=f(R^3)$. This is a contradiction to our assumptions and $f$ is therefore robust if it is a majoritarian dominant set rule.
\end{proof}

\Cref{lem:TCnosubset} presents a simple criterion for deciding when a strategyproof, homogeneous, pairwise, and strongly Condorcet-consistent SCC $f$ is a robust dominant set rule, namely when $\tc(R)\subseteq f(R)$ for all profiles $R$. Our next goal is to prove that every such SCC meets this condition without further assumptions. Hence, suppose for contradiction that this is not the case, i.e., there are an SCC $f$ that satisfies all our axioms and a profile $R$ such that $\tc(R)\not\subseteq f(R)$. If such a profile $R$ exists, we may as well focus on the profile $R^f$ that minimizes the size of the top cycle among all profiles $R$ with $\tc(R)\not\subseteq f(R)$. Furthermore, for every SCC $f$, we define $k_f\in \{1,\dots, m+1\}$ as the maximal value such that $\tc(R)\subseteq f(R)$ for all preference profiles $R$ with $|\tc(R)|<k_f$. Note that $k_f=|\tc(R^f)|$ if $\tc(R)$ is not always a subset of $f(R)$, and $k_f=m+1$ otherwise.

As the next step, we show that $k_f\geq 4$ for all pairwise SCCs $f$ that satisfy strategyproofness, homogeneity, and strong Condorcet-consistency. In general, this means that such SCCs can only fail to choose a superset of the top cycle if $\tc(R)$ is sufficiently large. For the special case where $m\leq 3$, \Cref{lem:TCnosubset,lem:3alt} already imply that $f$ needs to be a robust dominant set rule because the size of the top cycle is bounded by the number of alternatives.

\begin{lemma}\label{lem:3alt}
	Let $f$ denote a pairwise SCC that satisfies strong Condorcet-consistency, strategyproofness, and homogeneity. Then, $k_f\geq 4$. 
\end{lemma}
\begin{proof}
	Let $f$ denote a pairwise SCC that satisfies strong Condorcet-consistency, homogeneity, and strategyproofness. Furthermore, suppose for contradiction that there is a profile $R^*$ such that $k=|\tc(R^*)|\leq 3$ but $\tc(R^*)\not\subseteq f(R^*)$. We proceed with a case distinction with respect to $|\tc(R^*)|$ to derive a contradiction for the three possible cases.\bigskip
	
	\textbf{Case 1: $|\tc(R^*)|=1$}
	
	If $|\tc(R^*)|=1$, there has to be a Condorcet winner in $R^*$ since $\tc$ is strongly Condorcet-consistent. Consequently, the strong Condorcet-consistency of $f$ requires that $f(R^*)=\tc(R^*)$, which contradicts the assumption that $\tc(R^*)\not\subseteq f(R^*)$.\bigskip
	
	\textbf{Case 2: $|\tc(R^*)|=2$}
	
	The top cycle only elects two alternatives $x,y\in A$ if $x\sim_{R^*} y$ and $\{x,y\}\succ_{R^*}A\setminus \{x,y\}$. Hence, strong Condorcet-consistency requires that $f$ chooses at least two alternatives. In turn, \CWCS implies then that both $x$ and $y$ are chosen because $x$ is the only alternative that dominates $y$ and $y$ is the only alternative that dominates $x$. This shows that $\tc(R^*)\subseteq f(R^*)$ and we again have a contradiction.\bigskip
	
	\textbf{Case 3: $|\tc(R^*)|=3$}
	
	Next, assume there are three alternatives $a$, $b$, and $c$ such that $\tc(R^*)=\{a,b,c\}$, but $\{a,b,c\}\not\subseteq f(R^*)$. 
	First, note that strong Condorcet-consistency requires that $|f(R^*)|\geq 2$ because there is no Condorcet winner in $R^*$; otherwise, it would hold that $|\tc(R^*)|=1$ as $\tc$ uniquely chooses the Condorcet winner whenever it exists. 
	Next, observe that, according to \Cref{lem:HC}, there has to be a cycle $C$ that connects $a$, $b$, $c$ since $\tc(R^*)=\{a,b,c\}$. 
	Without loss of generality, suppose that $C=(a,b,c)$, i.e., $a\succsim_{R^*} b$, $b\succsim_{R^*}c$, $c\succsim_{R^*}a$. Moreover, we also suppose without loss of generality that $a\not\in f(R^*)$ because $\tc(R^*)\not\subseteq f(R^*)$.  \CWCS then implies that $b\sim_{R^*} c$ and $\{b,c\}\subseteq f(R^*)$ because, otherwise, no chosen alternative dominates $b$ or $c$. 
	Next, consider the profile $R^1$ derived from $R^*$ by reinforcing $b$ against $c$, i.e., we have $b\succ_{R^1} c$ instead of $b\sim_{R^*}c$. 
	First, note that there is no Condorcet winner in $R^1$ and thus, strong Condorcet-consistency requires us to choose at least two alternatives. 
	Hence, \CWCS implies that $a\in f(R^1)$ because it is the only alternative that dominates $b$ in $R^1$. 
	Moreover, \WMON shows that $b\in f(R^1)$ since we swap two chosen alternatives to derive $R^1$ from $R^*$. Finally, the contraposition of \IR implies that $c\not\in f(R^1)$ since $a\in f(R^1)\setminus f(R^*)$. 
	These observations entail that $a\sim_{R^1} b$ because, otherwise, no chosen alternative dominates $a$ which violates \CWCS. 
	Thus, we can repeat the previous steps by reinforcing $a$ against $b$, which results in a profile $R^2$ such that $a\succ_{R^2} b$, $b\succ_{R^2} c$, $\{a,c\}\subseteq f(R^2)$, and $b\not\in f(R^2)$. 
	Hence, \CWCS implies once again that $c\sim_{R^2} a$ and we can again break this majority tie to derive the profile $R^3$. 
	In more detail, all edges of $C$ are now strict, and $a\not\in f(R^3)$. This contradicts \CWCS because $b\succ_{R^3} A\setminus \{a,b\}$, i.e., no chosen alternative dominates $b$. 
	Hence, the assumption that $\tc(R^*)\not\subseteq f(R^*)$ was incorrect, which proves the lemma.
	\end{proof}
	
	Due to \Cref{lem:3alt}, it follows that every pairwise SCC that satisfies strategyproofness, homogeneity, and strong Condorcet-consistency can only fail to choose a superset of the top cycle if $|\tc(R)|\geq 4$. Hence, we subsequently investigate profiles with a top cycle that contains at least $4$ alternatives. For this, we first need to discuss some auxiliary lemmas and start by showing that a pairwise SCC $f$ that satisfies homogeneity, strategyproofness, and strong Condorcet-consistency must choose almost all alternatives of the top cycle for profiles $R$ with $|\tc(R)|=k_f\geq 4$. 
	
	\begin{lemma}\label{lem:size}
		Let $f$ denote a pairwise SCC that satisfies strong Condorcet-consistency, homogeneity, strategyproofness, and $4\leq k_f\leq m$. It holds that $|\tc(R)\cap f(R)|\geq k_f-1$ for all profiles $R$ with $|\tc(R)|=k_f$.
	\end{lemma}
	\begin{proof}
		Let $f$ denote a pairwise SCC that satisfies all axioms of the lemma and assume that $k_f\in \{4,\dots, m\}$. Furthermore, suppose for contradiction that there is a profile $R$ such that $|\tc(R)|=k_f$ and $|f(R)\cap\tc(R)|\leq k_f-2$. This means that at least two alternatives of the top cycle are not chosen, i.e., the set $X=\tc(R)\setminus f(R)$ contains at least two alternatives. Next, let $x\in X$ denote one of these alternatives. We proceed with a case distinction with respect to the connected set $A_x$ and first suppose that $X\not\subseteq A_x\cup\{x\}$. Because of the definition of connected sets, this assumption means that $X\cap \tc(R^{-x})\neq \emptyset$, where $R^{-x}=R|_{A\setminus \{x\}}$ denotes the profile derived from $R$ by removing $x$. We use this fact to derive a contradiction as follows: starting at $R$, we repeatedly weaken $x$ against all alternatives until we derive a profile $R'$ in which $x$ is the Condorcet loser. \WSMON entails for every step that the choice set does not change, which means that $f(R)=f(R')$. Moreover, $\tc(R')=\tc(R)\setminus (A_x\cup\{x\})$ because for the top cycle it is irrelevant whether an alternative is a Condorcet loser or not present at all. However, this means that $\tc(R')\not\subseteq f(R')$ because $X\cap \tc(R')\neq\emptyset$, but $X \cap f(R')=X\cap f(R)=\emptyset$. Since $|\tc(R')|\leq |\tc(R)\setminus \{x\}|<k_f$, this contradicts the definition of $k_f$ which requires that $\tc(\bar R)\subseteq f(\bar R)$ for all profiles $\bar R$ with $|\tc(\bar R)|<k_f$.
	
		As second case, suppose that $X\subseteq A_x\cup\{x\}$. In this case, consider a second alternative $y\in X\setminus \{x\}$, which means that $y\in A_x$. We want to use \Cref{lem:connectors}. Note that there is an alternative $z\in f(R)$ with $z\succsim_R x$ because of \CWCS and strong Condorcet-consistency. Hence, $x$ does not dominate all other alternatives but $y$ and \Cref{lem:connectors} consequently shows that $A_y\subseteq A_x$. In particular, this means that $x\not\in A_y\cup\{y\}$, and therefore $X\not\subseteq A_y\cup\{y\}$. Hence, we can use the same argument as in the last case to derive a contradiction by focusing on $y$. Since both cases result in a contradiction, it follows that the assumption $|\tc(R)\cap f(R)|\leq k_f-2$ was incorrect, i.e., $|\tc(R)\cap f(R)|\geq k_f-1$ holds for all profiles $R$ with $|\tc(R)|= k_f$.
	\end{proof}	
	
	\Cref{lem:size} is important because it implies for all profiles $R$ with $|\tc(R)\cap f(R)|<|\tc(R)|=k_f$ that there is a single alternative of the top cycle which is unchosen. As we demonstrate next, this insight can be used to strengthen the axioms in \Cref{app:axioms} when we restrict attention to profiles $R, R'$ with $|\tc(R)|=|\tc(R')|=k_f$. In particular, the next lemma is concerned with what happens when we weaken an alternative $y\in f(R)\cap\tc(R)$ against multiple alternatives $X\subseteq f(R)\cap \tc(R)$ when $|\tc(R)|=k_f$.
	
	\begin{lemma}\label{lem:swap}
		Let $f$ denote a pairwise SCC that satisfies strong Condorcet-consistency, homogeneity, strategyproofness, and $4\leq k_f\leq m$, and consider two preference profiles $R$, $R'$ such that $\tc(R)=\tc(R')$ and $|\tc(R)|=k_f$. If there are a set of alternatives $X\subseteq f(R)\cap \tc(R)$ and an alternative $y\in \left(f(R)\cap \tc(R)\right)\setminus X$ such that $g_{R'}(x,y)=2+g_R(x,y)$ for all $x\in X$, and $g_{R'}(x',y')=g_R(x',y')$ for all other pairs of alternatives, it holds that $f(R)=f(R')$ or $f(R')\cap\tc(R')=\tc(R')\setminus \{y\}$.  
	\end{lemma}
	\begin{proof}
		Consider a pairwise SCC $f$ that satisfies homogeneity, strategyproofness, and strong Condorcet-consistency and let $R$, $R'$, $X$, and $y$ be defined as in the lemma. In particular, it holds that $g_{R'}(x,y)=g_{R}(x,y)+2$ for all $x\in X$, and $g_{R'}(x',y')=g_R(x',y')$ for all other pairs of alternatives. This means that we can transform $R$ into a profile $R^*$ with the same majority margins as $R'$ by reinforcing all alternatives in $X$ against $y$. Consequently, the lemma follows if we show that $X\subseteq f(R')$: if also $y\in f(R')$, then \IR entails that $f(R)=f(R')$, and if $y\not\in f(R')$, then \Cref{lem:size} implies that $f(R')\cap\tc(R')=\tc(R')\setminus \{y\}$ because $|\tc(R')|=k_f\geq 4$ and $y\in \tc(R')\setminus f(R')$. 
		
		Thus, suppose for contradiction that $X\not\subseteq f(R')$. Then, \Cref{lem:size} shows that there is an alternative $z\in X$ such that $f(R')\cap\tc(R')=\tc(R')\setminus \{z\}$ because $X\subseteq \tc(R')$. On the other hand, we assume that $X\cup \{y\}\subseteq f(R)$. We can turn these observations into a manipulation of $f$ by adding two voters $i^*$ and $j^*$ with inverse preferences to $R$. In more detail, the profile $R^1$ consists of $R$ and the voters $i^*$ and $j^*$ whose preference relations are specified subsequently. In the definitions of these preference relations, we use $\bar f(R)=f(R)\setminus (X\cup \{y\})$ and $\bar X=X\setminus \{z\}$. 
		\begin{align*}
			\succ_{i^*}^1&=\lex\left(A\setminus f(R)\right), \lex\left(\bar f(R)\cap f(R')\right), y, \lex\left(\bar X\right), z, \lex\left(\bar f(R)\setminus f(R')\right)\\
			\succ_{j^*}^1&= \lex\left(\bar f(R)\setminus f(R')\right)^{-1}\!, z, \lex\left(\bar X\right)^{-1}\!, y, \lex\left(\bar f(R)\cap f(R')\right)^{-1}\!, \lex\left(A\setminus f(R)\right)^{-1}
		\end{align*}
		
		Since the preferences of these voters are inverse, it follows from pairwiseness that $f(R^1)=f(R)$. Next, we derive $R^2$ from $R^1$ by letting voter $i$ reinforce all alternatives in $X$ against $y$. Since we derive $R'$ from $R$ by the same modification, pairwiseness shows that $f(R^2)=f(R')$. However, this means that voter $i^*$ can manipulate by deviating from $R^1$ to $R^2$. Note for this that $f(R')\setminus f(R)\succ_{i^*} f(R)$ because $A\setminus f(R)\succ_{i^*} f(R)$. Furthermore, it holds that $f(R')\succ_{i^*} f(R)\setminus f(R')$ because the alternatives in $f(R)\setminus f(R')=(\bar f(R)\setminus f(R'))\cup \{z\}$ are bottom-ranked by voter $i^*$. Finally, since $z\in f(R)\setminus f(R')$, this is indeed a manipulation for voter $i^*$. Hence, the assumption that $X\not\subseteq f(R')$ was incorrect, which proves the lemma. 
	\end{proof}

\Cref{lem:swap} significantly strengthens \WMON for profiles $R, R'$ with $\tc(R)=\tc(R')$ and $|\tc(R)|=k_f\geq 4$ and alternatives in $\tc(R)$. In particular, we can now reinforce sets of alternatives against single alternatives and there are only two possible outcomes under the given assumptions. Therefore, we ensure in the following that the premises of \Cref{lem:swap} are always true: in all subsequent profiles $R$, it holds that $|\tc(R)|=k_f$, we only modify the preferences between alternatives in the top cycle, and the top cycle will never change. As the next step, we derive a profile $R$ for which all majority margins are known, $|\tc(R)|=k_f$, and $\tc(R)\not\subseteq f(R)$.

\begin{lemma}\label{lem:cycle}
	Let $f$ denote a pairwise SCC that satisfies strong Condorcet-consistency, homogeneity, and strategyproofness. If $4\leq k_f\leq m$, there is a profile $R$ such that $\tc(R)\not\subseteq f(R)$, $|\tc(R)|=k_f$ and there is a cycle $C=(a_1,\dots, a_{k_f})$ in $\succsim_{R}\!\!|_{\tc(R)}$ with $g_{R}(a_{k_f}, a_1)=2$ and $g_{R}(a_i, a_j)=2$ for all other indices $i,j\in\{1,\dots, k_f\}$ with $i<j$.
\end{lemma}
\begin{proof}
	Let $f$ denote a pairwise SCC that satisfies strong Condorcet-consistency, homogeneity, and strategyproofness, and suppose that $k_f\in \{4,\dots, m\}$. Moreover, consider a profile $R^*$ such that $\tc(R^*)\not\subseteq f(R^*)$ and $k_f=|\tc(R^*)|$; such a profile exists by the definition of $k_f$. Additionally, we assume in the sequel that $R^*$ is defined by an even number of voters. This is possible as we can simply duplicate the profile $R^*$ if it is defined by an odd number of voters. This step does neither affect the top cycle nor $f$ since both SCCs are homogeneous, and we can thus work with this larger profile if $R^*$ was defined by an odd number.
	
	We prove this lemma in two steps: first, we construct a profile $\hat R^1$ such that $\tc(R^*)=\tc(\hat R^1)\not\subseteq f(\hat R^1)$ and there is a pair of alternatives $a,b\in\tc(\hat R^1)$ with $b \succ_{\hat R^1} A\setminus \{a,b\}$. This profile is essential since \CWCS shows now that $a$ must be chosen, even after various manipulations. Based on this insight, we construct as a second step a profile $\hat R^2$ that satisfies all requirements of our lemma.\bigskip
	
	\textbf{Step 1: Constructing the profile $\hat R^1$.}
	
	Our first goal is to construct a profile $\hat R^1$ such that $\tc(R^*)=\tc(\hat R^1)\not\subseteq f(\hat R^1)$ and there is a pair of alternatives $a,b\in\tc(\hat R^1)$ with $b \succ_{\hat R^1} A\setminus \{a,b\}$. For this, consider a cycle $C=(a_1, \dots, a_{k_f})$ in $\succsim_{R^*}$ that contains all alternatives in $\tc(R^*)$; such a cycle exists because of \Cref{lem:HC}. Furthermore, let $b=a_{i+1}$ denote an arbitrary alternative in $\tc(R^*)\cap f(R^*)$ and let $a=a_i$ denote its predecessor on the cycle $C$. Our goal is to reinforce $b$ against all alternatives $A\setminus \{a,b\}$ such that $b\succ_R A\setminus \{a,b\}$.
	
	The first key insight for this is that strong Condorcet-consistency and \CWCS entail that there is always a chosen alternative $c$ that dominates $b$ if there is no Condorcet winner. Based on this observation, we repeat the following steps starting at profile $R^*$: in the current profile $R'$, we identify an alternative $c\in f(R')\setminus \{a,b\}$ with $c \succsim_{R'} b$ and reinforce $b$ against $c$. First, note that, during all of these steps, $b$ remains chosen because of \WMON and the fact that we only swap chosen alternatives. Next, observe that these steps do not affect the cycle $C$ because $c$ is not the predecessor of $b$. Hence, \Cref{lem:HC} implies that the top cycle does not change and that $c\in \tc(R^*)=\tc(R')$ because $c\succsim_{R'} b$. The latter observation and \Cref{lem:swap} also entail that not all alternatives in the top cycle are chosen after reinforcing $b$ against $c$ because either $c$ is now unchosen or the choice set is not allowed to change at all. Thus, this process terminates at a profile $R^1$ such that $b \succ_{R^1} f(R^1)\setminus \{a,b\}$, $\tc(R^1)=\tc(R^*)\not\subseteq f(R^1)$, and $\{a,b\}\subseteq f(R^1)$. The last point is true since \WMON shows that $b\in f(R^1)$ and \CWCS requires that $a\in f(R^1)$ because $a\succsim_{R^1} b$ and $b \succ_{R^1} f(R^1)\setminus \{a,b\}$. We are done after this step if $b \succ_{R^1} A\setminus \{a,b\}$, but this is not guaranteed. 
	
	Hence, assume that there are alternatives $x \in A\setminus f(R^1)$ with $x \succsim_{R^1} b$. Note that this assumption implies that $x\in \tc(R^1)$ because $b\in \tc(R^*)=\tc(R^1)$. We want to repeatedly identify such an alternative $x\in \tc(R^1)\setminus f(R^1)$ with $x\succsim_{R^1}b$ and reinforce $b$ against $x$. \WMON and \IR imply for each of these steps that either the choice set does not change, or $b$ becomes unchosen and $x$ chosen. In particular, this means that after such a step, not all alternatives of the top cycle are chosen because the top cycle is not affected by these changes. However, we cannot guarantee that $b$ remains chosen during these steps, and therefore we need to treat the case that we arrive at a profile $R^2$ with $b\not\in f(R^2)$ separately. Given such a profile $R^2$, we show how we can find another profile $R^3$ such that $b\in f(R^3)$, $\tc(R^*)=\tc(R^3)\not\subseteq f(R^3)$, and $g_{R^3}(b,x)=g_{R^2}(b,x)$ for all $x\in A$. For this, note that the cycle $C=(a_1, \dots, a_{k_f})$ exists also in $R^2$ and thus, $\tc(R^2)=\tc(R^*)$. Since $|\tc(R^*)|=k_f$, \Cref{lem:size} shows that $\tc(R^2)\setminus \{b\}\subseteq f(R^2)$ if $b\not\in f(R^2)$. In particular, this means that $c=a_{i+2}$, i.e., the successor of $b$ on $C$, is in $f(R^2)$. We apply next a similar idea as in the construction of $R^1$: at every preference profile $R'$, we reinforce $c$ against a chosen alternative $x\in f(R')\setminus \{b,c\}$ with $x \succsim_{R'} c$. Just as in the first step, \Cref{lem:swap} implies that $c$ remains chosen during these steps and that we never choose all alternatives of $\tc(R')$. Also, we do not flip any edge in the cycle $C$ during this process because we never reinforce $c$ against its predecessor $b$. Hence, neither the top cycle nor a majority margin involving $b$ change. Finally, this process terminates at a profile $R^3$ such that $c\succ_{R^3} f(R^3)\setminus \{b,c\}$ and $c\in f(R^3)$. Moreover, \CWCS now requires that $b\in f(R^3)$ because $c\succ_{R^3} A\setminus \{b,c\}$, i.e., if $b\not\in f(R^3)$, no chosen alternative dominates $c$. Hence, profile $R^3$ indeed satisfies all our requirements.
	
	Thus, if $b$ drops out of the choice set after reinforcing it against an unchosen alternative, we can apply this construction to derive a profile $R^3$ with $b\in f(R^3)$ and $\tc(R^*)= \tc(R^3)\not\subseteq f(R^3)$. At this point, we can simply repeat the same constructions used in the derivation of $R^1$ and $R^2$, and eventually, we will arrive at a profile $\hat R^1$ such that $b \succ_{\hat R^1} A\setminus \{a,b\}$ because the majority margins of $b$ are non-decreasing during all steps and strictly decreasing during the constructions of $R^1$ and $R^2$. Also, none of the constructions requires us to invert edges of the cycle $C$, and thus $\tc(\hat R^1)=\tc(R^*)$, whereas \Cref{lem:swap} shows that $\tc(\hat R^1)\not\subseteq f(\hat R^1)$.\bigskip
	
	\textbf{Step 2: Constructing the profile $\hat R^2$.}
	
	As a second step, we construct the profile $\hat R^2$ that satisfies all requirements of the lemma. In more detail, $\hat R^2$ has to satisfy that $\tc(R^*)=\tc(\hat R^2)\not\subseteq f(\hat R^2)$ and that there is a cycle $C=(a_1, \dots, a_{k_f})$ in $\succsim_{\hat R^2}$ that connects all alternatives in $\tc(\hat R^2)$ such that $g_{\hat R^2}(a_{k_f}, a_i)=2$ and $g_{\hat R^2}(a_i, a_j)=2$ for all other indices $i,j\in \{1,\dots, k_f\}$ with $i<j$. For the construction of this profile, let $R^1$ denote the profile constructed in the last step, and let $C=(a_1, \dots, a_{k_f})$ denote a cycle that connects all alternatives $x\in\tc(R^*)=\tc(R^1)$ in $\succsim_{R^1}$. By construction, there are alternatives $a,b\in \tc(R^1)$ such that $b\succ_{R^1} A\setminus \{a,b\}$. In the sequel, we assume without loss of generality that $b=a_1$ since we can pick the starting point of the cycle. This means that $a=a_{k_f}$, i.e., $a$ is the predecessor of $b$ on the cycle $C$, because $a$ is the only alternative that dominates $b$. Finally, recall that $R^*$ and therefore also $R^1$ are defined by an even number of voters, which implies that the majority margins are even. 
	
	The central observation for the construction of $\hat R^2$ is that \CWCS and strong Condorcet-consistency guarantee that $a\in f(R^1)$ because $a$ is the only alternative that dominates $b$. Even more, this is true as long as $b\succ_R A\setminus \{a,b\}$ and $a\succsim_{R} b$. We use this observation to reinforce $a=a_{k_f}$ against the alternatives $x\in \tc(R^1)\setminus \{a_1, a_{k_f-1}, a_{k_f}\}$: in each step, we identify an alternative $x\in \tc(R^1)\setminus \{a_1, a_{k_f-1}, a_{k_f}\}$ with $x\succsim_{R'} a$ in the current profile $R'$ and reinforce $a$ against $x$. As mentioned before, \CWCS implies that $a$ has to be chosen during all steps. Moreover, if $x\not\in f(R')$, it follows from \WMON that $x$ remains unchosen after this step; otherwise, we could revert the modification and \WMON entails that $x\in f(R')$, contradicting our previous assumption. Hence, \IR implies that the choice set cannot change in this case. On the other hand, if $x\in f(R')$, it follows either that $x$ is no longer chosen after this step, or that the choice set is not allowed to change because of \IR. In particular, this shows that not all alternatives in $\tc(R^*)=\{a_1, \dots, a_{k_f}\}$ are chosen after this step. Hence, we can repeat these steps until we arrive at a profile $R^2$ such that $a \succ_{R^2} A\setminus \{a,a_1, a_{k_f-1}\}$. Also note that $g_{R^2}(a,x)=2$ for all $x\in \tc(R^1)\setminus \{a, a_1, a_{k_f-1}\}$ with $x\succsim_{R^1} a$ because we only reinforce $a$ against such alternatives $x$ until $a$ strictly dominates them. Finally, none of these steps involves an edge of the cycle $C$, which implies that $\tc(R^2)=\tc(R^1)=\tc(R^*)$. Hence, $\tc(R^2)=\tc(R^*)\not\subseteq f(R^2)$.
	
	As next step, we reinforce $a$ against $b$ if $a \sim_{R^2} b$, and against its predecessor $a_{k_f-1}$ on the cycle $C$ until $g_{R^3}(a,a_{k_f-1})=-2$ if $g_{R^2}(a,a_{k_f-1})\leq -4$. This results in a new profile $R^3$ and, by the same arguments as before, it follows that not all alternatives in $\tc(R^*)$ are chosen. Also, it is easy to see that the top cycle did not change since $a \succ_{R^3} b$ and $a_{k_f-1} \succsim_{R^3} a$. As last point, observe that all new outgoing edges $a \succ_{R^2} x$ have weight $2$ and that the incoming edge from $a_{k_f-1}$ has a weight of at most $2$. 
	
	Finally, note that $a$ dominates each alternative $x\in A\setminus \{a, a_{k_f-1}\}$ in $R^3$, and thus, \CWCS and strong Condorcet-consistency imply now that $a_{k_f-1}$ must be chosen. Hence, we can repeat the previous steps for $a_{k_f-1}$, or more generally, we can traverse along the cycle $C$ using these steps. Thus, we repeat this process until we applied our constructions to $a_1$. It follows from the construction that each edge in the final profile $\hat R^2$ has weight $2$. Furthermore, the cycle $C$ also exists in the final profile $\hat R^2$ and thus, $\tc(\hat R^2)=\tc(R^*)$. Moreover, it is a consequence of \CWCS, \WMON, and \IR that $\tc(\hat R^2)\not\subseteq f(\hat R^2)$. Finally, it follows for the profile $\hat R^2$ that $a_1$ dominates each alternative but $a_{k_f}$, and each alternative $a_i$ with $1<i<k_f$ dominates all alternatives $a_j$ with $j>i$. This claim follows by inspecting our construction in more detail: if $j={i+1}$, i.e., if $a_j$ is the successor of $a_i$ in $C$, this follows immediately as we do not break the cycle. If $j>i+1$, we first apply our construction to $a_j$ ensuring that $a_j$ dominates $a_i$. Later, we apply our construction to $a_i$, which reverts this edge and ensures that it has a weight of $2$. Since this majority margin will not be modified anymore, this proves that the profile $\hat R^2$ indeed satisfies all criteria of the lemma.
\end{proof}

If we consider a pairwise SCC $f$ that satisfies all required axioms but is no robust dominant set rule, \Cref{lem:cycle} states the exact majority margins of a profile $R^*$ such that $\tc(R^*)\not\subseteq f(R^*)$ and $|\tc(R^*)|=k_f\leq m$. As the last step, we derive a contradiction to this by showing that $\tc(R^*)\subseteq f(R^*)$ is required. By considering the contraposition of \Cref{lem:cycle}, we infer from this that $k_f\not\in \{4,\dots, m\}$. Together with \Cref{lem:3alt}, this means that $k_f=m+1$, which shows that every pairwise SCC that satisfies strong Condorcet-consistency, homogeneity and strategyproofness is a robust dominant set rule. 

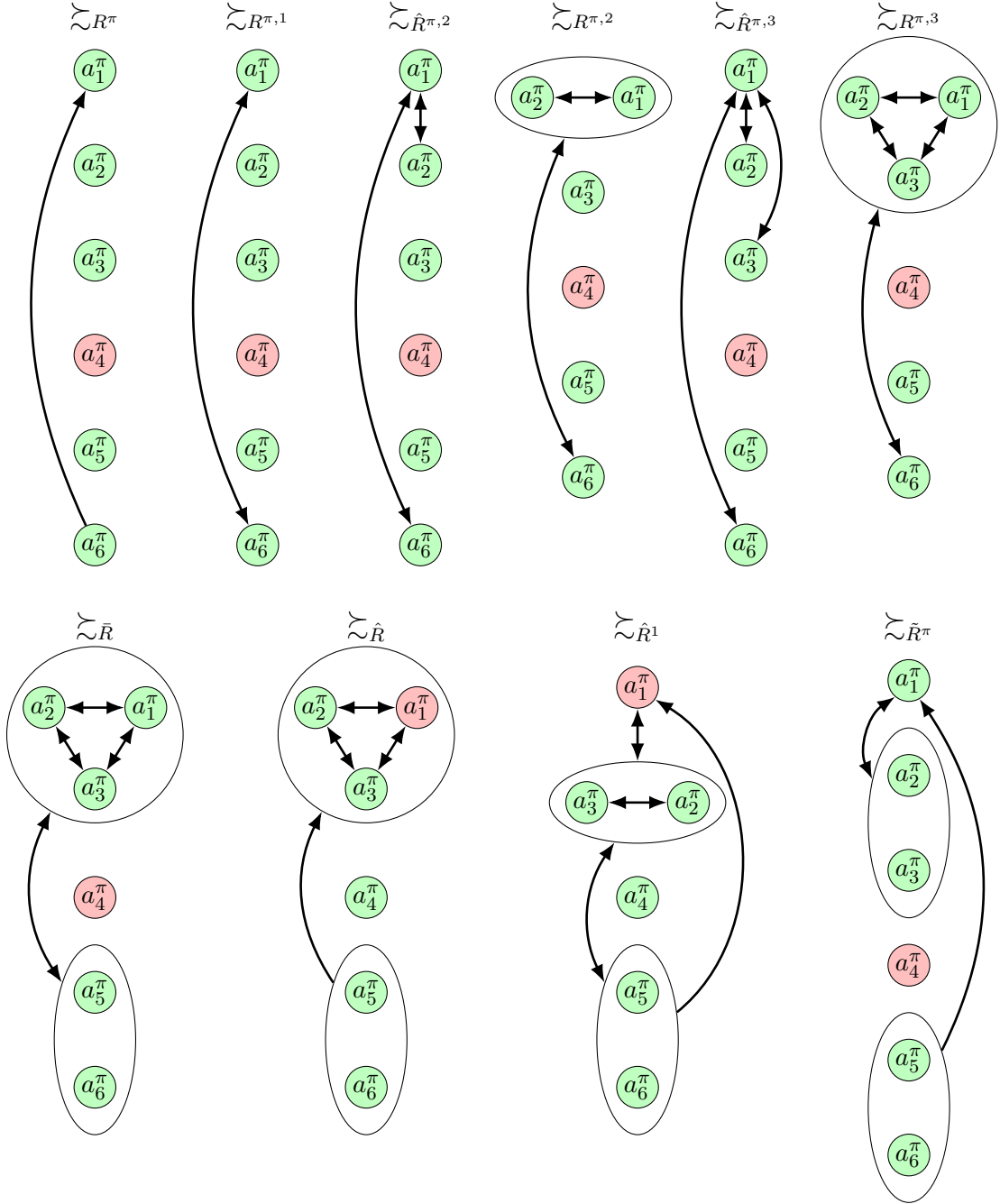
\begin{figure}[tbp]
%	\scalebox{1}{
		\begin{tikzpicture}
		\newcommand\x{0}
		\newcommand\y{0}
		\newcommand\step{1.4}
		\newcommand\titleheight{7.8}
		\tikzstyle{rednode}=[shape=circle,draw=black, minimum size = 1.6em, inner sep =0, fill=red!25]
		\tikzstyle{greennode}=[shape=circle,draw=black, minimum size = 1.6em, inner sep =0, fill=green!25]
		
		\node[greennode] (6) at (\x,\y) {$a_{6}^\pi$};
		\node[greennode] (5) at (\x,\y+1*\step) {$a_{5}^\pi$};
		\node[rednode] (4) at (\x,\y+2*\step) {$a_{4}^\pi$};
		\node[greennode] (3) at (\x,\y+3*\step) {$a_{3}^\pi$};
		\node[greennode] (2) at (\x,\y+4*\step) {$a_{2}^\pi$};
		\node[greennode] (1) at (\x,\y+5*\step) {$a_1^{\pi}$};	
		\node [align = center]  at (\x,\y+\titleheight) {\large$\succsim_{R^\pi}$};
		\draw [line width=1, -Latex] (6) to[bend left = 25] (1);
		
		\renewcommand{\x}{2.4}
		\node[greennode] (6) at (\x,\y) {$a_{6}^\pi$};
		\node[greennode] (5) at (\x,\y+1*\step) {$a_{5}^\pi$};
		\node[rednode] (4) at (\x,\y+2*\step) {$a_{4}^\pi$};
		\node[greennode] (3) at (\x,\y+3*\step) {$a_{3}^\pi$};
		\node[greennode] (2) at (\x,\y+4*\step) {$a_{2}^\pi$};
		\node[greennode] (1) at (\x,\y+5*\step) {$a_1^{\pi}$};
		\node [align = center]  at (\x,\y+\titleheight) {\large$\succsim_{R^{\pi,1}}$};
		\draw [line width=1, Latex-Latex] (6) to[bend left = 25] (1);
	
		\renewcommand{\x}{4.8}
		\node[greennode] (6) at (\x,\y) {$a_{6}^\pi$};
		\node[greennode] (5) at (\x,\y+1*\step) {$a_{5}^\pi$};
		\node[rednode] (4) at (\x,\y+2*\step) {$a_{4}^\pi$};
		\node[greennode] (3) at (\x,\y+3*\step) {$a_{3}^\pi$};
		\node[greennode] (2) at (\x,\y+4*\step) {$a_{2}^\pi$};
		\node[greennode] (1) at (\x,\y+5*\step) {$a_1^{\pi}$};
		\node [align = center]  at (\x,\y+\titleheight) {\large$\succsim_{\hat R^{\pi,2}}$};
		\draw [line width=1, Latex-Latex] (6) to[bend left = 25] (1);
		\draw [line width=1, Latex-Latex] (2) to (1);
		
		\renewcommand{\x}{7.2}
		\node[greennode] (6) at (\x,\y+1) {$a_{6}^\pi$};
		\node[greennode] (5) at (\x,\y+2.4) {$a_{5}^\pi$};
		\node[rednode] (4) at (\x,\y+3.8) {$a_{4}^\pi$};
		\node[greennode] (3) at (\x,\y+5.2) {$a_{3}^\pi$};
		\node[greennode] (2) at (\x-0.75,\y+6.6) {$a_{2}^\pi$};
		\node[greennode] (1) at (\x+0.75,\y+6.6) {$a_1^{\pi}$};	
		\node[shape=ellipse, draw=black, minimum width = 2.6cm, minimum height=1.2cm] (e1) at (\x,\y+6.6) {};
		\node [align = center]  at (\x,\y+\titleheight) {\large$\succsim_{R^{\pi,2}}$};
		\draw [line width=1, Latex-Latex] (6) to[bend left = 25] (e1);
		\draw [line width=1, Latex-Latex] (2) to (1);
		
		\renewcommand{\x}{9.6}
		\node[greennode] (6) at (\x,\y) {$a_{6}^\pi$};
		\node[greennode] (5) at (\x,\y+1*\step) {$a_{5}^\pi$};
		\node[rednode] (4) at (\x,\y+2*\step) {$a_{4}^\pi$};
		\node[greennode] (3) at (\x,\y+3*\step) {$a_{3}^\pi$};
		\node[greennode] (2) at (\x,\y+4*\step) {$a_{2}^\pi$};
		\node[greennode] (1) at (\x,\y+5*\step) {$a_1^{\pi}$};
		\node [align = center]  at (\x,\y+\titleheight) {\large$\succsim_{\hat R^{\pi,3}}$};
		\draw [line width=1, Latex-Latex] (6) to[bend left = 25] (1);
		\draw [line width=1, Latex-Latex] (2) to (1);
		\draw [line width=1, Latex-Latex] (3) to[bend right = 30] (1);
		
		\renewcommand{\x}{12}
		\node[greennode] (6) at (\x,\y+1) {$a_{6}^\pi$};
		\node[greennode] (5) at (\x,\y+2.4) {$a_{5}^\pi$};
		\node[rednode] (4) at (\x,\y+3.8) {$a_{4}^\pi$};
		\node[greennode] (3) at (\x,\y+5.4) {$a_{3}^\pi$};
		\node[greennode] (2) at (\x-0.75,\y+6.6) {$a_{2}^\pi$};
		\node[greennode] (1) at (\x+0.75,\y+6.6) {$a_1^{\pi}$};	
		\node[shape=ellipse, draw=black, minimum width = 2.6cm, minimum height=2.6cm] (e1) at (\x,\y+6.2) {};
		\node [align = center]  at (\x,\y+\titleheight) {\large$\succsim_{R^{\pi,3}}$};
		\draw [line width=1, Latex-Latex] (6) to[bend left = 20] (e1);
		\draw [line width=1, Latex-Latex] (2) to (1);
		\draw [line width=1, Latex-Latex] (3) to (1);
		\draw [line width=1, Latex-Latex] (2) to (3);
		
		\renewcommand{\y}{-9}
		\renewcommand{\x}{0}
		\node[greennode] (6) at (\x,\y+1) {$a_{6}^\pi$};
		\node[greennode] (5) at (\x,\y+2.4) {$a_{5}^\pi$};
		\node[rednode] (4) at (\x,\y+3.8) {$a_{4}^\pi$};
		\node[greennode] (3) at (\x,\y+5.4) {$a_{3}^\pi$};
		\node[greennode] (2) at (\x-0.75,\y+6.6) {$a_{2}^\pi$};
		\node[greennode] (1) at (\x+0.75,\y+6.6) {$a_1^{\pi}$};	
		\node[shape=ellipse, draw=black, minimum width = 2.6cm, minimum height=2.6cm] (e1) at (\x,\y+6.2) {};
		\node[shape=ellipse, draw=black, minimum width = 1.2cm, minimum height=2.8cm] (e2) at (\x,\y+1.7) {};
		\node [align = center]  at (\x,\y+\titleheight) {\large$\succsim_{\bar R}$};
		\draw [line width=1, Latex-Latex] (e2) to[bend left = 30] (e1);
		\draw [line width=1, Latex-Latex] (2) to (1);
		\draw [line width=1, Latex-Latex] (3) to (1);
		\draw [line width=1, Latex-Latex] (2) to (3);
		
		\renewcommand{\x}{4}
		\node[greennode] (6) at (\x,\y+1) {$a_{6}^\pi$};
		\node[greennode] (5) at (\x,\y+2.4) {$a_{5}^\pi$};
		\node[greennode] (4) at (\x,\y+3.8) {$a_{4}^\pi$};
		\node[greennode] (3) at (\x,\y+5.4) {$a_{3}^\pi$};
		\node[greennode] (2) at (\x-0.75,\y+6.6) {$a_{2}^\pi$};
		\node[rednode] (1) at (\x+0.75,\y+6.6) {$a_1^{\pi}$};	
		\node[shape=ellipse, draw=black, minimum width = 2.6cm, minimum height=2.6cm] (e1) at (\x,\y+6.2) {};
		\node[shape=ellipse, draw=black, minimum width = 1.2cm, minimum height=2.8cm] (e2) at (\x,\y+1.7) {};
		\node [align = center]  at (\x,\y+\titleheight) {\large$\succsim_{\hat R}$};
		\draw [line width=1, -Latex] (e2) to[bend left = 30] (e1);
		\draw [line width=1, Latex-Latex] (2) to (1);
		\draw [line width=1, Latex-Latex] (3) to (1);
		\draw [line width=1, Latex-Latex] (2) to (3);
		
		\renewcommand{\x}{8}
		\node[greennode] (6) at (\x,\y+1) {$a_{6}^\pi$};
		\node[greennode] (5) at (\x,\y+\step+1) {$a_{5}^\pi$};
		\node[greennode] (4) at (\x,\y+2*\step+1) {$a_{4}^\pi$};
		\node[greennode] (3) at (\x-0.75,\y+3*\step+1) {$a_{3}^\pi$};
		\node[greennode] (2) at (\x+0.75,\y+3*\step+1) {$a_{2}^\pi$};
		\node[rednode] (1) at (\x,\y+4*\step+1.3) {$a_1^{\pi}$};	
		\node[shape=ellipse, draw=black, minimum width = 1.2cm, minimum height=2.8cm] (e2) at (\x,\y+1+0.5*\step) {};
		\node[shape=ellipse, draw=black, minimum width = 2.6cm, minimum height=1.2cm] (e1) at (\x,\y+1+3*\step) {};
		\node [align = center]  at (\x,\y+\titleheight) {\large$\succsim_{\hat R^1}$};
		\draw [line width=1, Latex-Latex] (e2) to[bend left = 30] (e1);
		\draw [line width=1, Latex-Latex] (1) to (e1);
		\draw [line width=1, Latex-Latex] (2) to (3);
		\draw [line width=1, -Latex] (e2) to[bend right = 55] (1);
		
		\renewcommand{\x}{12}
		\node[greennode] (6) at (\x,\y) {$a_{6}^\pi$};
		\node[greennode] (5) at (\x,\y+1*\step) {$a_{5}^\pi$};
		\node[rednode] (4) at (\x,\y+2*\step) {$a_{4}^\pi$};
		\node[greennode] (3) at (\x,\y+3*\step) {$a_{3}^\pi$};
		\node[greennode] (2) at (\x,\y+4*\step) {$a_{2}^\pi$};
		\node[greennode] (1) at (\x,\y+5*\step) {$a_1^{\pi}$};
		\node[shape=ellipse, draw=black, minimum width = 1.2cm, minimum height=2.8cm] (e2) at (\x,\y+0.5*\step) {};
		\node[shape=ellipse, draw=black, minimum width = 1.2cm, minimum height=2.8cm] (e1) at (\x,\y+3.5*\step) {};
		\node [align = center]  at (\x,\y+\titleheight) {\large$\succsim_{\tilde{R}^\pi}$};
		\draw [line width=1, Latex-Latex] (e1) to[bend left = 40] (1);
		\draw [line width=1, -Latex] (e2) to[bend right = 30] (1);
		\end{tikzpicture}
%	}
	\captionsetup{belowskip=0pt}
	\captionsetup{aboveskip=2pt}
	\caption{The (weighted) majority relations used in the proof of \Cref{lem:domset} for the case that $k_f=6$. Alternatives outside of the top cycle are not depicted, and we assume that $j^*=4$ and $a_1^\pi\not\in f(\hat R)$, where $\pi$ denotes a permutation in $\Pi^{j^*}_{\pi^*}$. Alternatives placed in an ellipse have identical relationships to all alternatives outside of the ellipse, and all missing edges point downwards. All directed edges indicate a majority margin of $2$ and all bidirectional edges indicate a majority margin of $0$. Green alternatives are chosen and red ones are unchosen by $f$.}
	\label{fig:exdomset}
	\end{figure}

\begin{lemma}\label{lem:domset}
	Every pairwise SCC that satisfies strong Condorcet-consistency, homogeneity, and strategyproofness is a robust dominant set rule.
\end{lemma}
\begin{proof}
	Assume for contradiction that there is a pairwise SCC $f$ that satisfies strategyproofness, homogeneity, and strong Condorcet-consistency, but is no robust dominant set rule. The contraposition of \Cref{lem:TCnosubset} shows that there is a profile $R$ such that $\tc(R)\not\subseteq f(R)$. On the other hand, \Cref{lem:3alt} shows that $\tc(R)\subseteq f(R)$ for all profiles $R$ with $|\tc(R)|\leq 3$. These claims contradict each other if $m\leq 3$ and thus, we focus on the case that $m\geq 4$. Hence, let $k_f\in \{4,\dots, m\}$ denote the maximal value such that $\tc(R)\subseteq f(R)$ for all profiles $R$ with $|\tc(R)|<k_f$. Moreover, let $\bar R$ denote a profile such that $\tc(\bar R)\not\subseteq f(\bar R)$ and $|\tc(\bar R)|=k_f$; such a profile exists because of the definition of $k_f$. Next, we apply \Cref{lem:cycle} to derive a profile $R^*$ such that $\tc(R^*)=\tc(\bar R)$, $\tc(R^*)\not\subseteq f(R^*)$, and the alternatives $\tc(R^*)=\{a_1, \dots, a_{k_f}\}$ can be ordered such that $g_{R^*}(a_{k_f}, a_1)=2$ and $g_{R^*}(a_i, a_j)=2$ for all other indices $i,j\in \{1,\dots, k_f\}$ with $i<j$. Furthermore, $\tc(R^*)\setminus f(R^*)$ contains a single alternative $a_j$ because of \Cref{lem:size}.
	
For deriving a contradiction to this assumption, we will consider a number of profiles related to $R^*$. In particular, our subsequent construction will mimic neutrality since $f$ needs not be neutral. Thus, we define  profile $R^\pi$ given some permutation $\pi:\tc(R^*)\rightarrow\tc(R^*)$ as follows: 
$g_{R^\pi}(x,y)=g_{R^*}(x,y)$ if $x\in A\setminus \tc(R^*)$ or $y\in A\setminus \tc(R^*)$ and $g_{R^\pi}(\pi(a_{i}), \pi(a_j))=g_{R^*}(a_{i}, a_{j})$ for all $a_i, a_j\in \tc(R^*)$. Less formally, $R^\pi$ is constructed as follows: we derive the majority margins of $R^\pi$ by reordering the edges between alternatives $x,y\in \tc(R^*)$ according to $\pi$ but we do not reorder the edges to alternatives outside of the top cycle. For a better readability, we refer to $\pi(a_i)$ as $a_i^\pi$ from now on. In particular, the construction of $R^\pi$ implies that $g_{R^\pi}(a_{k_f}^\pi, a_1^\pi)=2$ and $g_{R^\pi}(a_{i}^\pi, a_j^\pi)=2$ for all $i,j\in \{1,\dots, k_f\}$ with $i<j$. Furthermore, we define $\Pi^{l}_\pi$ as the set of permutations $\pi'$ with $\pi(a_i)=\pi'(a_i)$ for $i\in \{l,\dots,k_f\}$, i.e., all permutations $\pi'\in\Pi^l_\pi$ agree with $\pi$ on the alternatives $\{a_l,\dots, a_{k_f}\}$.

Based on profiles $R^\pi$, we next prove the lemma. For this, let $j^*$ denote the smallest index such that $a_{j^*}^\pi\not\in f(R^\pi)$ for some permutation $\pi$ on $\tc(R^*)$. Moreover, let $\pi^*$ denote the corresponding permutation, i.e., $a_{j^*}^{\pi^*}\not\in f(R^{\pi^*})$. Given the value $j^*$ and the profile $R^{\pi^*}$, we prove the lemma in three steps. First, we show that $\{a_1^\pi, a_2^\pi, a_{k_f}^\pi\}\subseteq f(R^\pi)$ for all permutations $\pi$. This means in particular that $j^*\in \{3,\dots, k_f-1\}$. Next, we show that $f(R^\pi)=f(R^{\pi^*})$ for all permutations $\pi\in \Pi^{j^*}_{\pi^*}$. This observation implies that $a_{j^*}^\pi=a_{j^*}^{\pi^*}\not\in f(R^\pi)$ for all these permutations. Finally, we use this insight to derive a contradiction. All profiles used for Steps 2 and 3 are depicted exemplarily in \Cref{fig:exdomset} for the case that $k_f=6$.\bigskip

\textbf{Step 1: $\{a_1^\pi, a_2^\pi, a_{k_f}^\pi\}\subseteq f(R^\pi)$ for all permutations $\pi:\tc(R^*)\rightarrow\tc(R^*)$}

Consider an arbitrary preference profile $R^\pi$. First, note that there is no Condorcet winner in this profile since there is no Condorcet winner in $R^*$ and thus, strong Condorcet-consistency requires that $|f(R^\pi)|\geq 2$. As a consequence, \CWCS requires that $a_1^\pi\in f(R^{\pi})$ and $a_{k_f}^\pi\in f(R^{\pi})$ because $a_1^\pi$ is the only alternative that dominates $a_2^\pi$ and $a_{k_f}^\pi$ is the only alternative that dominates $a_1^\pi$. As last point, suppose for contradiction that $a_2^\pi\not\in f(R^\pi)$. Since $|\tc(R^\pi)|=k_f$, it follows from \Cref{lem:size} that $\tc(R^\pi)\setminus \{a_2^\pi\} \subseteq f(R^\pi)$, in particular that $a_3^\pi\in f(R^\pi)$. As a next step, we reinforce $a_3^\pi$ twice against $a_1^\pi$ to derive a profile $R'$ with $g_{R'}(a_3^\pi, a_1^\pi)=2$. Note that $a_1^\pi$ needs to stay chosen during these steps because it is still the only alternative dominating $a_2^\pi$ and \WMON implies that $a_3^\pi$ remains also chosen. Hence, it follows from \IR that $f(R')=f(R^\pi)$, which means that $a_2^\pi\not\in f(R')$. However, $a_2^\pi$ is the only alternative that dominates $a_3^\pi$ in $R'$ and thus, \CWCS is violated. This is a contradiction and thus, the assumption that $a_2^\pi\not\in f(R^\pi)$ was incorrect.\bigskip

\textbf{Step 2: $f(R^{\pi})=f(R^{\pi^*})$ for all permutations $\pi\in \Pi^{j^*}_{\pi^*}$}
	
	For proving this step, we consider the profiles $R^{\pi, l}$ for every $l\in \{1,\dots, j^*-1\}$, which differs from $R^\pi$ in the fact that $g_{R^{\pi,l}}(a_{i}^\pi, a_j^\pi)=0$ for all $i,j\in \{1,\dots, l, k_f\}$. Intuitively, $R^{\pi, l}$ is derived $R^\pi$ by introducing a large set of tied alternatives $\{a_1^\pi,\dots, a_l^\pi, a_{k_f}^\pi\}$ in the majority relation. Our goal is to show that $f(R^\pi)=f(R^{\pi, j^*-1})$ for all permutations $\pi\in \Pi^{j^*}_{\pi^*}$. This implies that $f(R^\pi)=f(R^{\pi'})$ for all such permutation $\pi$, $\pi'$ because $g_{R^{\pi, l}}=g_{R^{\pi', l}}$ for all $\pi, \pi'\in \Pi^{l+1}_{\pi^*}$ and $l\in \{1,\dots, j^*-1\}$. For deriving this statement, we show inductively that $f(R^\pi)=f(R^{\pi, l})$ for all $\pi\in \Pi^{j^*}_{\pi^*}$ and all $l\in \{1,\dots, j^*-1\}$.
	
	First, we focus on the induction basis $l=1$ and consider therefore an arbitrary permutation $\pi\in \Pi^{j^*}_{\pi^*}$. Note that $R^{\pi,1}$ only differs from $R^\pi$ by the fact that $g_{R^{\pi, 1}}(a_{k_f}^{\pi}, a_1^\pi)=0$ instead of $2$. Hence, we can derive $R^{\pi, 1}$ from $R^\pi$ by reinforcing $a_1^\pi$ against $a_{k_f}^\pi$. Furthermore, we have shown in the last step that $\{a_1^\pi, a_{k_f}^\pi\}\subseteq f(R^\pi)$, and \CWCS requires that both alternatives are chosen in $f(R^{\pi, 1})$ because $a_{k_f}^\pi$ is still the only alternative that dominates $a_1^\pi$ and $a_1^\pi$ is the only alternative that dominates $a_2^\pi$. Consequently, \IR shows that $f(R^\pi)=f(R^{\pi, 1})$. 
	
	For the induction step, assume that there is a value $l\in \{1,\dots, j^*-2\}$ such that $f(R^\pi)=f(R^{\pi, l})$ for all $\pi\in \Pi^{j^*}_{\pi^*}$. 
	We need to prove that this claim is also true for $l+1$. Hence, note that for every permutation $\pi\in\Pi^{j^*}_{\pi^*}$, it holds that $\{a_1^\pi, \dots, a_{j^*-1}^\pi, a_{k_f}^\pi\}\subseteq f(R^{\pi})$ because of the definition of $j^*$ and Step 1. Next, we explain how to derive $R^{\pi, l+1}$ from $R^\pi$ for an arbitrary permutation $\pi\in \Pi^{j^*}_{\pi^*}$: first we reinforce $a_1^\pi$ against $a_{k_f}^\pi$ to derive the profile $\bar R^{\pi}$. 
	It follows from the same argument as in the in the induction basis that $f(\bar R^\pi)=f(R^\pi)$.
	Next, we reinforce all alternatives $a_i^\pi$ with $i\in \{2, \dots, l+1\}$ against $a_1^\pi$. 
	\CWCS requires for the resulting profile $\hat R^\pi$ that $a_1^\pi$ is chosen because it is still the only alternative dominating $a_2^\pi$ and \Cref{lem:swap} shows therefore that $f(\hat R^\pi)=f(R^\pi)$. Furthermore, observe that $g_{\hat R^\pi}(a_1^\pi, a_i^\pi)=0$ for all $i\in \{2,\dots, l+1, k_f\}$. 
	Finally, we can derive $R^{\pi, l+1}$ from $\hat R^{\pi}$ by letting a voter $i$ with preference relation $\succ_i= a_2^\pi,\dots,a_{l+1}^\pi, a_{k_f}^\pi, \lex(A\setminus \{a_2^\pi,\dots, a_{l+1}^\pi, a_{k_f}^\pi\})$ change his preference relation to $\succ_i'= a_{k_f}^\pi, a_{l+1}^\pi,\dots, a_{2}^\pi, \lex(A\setminus \{a_2^\pi,\dots, a_{l+1}^\pi, a_{k_f}^\pi\})$. We can assume that such a voter exists since pairwiseness allows us to add voters with inverse preferences without affecting the choice set. This step ensures that $g_{R^{\pi, l+1}}(a_{i}^\pi, a_j^\pi)=0$ for all $i,j\in \{2,\dots, l+1, k_f\}$ and it therefore transforms $\hat R^\pi$ into $R^{\pi, l+1}$. 
	
	Now, since $\{a_2^\pi, \dots, a_{l+1}^\pi, a_{k_f}^\pi\}\subseteq f(R^{\pi})=f(\hat R^\pi)$, \IR implies that $f(R^\pi)=f(R^{\pi, l+1})$ if $\{a_2^\pi, \dots, a_{l+1}^\pi, a_{k_f}^\pi\}\subseteq f(R^{\pi, l+1})$. Hence, our next goal is to prove this set inclusion and we assume for contradiction that there is an alternative $a_j^\pi$ with $j\in \{2,\dots, l+1, k_f\}$ such that $a_j^\pi\not\in f(R^{\pi, l+1})$. First, suppose that $a_j^\pi\in \{a_2^\pi, \dots, a_{l+1}^\pi\}$, which means that $\tc(R^\pi)\setminus \{a_j^\pi\}\subseteq f(R^\pi)$ because of \Cref{lem:size}. In this case, we derive a contradiction by considering the permutation $\pi'$ with $a_1^{\pi'}=a_j^\pi$,  $a_j^{\pi'}=a_1^\pi$, and $a_i^{\pi'}=a_i^\pi$ for all other $i\in \{1,\dots, k_f\}$. In more detail, we can use the same construction as for $\pi$ to transform $R^{\pi'}$ into $R^{\pi', l+1}=R^{\pi, l+1}$. In particular, the analysis of the previous paragraph shows that $\{a_1^{\pi'}, \dots, a_{l+1}^{\pi'}, a_{k_f}^{\pi'}\}\subseteq f(R^{\pi'})=f(\hat R^{\pi'})$ and we derive $R^{\pi', l+1}=R^{\pi, l+1}$ from the profile $\hat R^{\pi'}$ by a manipulation that only involves the alternatives $\{a_2^{\pi'}, \dots, a_{l+1}^{\pi'}, a_{k_f}^{\pi'}\}$. Hence, the assumption that $\tc(R^\pi)\setminus \{a_j^\pi\}=\tc(R^{\pi'})\setminus \{a_1^{\pi'}\}\subseteq f(R^{\pi,l+1})$ implies that $f(\hat R^{\pi'})=f(R^{\pi', l+1})$ because of \IR. However, this contradicts that $a_1^{\pi'}=a_j^\pi\not\in f(R^{\pi, l+1})$, which proves that $\{a_2^\pi,\dots, a_{l+1}^\pi\}\subseteq f(R^{\pi, l+1})$. 
	
	As second case, suppose that $a_{k_f}^\pi\not\in f(R^{\pi, l+1})$. 
	In this case, we derive a contradiction to the induction hypothesis by deriving $R^{\pi, l+1}$ from $R^{\pi, l}$. 
	Thus, note that these two preference profiles only differ in majority margins involving $a_{l+1}^\pi$: we have $g_{R^{\pi, l}}(a_{l+1}^\pi, a_{k_f}^\pi)=2$ but $g_{R^{\pi, l+1}}(a_{l+1}^\pi, a_{k_f}^\pi)=0$ and, for all $i\in \{1,\dots, l\}$, $g_{R^{\pi, l}}(a_{i}^\pi, a_{l+1}^\pi)=2$ but $g_{R^{\pi, l+1}}(a_{i}^\pi, a_{l+1}^\pi)=0$. 
	Also, observe that the induction hypothesis implies that $f(R^\pi)=f(R^{\pi,l})$, which means that $\{a_1^\pi, \dots, a_{l+1}^\pi, a_{k_f}^\pi\}\subseteq f(R^{\pi, l})$.
	Hence, we can transform $R^{\pi, l}$ into $R^{\pi, l+1}$ as follows: first, we reinforce $a_{l+1}^\pi$ one after another against all alternative $a_i^\pi$ with $i\in \{1,\dots, l\}$. 
	For each swap, \Cref{lem:swap} implies that the choice set does either not change at all, or all alternatives in $\tc(R^\pi)\setminus \{a_i^\pi\}$ are chosen (where $a_i^\pi$ denotes the weakened alternative).
	In particular, this shows that $a_{l+1}^\pi$ and $a_{k_f}^\pi$ stay chosen during this process. 
	Finally, we reinforce $a_{k_f}^\pi$ against $a_{l+1}^\pi$ to derive $R^{\pi, l+1}$. 
	Then, \WMON implies that $a_{k_f}^{\pi}\in f(R^{\pi, l+1})$, contradicting our assumption. 
	Hence, it follows that $\{a_2^\pi, \dots, a_{l+1}^\pi, a_{k_f}^\pi\}\subseteq f(R^{\pi, l+1})$, which proves the induction step. 
	As a consequence, we infer that $f(R^\pi)=f(R^{\pi, j^*-1})=f(R^{\pi', j^*-1})=f(R^{\pi'})$ for all permutations $\pi,\pi'\in \Pi^{j^*}_{\pi^*}$.\bigskip

\textbf{Step 3: Deriving the contradiction}

	As last step, we derive a contradiction by showing that $a_{j^*}^{\pi^*}\in f(R^{\pi^*})$. This claim is in conflict with the definitions of $j^*$ and $R^{\pi^*}$ which require that $a_{j^*}^{\pi^*}\not\in f(R^{\pi^*})$. For proving this claim, we consider first the profile $\bar R$ which differs in the following majority margins from $R^{\pi^*}$: $g_{\bar R}(a_i^{\pi^*}, a_j^{\pi^*})=0$ for all $i,j\in \{1,\dots, j^*-1\}$ and $g_{\bar R}(a_i^{\pi^*}, a_j^{\pi^*})=0$ for all $i\in \{1,\dots, j^*-1\}$, $j\in \{j^*+1,\dots, k_f\}$. A similar analysis as in Step 2 shows that that $f(R^{\pi^*})=f(\bar R)$. In more detail, we can use an induction on the profiles $\bar R^{\pi, l}$ for all $l\in \{1,\dots, j^*-1\}$ and $\pi\in \Pi^{j^*}_{\pi^*}$, which are defined by the following majority margins: $g_{\bar R^{\pi, l}}(a_i^\pi, a_j^\pi)=0$ for all $i,j\in \{1,\dots, l\}$, $g_{\bar R^{\pi, l}}(a_i^\pi, a_j^\pi)=0$ for all $i\in \{1,\dots, l\}$, $j\in \{j^*+1,\dots, k_f\}$, and $g_{\bar R^{\pi, l}}(x, y)=g_{R^{\pi}}(x, y)$ for all remaining majority margins. More intuitively, the profiles $\bar R^{\pi, l}$ differ from the profiles $R^{\pi, l}$ only in the fact that all alternatives $\{a_1^\pi, \dots, a_l^\pi\}$ are in a majority tie with all alternatives in $\{a^\pi_{j^*+1},\dots, a_{k_f}^\pi\}$ instead of just $a_{k_f}^\pi$. Since the last step shows that $\tc(R^\pi)\setminus \{a_{j^*}^\pi\}\subseteq f(R^\pi)$ for all permutations $\pi\in \Pi^{j^*}_{\pi^*}$, an almost identical induction as in Step 2 shows that $f(R^\pi)=f(\bar R^{\pi, l})$ for all permutations $\pi\in \Pi^{j^*}_{\pi^*}$ and $l\in \{1,\dots, j^*-1\}$. This means in particular that $f(R^{\pi^*})=f(\bar R^{\pi^*, j^*-1})=f(\bar R)$.
	
	Departing from this observation, we now consider profile $\hat R$, which can be derived from $\bar R$ by reinforcing all alternatives in $X_2=\{a_{j^*+1}^{\pi^*}, \dots, a_{k_f}^{\pi^*}\}$ against all alternatives in $X_1=\{a_{1}^{\pi^*}, \dots, a_{j^*-1}^{\pi^*}\}$. 
	This means for the majority margins that $g_{\hat R}(x, y)=2$ for all $x\in X_2$, $y\in X_1$. 
	As a consequence of this observation, $a_{j^*}^{\pi^*}$ is now the only alternative that dominates $a_{j^*+1}^{\pi^*}$ and thus, \CWCS requires that $a_{j^*}^{\pi^*}\in f(\hat R)$. 
	Next, note that a repeated application of \Cref{lem:swap} shows that $X_2\subseteq f(\hat R)$ because we can transform $\bar R$ into $\hat R$ by reinforcing the alternatives $X_2$ against each alternative $x\in X_1$ individually. 
	For each of these steps, \Cref{lem:swap} shows that either the choice set does not change or all alternatives in $\tc(R^{\pi^*})\setminus \{x\}$ are chosen. 
	In particular, this means that $X_2\subseteq f(\hat R)$ and that $\tc(\hat R)\not\subseteq f(\hat R)$. 
	Hence, there is an alternative $a_j^{\pi^*}\in X_1$ such that $a_j^{\pi^*}\not\in f(\hat R)$. 
	
	As a next step, consider the profile $\hat R^j$ derived from $\bar R$ by reinforcing the alternatives in $X_2$ only against $a_j^{\pi^*}$. We show that $a_j^{\pi^*}\not\in f(\hat R^j)$ and assume for the sake of contradiction that this is not the case. Hence, \Cref{lem:swap} implies that $f(\hat R^j)=f(\bar R)$. Moreover, we can now transform $\hat R^j$ into $\hat R$ by reinforcing the alternatives in $X_2$ once against each alternative $x\in X_1\setminus \{a_j^{\pi^*}\}$. For every step, \Cref{lem:swap} shows that $a_j^{\pi^*}$ needs to stay chosen, and thus, we have a contradiction to the assumption that $a_j^{\pi^*}\not\in f(\hat R)$. Hence, it must hold that $a_j^{\pi^*}\not\in f(\hat R^j)$, which implies that $\tc(\hat R^j)\setminus \{a_j^{\pi^*}\}\subseteq f(\hat R^j)$ because of \Cref{lem:size}.
	
	Finally, we derive a contradiction to this observation. Consider for this a permutation $\pi\in \Pi^{j^*}_{\pi^*}$ such that $a_1^\pi=a_j^{\pi^*}$. We show that $a_j^{\pi^*}\in f(\hat R^j)$ by transforming $R^\pi$ into $\hat R^j$ and observe for this $X_1\cup X_2\subseteq f(R^\pi)$ because of Step 2. As first step, we reinforce all alternatives in $X_2\setminus \{a_{k_f}^\pi\}$ against $a_1^\pi$ twice, and the alternatives in $X_1\setminus \{a_1^\pi\}$ once against $a^\pi_1$. During all these steps, \CWCS requires that $a_1^\pi$ remains chosen because it is the only alternative that dominates $a_2^\pi$. In turn, \Cref{lem:swap} implies therefore that the choice set cannot change, i.e., this process results in a profile $\tilde{R}^\pi$ with $f(\tilde{R}^\pi)=f(R^\pi)$. Moreover, note that $g_{\tilde{R}^\pi}(a_1^\pi, x)=g_{\hat R^j}(a_1^\pi, x)$ for all $x\in A\setminus \{a_1^\pi\}$. For the next step, consider a voter $i$ with the preference relation $\succ_i=a_2^\pi, \dots, a_{j^*-1}^\pi, a_{j^*+1}^\pi, \dots, a_{k_f}^\pi, \lex(X)$, where $X$ contains all missing alternatives. We can assume that such a voter exists as pairwiseness allows us to add pairs of voters with inverse preferences without affecting the choice set. Next, we let voter $i$ deviate to the preference relation $a_{j^*+1}^\pi, \dots, a_{k_f}^\pi, a_{j^*-1}^\pi, \dots, a_2^\pi, \lex(X)$, which transforms the profile $\tilde{R}^\pi$ into $\hat R^j$. In particular, we know that all alternatives in $\{a_2^\pi, \dots, a_{k_f}^\pi\}$ are chosen both in $f(\tilde{R}^\pi)$ (because $f(R^\pi)=f(\tilde{R}^\pi)$) and in $f(\hat R^j)$ (because $\tc(R^{\pi^*})\setminus \{a_j^{\pi^*}\}=\tc(R^\pi)\setminus \{a_1^\pi\}\subseteq f(\hat R^j)$). Thus, \IR implies that $f(\tilde{R}^\pi)=f(\hat R^j)$, which conflicts with $a_j^{\pi^*}=a_1^\pi\not\in f(\hat R^j)$. This contradiction proves the lemma.
\end{proof}

	\subsection{Proofs of the Main Results}\label{app:mainresults}

Finally, we are ready to prove our main results. First, we discuss the proof of \Cref{thm:F-SP}: a pairwise, non-imposing, neutral, and homogeneous SCC is strategyproof iff it is a robust dominant set rule. In order to be able to use the results of the previous section, we show that every SCC which satisfies these requirements is strongly Condorcet-consistent.
	
	\begin{lemma}\label{lem:CC}
		Every pairwise SCC that satisfies strategyproofness, non-imposition, homogeneity, and neutrality is strongly Condorcet-consistent. 
	\end{lemma}
	\begin{proof}
		Consider a pairwise SCC $f$ that satisfies non-imposition, homogeneity, neutrality, and strategyproofness. We need to show two claims: if there is a Condorcet winner, it is chosen uniquely by $f$, and if an alternative is the unique winner of $f$, it is the Condorcet winner. We prove these claims separately. \bigskip
		
		\textbf{Claim 1: If $x$ is the Condorcet winner in $R$, then $f(R)=\{x\}$.}
		
		Consider an arbitrary profile $R$ with Condorcet winner $x$. The claim follows by showing that $f(R)=\{x\}$ and thus, let $R'$ denote a profile such that $f(R')=\{x\}$. Such a profile exists because $f$ is non-imposing. Next, we repeatedly use \WSMON to push down the best alternative of every voter until we arrive at a profile $R^1$ such that every voter top-ranks $x$. Since $R^1$ is constructed by repeated application of \WSMON, it follows that $f(R^1)=f(R')=\{x\}$. As a next step, we let all voters order the alternatives in $A\setminus \{x\}$ lexicographically. This leads to the profile $R^2$ and \IUA implies that the choice set does not change, i.e., $f(R^2)=\{x\}$. Moreover, all voters have the same preference relation in $R^2$. Thus, it follows from homogeneity that $f(R^3)=\{x\}$, where $R^3$ consists of a single voter who has the same preference relation as the voters in $R^2$. Next, let $c=\min_{y\in A\setminus \{x\}} g_{R}(x,y)$ denote the smallest majority margin of $x$ in $R$ and note that $c\geq 1$ because $x$ is the Condorcet winner in $R$. We use again homogeneity to construct a profile $R^4$ that consists of $c$ copies of $R^3$, which means that $f(R^4)=\{x\}$. Furthermore, observe that the parity of the number of voters used in $R^4$ is equal to the parity of the number of voters used in $R$. The reason for this is that $c$ is odd iff $R$ is defined by an odd number of voters. 
		
		As the last step, we need to set the majority margins to their values in $R$. For this, we repeat the following procedure on each pair of alternatives $y,z$ with $g_{R^4}(y,z)<g_{R}(y,z)$ until we arrive at a profile $R^5$ with $g_{R^5}=g_{R}$. First, we add two voters $i$ and $j$ to the preference profile such that voter $i$ prefers $x$ the least and ranks $z$ directly over $y$, and voter $j$'s preference relation is inverse to voter $i$'s. Observe that we can assign such a preference relation to voter $i$ because $g_{R^4}(x,z')=c\leq g_{R}(x,z')$ for all $z'\in A\setminus \{x\}$ implies that $x\neq z$. Since the preference relations of these two voters are inverse, the majority margins do not change and pairwiseness requires thus that $x$ is still the unique winner. Next, we let voter $i$ swap $y$ and $z$, which increases the majority margin between $y$ and $z$ by $2$. Moreover, the choice set cannot change during this step because $x$ is voter $i$'s least preferred alternative. Hence, if another set would be chosen, this step is a manipulation for voter $i$ which contradicts strategyproofness. Therefore, we can repeat this process for every pair of alternatives until we derive a profile $R^5$ with $g_{R^5}=g_{R}$ and our arguments show that $f(R^5)=\{x\}$. This proves that $f(R)=\{x\}$ because of pairwiseness, which shows that $f$ chooses the Condorcet winner uniquely whenever it exists.\bigskip
		
		\textbf{Claim 2: If $f(R)=\{x\}$, then $x$ is the Condorcet winner in $R$.}
		
		Next, we focus on the opposite direction and show that if an alternative is chosen as unique winner by $f$, then it is the Condorcet winner. Assume for contradiction that this is not the case, i.e., that there is a preference profile $R$ and an alternative $x$ such $f(R)=\{x\}$ even though $x$ is not the Condorcet winner in $R$. Then, there is an alternative $y\in A\setminus \{x\}$ such that $g_R(y,x)\geq 0$. We continue with a case distinction with respect to whether $g_R(y,x)=0$ or $g_R(y,x)>0$. First assume that $g_R(y,x)>0$. In this case, we can repeatedly reinforce $y$ against all other alternatives $z\in A\setminus \{x\}$.  This process eventually results in a profile $R'$ in which $y$ is the Condorcet winner. Since Claim 1 proves that $f$ elects the Condorcet winner whenever it exists, we derive that $f(R')=\{y\}$. On the other side, it follows from \IUA that these steps do not change the choice set because we only swap unchosen alternatives, i.e., $f(R')=\{x\}$. These two observations contradict each other and thus, the assumption that $g_R(y,x)>0$ was incorrect. 
		
		Next, assume that $g_R(y,x)=0$. In this case, we partition the voters $N$ according to their preferences between $x$ and $y$: we denote with $N_{x\succ y}=\{i\in N\colon x \succ_i y\}$ the set of voters who prefer $x$ to $y$, and with $N_{y\succ x}=\{i\in N\colon y \succ_i x\}$ the set of voters who prefer $y$ to $x$. We let all voters in $N_{x\succ y}$ change their preferences such that $y$ is directly below $x$, and all voters in $N_{y\succ x}$ change their preferences such that $y$ it is directly above $x$. For these steps, \IUA implies that $x$ remains the unique winner as we only reorder unchosen alternatives. Hence, it follows for the resulting profile $R'$ that $f(R')=\{x\}$. However, it holds that $g_{R'}(x,z)=g_{R'}(y,z)$ for all $z\in A\setminus \{x,y\}$ and $g_{R'}(x,y)=0$. Neutrality and pairwiseness thus require that either $\{x,y\}\subseteq f(R')$ or $\{x,y\}\cap f(R')\neq \emptyset$ because renaming $x$ and $y$ does not change the majority margins. This is in conflict with the previous claim, and hence the assumption that $f(R)=\{x\}$ and $g_R(x,y)=0$ was incorrect. We have derived a contradiction in both cases, which proves that $f(R)=\{x\}$ can only be true if $x$ is the Condorcet winner in $R$. 
	\end{proof}

 	Since we established that every SCC that satisfies the requirements of \Cref{thm:F-SP} is strongly Condorcet-consistent, this result follows now easily from \Cref{lem:domset}.

	\FSP*
	\begin{proof}
	Let $f$ denote a pairwise SCC that satisfies homogeneity, neutrality, and non-imposition. The direction from left to right follows from \Cref{lem:domset} and \Cref{lem:CC}: if $f$ is additionally strategyproof, \Cref{lem:CC} shows that $f$ is strongly Condorcet-consistent and, in turn, \Cref{lem:domset} implies that $f$ is a robust dominant set rule. 
	
	Next, we discuss the direction from right to left and assume thus that $f$ is a robust dominant set rule. Furthermore, suppose for contradiction that $f$ is not strategyproof. Hence, there are two preference profiles $R$ and $R'$ and a voter $i$ such that ${\succ_j}={\succ_j'}$ for all $j\in N\setminus \{i\}$ and $f(R')\succ_i^F f(R)$. 
First, assume that $f(R')\setminus f(R)\neq\emptyset$ and observe that $f(R) \succ_R f(R')\setminus f(R)$ since $f(R)$ is a dominant set. Deviating from $R$ to $R'$ is only a manipulation for voter $i$ if $f(R')\setminus f(R)\succ_i f(R)$. However, this means that $f(R) \succ_{R'} f(R')\setminus f(R)$ as voter $i$ can only weaken the alternatives in $f(R')\setminus f(R)$ against those in $f(R)$. Since $f$ is a dominant set rule and $f(R')\setminus f(R)\neq\emptyset$, this implies that $f(R)\subseteq f(R')$. Hence, $f(R) \succ_{R'} A\setminus f(R')$ because $f(R')$ is a dominant set, which proves that $f(R)$ is also a dominant set in $\succsim_{R'}$. As a consequence, robustness requires that $f(R')\subseteq f(R)$, which contradicts the assumption that $f(R')\setminus f(R)\neq \emptyset$. Hence, no manipulation is possible in this case.
	
	As second case, suppose that $f(R')\subsetneq f(R)$. First, note that $f(R') \succ_{R'} A\setminus f(R')$ because $f(R')$ is a dominant set. Moreover, since deviating from $R$ to $R'$ is a manipulation for voter $i$, it holds that $f(R') \succ_i f(R)\setminus f(R')$. As a consequence of these two observations, it follows that $f(R') \succ_R f(R)\setminus f(R')$ because voter $i$ can only weaken the alternatives in $f(R')$ against those in $f(R)\setminus f(R')$. Finally, since $f(R')\subseteq f(R)$, it follows that $f(R')\succ_R A\setminus f(R)$, and thus, $f(R')$ is a dominant set in $\succsim_R$. Hence, robustness from $R'$ to $R$ implies that $f(R)\subseteq f(R')$, which contradicts our assumption that $f(R')\subsetneq f(R)$. Thus, $f$ is also in this case not manipulable, which shows that it is strategyproof. 
\end{proof}

Next, we focus on \Cref{thm:TC} and prove this result using \Cref{lem:domset}. As a first step, we show that pairwiseness, strategyproofness, homogeneity, and set non-imposition imply strong Condorcet-consistency. 

\begin{lemma}\label{lem:CC2}
		Every pairwise SCC that satisfies set non-imposition, homogeneity, and strategyproofness is strongly Condorcet-consistent. 
	\end{lemma}
	\begin{proof}
		Let $f$ denote an SCC as specified by the lemma. First, note the proof of Claim 1 in \Cref{lem:CC} does not require neutrality and it thus shows that $f$ is Condorcet-consistent. Hence, we focus on the converse direction and show that $f(R)=\{x\}$ can only be true if $x$ is the Condorcet winner in $R$. For this, suppose for contradiction that there is a profile $R$ and an alternative $x$ such that $f(R)=\{x\}$, but $x$ is not the Condorcet winner in $R$. This means that there is another alternative $y\in A\setminus \{x\}$ such that $g_R(y,x)\geq 0$. If $g_R(y,x)>0$, we can use the same construction as in the proof of \Cref{lem:CC} to derive that $f$ violates Condorcet-consistency since this construction works again without neutrality. Hence, suppose that $g_R(x,y)=0$. In this case, we first weaken $y$ in the preference relation of every voter $i\in N$ with $y \succ_i x$ such that it is directly over $x$ and reinforce $y$ in the preference relation of every voter $i\in N$ with $x\succ_i y$ such that it is placed directly below $x$. We infer from \IUA that $x$ is still the unique winner. Next, we iterate over the voters $i\in N$ and use \WSMON to repeatedly push down voter $i$'s best alternative until he top-ranks $x$ or $y$. It follows for the resulting profile $R^1$ that $f(R^1)=\{x\}$ because of \WSMON and that all voters report $x$ and $y$ as their best two alternatives. Thereafter, we let the voters reorder the alternatives in $A\setminus \{x,y\}$ lexicographically. \IUA implies that this step does not affect the choice set and thus, it holds for the new profile $R^2$ that $f(R^2)=f(R^1)=\{x\}$. %Moreover, note that $g_{R^2}(a,b)=0$, $g_{R^2}(x,y)=n$ for $x\in \{a,b\}$, $y\in A\setminus \{a,b\}$, and $g_{R^2}(x,y)=n$ for $x,y\in A\setminus \{a,b\}$ if $x$ is lexicographically smaller than $y$. Hence, we can now use homogeneity to reduce these margins. In more detail, homogeneity and pairwiseness imply that $f(R^3)=f(R^2)=\{a\}$ for all profiles $R^3$ with $g_{R^3}(a,b)=0$, $g_{R^3}(x,y)=n$ for $x\in \{a,b\}$, $y\in A\setminus \{a,b\}$, and $g_{R^3}(x,y)=n$ for $x,y\in A\setminus \{a,b\}$ if $x$ is lexicographically smaller than $y$. 
	
		Next, we show that $f(R^2)=\{x\}$ is in conflict with set non-imposition. For this, consider a preference profile $R^3$ with $f(R^3)=\{x,y\}$; such a profile exists since $f$ is set non-imposing. Our goal is to transform $R^3$ into $R^2$ while showing that both $x$ and $y$ must be chosen. As a first step, we repeatedly add voters with inverse preferences and use \WSMON to weaken the alternatives $z\in A\setminus\{x,y\}$ until we derive a profile $R^4$ with $\{x,y\}\succ_{R^4}A\setminus \{x,y\}$. It follows from pairwiseness and \WSMON that $f(R^4)=f(R^3)=\{x,y\}$. This implies that $x \sim_{R^4} y$ because, otherwise, there is a Condorcet winner which must be chosen uniquely. Even more, note that, as long as $x \sim_R y$ and $\{x,y\} \succ_R A\setminus \{x,y\}$, it holds that either $\{x,y\}\subseteq f(R)$ or $f(R)\subseteq \{x,y\}$. Otherwise, there is a profile $R'$ and alternatives $z_1$, $z_2$ such that $z_1\in f(R')\setminus \{x,y\}$ and $z_2\in \{x,y\}\setminus f(R')$. Hence, if a voter reinforces $z_3\in \{x,y\}\setminus \{z_2\}$ against $z_2$, $z_3$ is the Condorcet winner and it must thus be chosen uniquely. However, this is in conflict with strategyproofness because \IR (if $z_3\in f(R')$) or \IUA (if $z_3\not\in f(R')$) is violated. 
	
		We use the last observation to repeatedly reinforce the alternatives $\{x,y\}$ against the alternatives $A\setminus \{x,y\}$ in $R^4$ until all voters report $x$ and $y$ as their best two alternatives. For each swap, \WMON implies that either $z_1\in \{x,y\}$ remains chosen and $z_2\in A\setminus \{x,y\}$ remains unchosen, or $z_1$ becomes unchosen and $z_2$ chosen. However, the latter is impossible because of our previous observation and thus, we derive from \WMON and \IR that the choice set is not allowed to change. Hence, it holds for the resulting profile $R^5$ that $f(R^5)=\{x,y\}$ and that all voters report $x$ and $y$ as their best two alternatives. Thereafter, we derive the profile $R^6$ from $R^5$ by arranging the alternatives in $A\setminus \{x,y\}$ in lexicographic order, which does not affect the choice set because of \IUA. Finally, note that in $R^6$, half of the voters report $\succ_1=x,y,\lex(A\setminus \{x,y\})$ and the other half report $\succ_2=y,x,\lex(A\setminus \{x,y\})$. Using homogeneity, it follows therefore that $f(R^7)=f(R^6)$, where $R^7$ consists of two voters who report $\succ_1$ and $\succ_2$, respectively. Finally, the profile $R^2$ consists of multiple copies of $R^7$ and hence, it again follows from homogeneity that $f(R^2)=f(R^7)=\{x,y\}$. However, this contradicts the previous observation that $f(R^2)=\{x\}$ and hence, $f$ can only choose a single winner if it is the Condorcet winner. 
	\end{proof}
	
	Finally, we prove \Cref{thm:TC} based on Lemmas \ref{lem:TCweak}, \ref{lem:domset}, and \ref{lem:CC2}.
	
	\TC*
	\begin{proof}
		We have already shown in \Cref{lem:TCweak} that the top cycle satisfies set non-imposition. Moreover, by definition, \tc is majoritarian and therefore also pairwise and homogeneous. Finally, the top cycle is a robust dominant set rule and therefore strategyproof by \Cref{thm:F-SP}. For the other direction, consider an arbitrary pairwise SCC $f$ that satisfies strategyproofness, set non-imposition, and homogeneity. Since all criteria of \Cref{lem:CC2} are satisfied, it follows that $f$ is strongly Condorcet-consistent. Next, we use \Cref{lem:domset} to derive that $f$ is a robust dominant set rule. As the last step, \Cref{lem:TCweak} shows that $f$ is the top cycle since this is the only robust dominant set rule that satisfies set non-imposition.
	\end{proof}

\end{document}